\documentclass[opre,nonblindrev]{informs3mod}
\usepackage{caption}
\usepackage{xcolor}
\usepackage{stackengine}
\definecolor{green}{rgb}{1,0.5,0}
\usepackage{fancyhdr}
\usepackage{subcaption}
\usepackage{natbib}
 \bibpunct[, ]{(}{)}{,}{a}{}{,}%
 \def\bibfont{\small}%
\usepackage{mathrsfs}
\usepackage{booktabs,tablefootnote}
\usepackage{hyperref}
\usepackage{thmtools}
\usepackage[noabbrev, capitalise]{cleveref}
\usepackage{footmisc}
\usepackage{bm}

\usepackage{algorithm}
\usepackage{algpseudocode}%

\theoremstyle{plain}
\newtheorem{theorem}{Theorem}[section]

\newtheorem{lemma}[theorem]{Lemma}
\newtheorem{proposition}[theorem]{Proposition}

\newtheorem{definition}[theorem]{Definition}

\DeclareSymbolFont{symbolsC}{U}{txsyc}{m}{n}
\DeclareMathSymbol{\notniFromTxfonts}{\mathrel}{symbolsC}{61}

\usepackage{color-edits}
\addauthor[Fang-Yi]{fang}{green}
\addauthor[Amin]{mar}{blue}
\addauthor[Space Save]{ss}{yellow}

\usepackage{xcolor}

\newcommand{\BibTeX}{\rm B\kern-.05em{\sc i\kern-.025em b}\kern-.08em\TeX}

\begin{document}

\RUNAUTHOR{}

\RUNTITLE{Privacy-Aware Sequential Learning}

\TITLE{Privacy-Aware Sequential Learning}


\ARTICLEAUTHORS{%
\AUTHOR{Yuxin Liu, M. Amin  Rahimian}
\AFF{Industrial Engineering, University of Pittsburgh\\ \EMAIL{yul435@pitt.edu, rahimian@pitt.edu }}
}

\ABSTRACT{

Sequential learning often relies on individuals reporting private information, with early reports shaping later beliefs and collective decisions. When such reports are sensitive, privacy protection creates a fundamental trade-off in the design of sequential feedback systems. We study this trade-off as a platform design problem: before reporting begins, the platform commits to a local privacy-preserving reporting protocol, seeking to improve final-decision accuracy while limiting the time required to accumulate sufficient evidence. Surprisingly, privacy protection can accelerate learning. With continuous Gaussian signals, smooth randomized response under metric differential privacy preserves asymptotic learning and yields a public log-likelihood ratio growing at rate $\Theta_{\varepsilon}(\log n)$, faster than the nonprivate $\Theta(\sqrt{\log n})$ benchmark, substantially reducing high-confidence stopping times. With heterogeneous privacy parameters, learning can be faster still; when privacy budgets are uniformly distributed on $[0,1]$, public belief grows at rate $\Theta(n^{1/4})$. With binary signals, randomized response generates a nonmonotone relationship between privacy and decision accuracy because privacy affects both report informativeness and cascade thresholds. Stopping time can also be nonmonotone because stronger privacy reduces report informativeness while encouraging participation. Overall, privacy is not merely a constraint on information release, but a platform design lever shaping participation, stopping, and collective learning.

}

\KEYWORDS{sequential learning, privacy-preserving learning, information cascade, learning efficiency}

\maketitle


\pagestyle{fancy}
\fancyhead{}


\maketitle 


\section{Introduction}
\vspace{-5pt}
When making decisions, people often combine their private information with public information, such as observable behaviors or reports of others \citep{bikhchandani2024information}. For example, on online platforms, users may decide whether to purchase or use a product by combining their own private assessments with ratings or recommendations provided by earlier users. This process can be modeled within the framework of sequential learning to understand how individual recommendations and collective decisions evolve over time based on the reports of predecessors \citep{Banerjee1992herd, Bikhchandani1992cascades, Acemoglu2011network}.

Traditionally, sequential learning models examine how individuals infer from others' actions, typically assuming that observed actions or reports directly reveal the decision-maker's private information \citep{Ali2018cost}. While such models have been successful in explaining herding behaviors and information cascades, they overlook privacy concerns that arise when feedback itself may expose sensitive workplace information. For example, when members of an organization want to collectively decide whether to adopt an AI assistant, an employee's feedback may reveal aspects of her task environment, productivity constraints, skill mismatch, or vulnerability to automation. Anticipating such privacy risks, some employees may choose not to participate in the feedback process, leading to an incomplete public record. A platform-designed privacy-preserving reporting protocol offers a middle ground: it allows employees to contribute useful information while limiting what can be inferred from any individual report. Yet privacy protection is not costless. By randomizing reported feedback, a privacy-preserving protocol can weaken the informational content of each report, slow down belief updating, and make it harder for the organization to distinguish genuine evidence from noise.

However, randomization can sometimes be beneficial. Occasional deviations from predecessors' actions can prevent fragile, self-reinforcing cascades \citep{peres2020fragile}; similarly, privacy-preserving noise can dilute early trends driven by misinformation, social pressure, or institutional incentives, reducing misleading cascades and promoting more robust learning. This tension raises a central question: How should a platform design privacy-preserving reporting protocols to balance individual privacy concerns, participation, and sequential learning efficiency?

To address these questions, we adopt the metric differential privacy (mDP) framework \citep{chatzikokolakis2013broadening} to model privacy-preserving reporting in sequential learning. In our setting, the platform commits to a local privacy-preserving reporting protocol before feedback collection begins. This protocol randomizes agents' reported recommendations to limit inferences from individual private signals, thereby altering the dynamics of learning. Contrary to the common intuition that additional noise should reduce the informativeness of public observations \citep{Le2017noise, papachristou2025differentially}, we show that privacy-preserving randomization can enhance learning by asymmetrically amplifying the log-likelihood ratio across states, thus increasing the information conveyed by each report. This asymmetry arises because reports are distorted differently under the correct and wrong states: under the correct state, agents' signals tend to reinforce the prevailing belief and require less randomization; under the wrong state, contradictory signals are more sensitive and induce stronger privacy-preserving randomization. Unlike prior work that ties learning speed solely to the structure of private signals \citep{Hann-Caruthers2018speed, rosenberg2019efficiency}, we demonstrate that privacy-preserving noise reshapes aggregation dynamics and can accelerate convergence.

These findings carry important managerial implications for platforms and policymakers. By identifying the mechanisms through which privacy-motivated randomization accelerates learning, designers can tailor data-collection protocols to the sensitivity of the information and the stage of the learning process—for example, adjusting the frequency and granularity of data requests when privacy risks are high, or selectively soliciting input from less-represented groups to counteract early cascades. In data aggregation, incorporating models that account for the presence and strategic use of noise can help platforms weight incoming signals appropriately, filter out spurious early trends, and reinforce robust patterns. Such adaptive strategies can be applied in domains such as organizational AI adoption, employee feedback systems, and technology rollout decisions, enabling faster and more reliable collective decision-making while safeguarding individual privacy.

\vspace{-5pt}
\subsection{Main related work}\label{sec:lit-rev}
\vspace{-5pt}

Our work contributes to the literature on sequential learning, particularly regarding whether asymptotic learning occurs and how fast public beliefs aggregate information. The failure of asymptotic learning due to information cascades was first introduced by \citet{Banerjee1992herd} and \citet{Bikhchandani1992cascades}, who showed that binary private signals can trigger cascades that stop further learning. \citet{Smith2000pathological} later established that unbounded private beliefs are necessary for asymptotic learning, a result extended by \citet{Acemoglu2011network} to network settings with partial observations. Subsequent studies show that learning may persist even when agents observe random predecessors without time indices \citep{smith2013rational}. Whereas these studies examine sequential learning without privacy-preserving reporting constraints, we study how a platform-designed local reporting protocol reshapes information aggregation. In the binary-signal setting, the platform's privacy protocol changes both the informativeness of each released report and the information-cascade threshold, leading to a non-monotone relationship between privacy protection and the probability of a correct final decision. In the continuous-signal setting, the platform uses a smooth randomized response protocol under metric differential privacy, which preserves asymptotic learning and can accelerate the growth of the public log-likelihood ratio relative to the non-private benchmark.

Beyond the question of whether asymptotic learning occurs, another strand of research studies the speed of learning. \citet[Chapter~4]{chamley2004rational} investigates convergence rates in Gaussian settings through computational simulations. \citet{rosenberg2019efficiency} analyze learning efficiency by quantifying the time until the first correct action appears and the total number of incorrect actions, deriving conditions under which these quantities have finite expectations. \citet{Hann-Caruthers2018speed} precisely characterize the asymptotic speed of learning, showing that for Gaussian signals, the public log-likelihood ratio evolves at rate $\Theta(\sqrt{\log n})$, so learning can be asymptotically correct yet extremely slow. \citet{arieli2025hazards} examine the role of overconfidence, demonstrating that mildly overconfident agents, who underestimate the quality of their peers' information, can accelerate learning in standard sequential models. We extend this literature by focusing on a platform's stopping problem under privacy-preserving reporting. Instead of using the time to the first correct action or the total number of incorrect actions as our primary metric, we evaluate the time required for the public belief to reach a prescribed threshold. Under a homogeneous privacy budget $(\varepsilon)$, the platform-designed smooth randomized response protocol yields a public log-likelihood ratio $(\ell_n)$ that grows at rate $\Theta_{\varepsilon}(\log n)$, improving on the $\Theta(\sqrt{\log n})$ non-private benchmark. Under personalized privacy mechanisms with heterogeneous privacy parameters, the public belief can grow even faster, at rate $\Theta(n^{1/4})$. Together with our stopping-time analysis, these convergence-rate results establish shorter high-confidence stopping times for the platform's final decision. Under sufficiently strong homogeneous privacy, the expected time to the first correct submitted report is finite, and it is also finite under personalized privacy with privacy parameters drawn from \(U[0,1]\). However, the total number of incorrect submitted reports is almost surely infinite under both mechanisms when the reporting process continues indefinitely.

Our study also contributes to the broader literature on how privacy constraints affect learning and decision-making. In the context of sequential learning, \citet{tsitsiklis2021private} and \citet{xu2018query} analyze the query complexity of private learning and show that stronger privacy constraints require more information acquisition for accurate inference. Although their models share structural similarities with ours, they focus on independent queries, whereas we study path-dependent sequential feedback in which each released report shapes the information available to later agents and to the platform. In distributed learning and estimation, privacy-preserving noise is often modeled as an externally imposed constraint or an exogenous perturbation \citep{papachristou2025differentially, rizk2023enforcing, tao2023distributed}. Our model is closer to this design perspective: the platform commits to a local privacy-preserving reporting protocol before feedback collection begins, and this protocol determines how private signals are transformed into public reports. We further allow privacy parameters to vary across agents, consistent with empirical evidence of heterogeneous privacy preferences \citep{acquisti2005privacy} and with personalized or heterogeneous differential privacy frameworks \citep{alaggan2015heterogeneous, acharya2024personalized}. This perspective highlights how platform-designed privacy protection affects report informativeness, participation, belief updating, and stopping decisions in sequential learning environments. To our knowledge, the only work examining social learning in privacy-preserving data collection is \citet{Wang2021privacy_sociallearning}, which focuses on simultaneous actions and does not address sequential learning, information cascades, or belief-threshold stopping.

Our work is the first to examine platform-designed privacy-preserving reporting in sequential learning with an explicit stopping-time objective. Counter to the common intuition that privacy noise necessarily slows learning, we show that appropriately structured local randomization can improve learning efficiency by reshaping the dynamics of public belief formation. In binary environments, randomized response changes both the probability of reaching a correct cascade and the expected stopping time. In continuous Gaussian environments, smooth randomized response induces asymmetric likelihood-ratio increments across states, allowing the public belief to reach high-confidence thresholds faster than in the standard non-private sequential-learning benchmark. These results show that privacy protection is not merely a constraint on information release, but also a design lever that affects participation, report informativeness, and the speed of platform decision-making. Further related literature is reviewed in Appendix~\ref{app:related_work}.

\vspace{-5pt}
\subsection{Main contribution}
\vspace{-5pt}
In this paper, we introduce a structured approach to analyzing platform-designed privacy-preserving sequential learning, offering new theoretical insights into how externally specified reporting protocols affect learning efficiency. When private signals are binary, we reformulate the sequential feedback process as a Markov chain with absorbing states and characterize the probability that the platform reaches a correct final decision. Using techniques from the gambler’s ruin problem, we quantify how randomized response changes belief propagation, cascade thresholds, and the likelihood of reaching a correct cascade. This analysis shows that a platform-imposed privacy-preserving reporting protocol affects learning through two channels: it changes the informativeness of each released report and alters the stopping behavior of the sequential feedback process.

In the continuous signal case, we study a platform-designed smooth randomized response protocol that selectively randomizes reports according to the distance between an agent's private signal and the current decision threshold. This protocol introduces stronger randomization for signals near the threshold, where small changes in the signal may alter the recommendation, while preserving more information from signals that provide clearer evidence for one state. We then analyze the evolution of the public log-likelihood ratio over time using Gaussian tail approximations, conditional-drift analysis, and martingale arguments, which enables a rigorous characterization of how information accumulates under the platform-imposed privacy-preserving reporting protocol. The resulting belief dynamics lead to an accelerated asymptotic learning speed of \( \Theta_{\varepsilon}(\log(n)) \), in contrast to the classical \( \Theta(\sqrt{\log(n)}) \) rate in the non-private regime \citep{rosenberg2019efficiency, Hann-Caruthers2018speed}. 

This acceleration arises from the inherent asymmetry induced by smooth randomized response across the two states. Specifically, under the true state, the average probability of flipping a report is lower than under the false state. As a result, released reports become more informative for public inference, enabling the platform's belief to update more efficiently. To evaluate learning efficiency from the platform's perspective, we further analyze the stopping time required for the public log-likelihood ratio to reach a prescribed threshold. This stopping-time analysis connects the asymptotic growth rate of public belief to the time needed for the platform to make a high-confidence adoption decision, highlighting the trade-off between privacy protection, report informativeness, and decision speed.

We also extend our analysis to a personalized privacy setting in which the platform adapts the reporting protocol to agents' heterogeneous privacy tolerances. Specifically, each agent \(n\) is associated with an individual privacy parameter \(\varepsilon_n\), and the platform applies a corresponding local privacy-preserving reporting mechanism to that agent's feedback. Focusing on the case in which the privacy parameters are uniformly distributed over \([0,1]\), we show that the public log-likelihood ratio grows at rate \(\Theta(n^{1/4})\). This is substantially faster than the non-private sequential learning benchmark, where the convergence rate is only \(\Theta(\sqrt{\log n})\) \citep{rosenberg2019efficiency, Hann-Caruthers2018speed}. The key mechanism is the presence of agents with very small privacy parameters in the heterogeneous population. Under smooth randomized response, these agents have a disproportionate effect on belief updating because the state-dependent asymmetry of the mechanism makes their likelihood-ratio contributions decay much more slowly as public belief becomes strong. As a result, their reports increasingly drive belief updating in the high-confidence regime, leading to substantially faster information aggregation.

This result offers two key insights. First, platform-designed privacy protection can serve as more than a constraint on information release: when properly structured, it can change the dynamics of sequential learning and accelerate public belief formation. Second, although personalized privacy improves learning relative to the standard non-private sequential learning benchmark, it still falls short of the fully transparent benchmark in which the platform directly observes i.i.d. private signals. In that idealized case, the public log-likelihood ratio grows at the linear rate \(\Theta(n)\), as implied by the law of large numbers. Thus, even under platform-designed personalized privacy, the sequential nature of observation imposes a fundamental limit on the speed of information aggregation. A comprehensive comparison of these results across non-private, homogeneous privacy, and personalized privacy settings is summarized in Table~\ref{tab:theory_summary}.

\begin{table}[htbp]
\centering
\small
\renewcommand{\arraystretch}{1.3}
\resizebox{\textwidth}{!}{%
\begin{tabular}{|l|c|c|c|}
\hline
\textbf{Learning Metrics}
& \textbf{Non-private}
& \textbf{Platform-designed Homogeneous Privacy}
& \textbf{Platform-designed Personalized Privacy} \\
\hline

\multicolumn{4}{|c|}{\textbf{Binary Signals}} \\
\hline

Probability of correct cascade &
$\displaystyle \frac{\rho^{k}-1}{\rho^{2k}-1}$ &
$\displaystyle
\frac{\rho(\varepsilon,p)^{k(\varepsilon,p)}-1}
{\rho(\varepsilon,p)^{2k(\varepsilon,p)}-1}$ &
$\displaystyle
\frac{\rho(\bar{u},p)^{k(\bar{u},p)}-1}
{\rho(\bar{u},p)^{2k(\bar{u},p)}-1}$ \\
\hline

Information cascade threshold &
$\displaystyle k=2$ &
$\displaystyle
k(\varepsilon,p)=
\left\lfloor
\log_{\rho(\varepsilon,p)}\frac{1-p}{p}
\right\rfloor+1$ &
$\displaystyle
k(\bar{u},p)=
\left\lfloor
\log_{\rho(\bar{u},p)}\frac{1-p}{p}
\right\rfloor+1$ \\
\hline

Expected belief-threshold stopping time &
$\displaystyle
\mathbb{E}[\tau]
=
\frac{k}{\mu}
\cdot
\frac{1-\rho^{k}}{1+\rho^{k}}$ &
$\displaystyle
\mathbb{E}[\tau_{\varepsilon,p}]
=
\frac{k(\varepsilon,p)}{\mu(\varepsilon,p)}
\cdot
\frac{1-\rho(\varepsilon,p)^{k(\varepsilon,p)}}
{1+\rho(\varepsilon,p)^{k(\varepsilon,p)}}$ &
$\displaystyle
\mathbb{E}[\tau_{\bar{u},p}]
=
\frac{k(\bar{u},p)}{\mu(\bar{u},p)}
\cdot
\frac{1-\rho(\bar{u},p)^{k(\bar{u},p)}}
{1+\rho(\bar{u},p)^{k(\bar{u},p)}}$ \\
\hline

\multicolumn{4}{|c|}{\textbf{Continuous Signals}} \\
\hline

Convergence rate of public log-likelihood ratio &
$\displaystyle \Theta(\sqrt{\log n})$ &
$\displaystyle \Theta_{\varepsilon}(\log n)$ &
$\displaystyle \Theta(n^{1/4})$ \\
\hline

Expected belief-threshold stopping time &
$\displaystyle
\mathbb{E}[\tau_B]
=
\exp\!\left(\frac{\sigma^2 B^2}{8}+o(B^2)\right)$&
$\displaystyle
\mathbb{E}[\tau_B]
\sim
\frac{1}{K(\varepsilon,\sigma)}
\exp\left(\frac{\varepsilon\sigma^2}{2}B\right)$ &
$\displaystyle
\mathbb{E}[\tau_B]
\sim
\frac{\sigma^6}{32}B^4$ \\
\hline

Expected time to first correct submitted report &
$\displaystyle
\mathbb{E}[\tau_{\mathrm{fc}}]=\infty$ &
$\displaystyle
\mathbb{E}[\tau_{\mathrm{fc}}]<\infty
\quad\text{if }\varepsilon<\frac{2}{\sigma^2}$ &
$\displaystyle
\mathbb{E}[\tau_{\mathrm{fc}}]<\infty$ \\
\hline

\end{tabular}%
}
\captionsetup{justification=raggedright,singlelinecheck=false}
\caption[Summary of theoretical results under platform-designed privacy protocols]{
Summary of theoretical results under non-private learning and
platform-designed privacy-preserving reporting protocols,
measured at the submitted-report level.
Learning rates are stated conditional on \(\theta=+1\);
the stopping-time and accuracy results are identical under
the two states by symmetry.
Here, \(\varepsilon\) denotes the privacy budget,
\(p\) is the probability that a binary private signal matches
the true state, \(\ell_n\) is the public log-likelihood ratio
after \(n\) submitted reports, and \(\sigma\) is the standard
deviation of the continuous Gaussian private signal.
Moreover,
\(\rho=(1-p)/p\),
\(\mu=2p-1\), and
\(u(\varepsilon)=1/(1+e^\varepsilon)\)
is the randomized-response flipping probability.
The quantity
\(\bar{u}=\mathbb{E}_{\varepsilon}[u(\varepsilon)]\)
denotes the average flipping probability under personalized privacy.
Further,
\(\rho(\varepsilon,p)
=\frac{(1-u(\varepsilon))(1-p)+u(\varepsilon)p}
{u(\varepsilon)(1-p)+p(1-u(\varepsilon))}\),
\(\rho(\bar{u},p)
=\frac{(1-\bar{u})(1-p)+\bar{u}p}
{\bar{u}(1-p)+p(1-\bar{u})}\),
\(\mu(\varepsilon,p)=(2p-1)(1-2u(\varepsilon))\), and
\(\mu(\bar{u},p)=(2p-1)(1-2\bar{u})\).
In the continuous homogeneous privacy case,
\(K(\varepsilon,\sigma)
=C(\varepsilon)
-\frac12e^{-\varepsilon+\varepsilon^2\sigma^2/2}
(1-e^{-\varepsilon^2\sigma^2})\), where
\(C(\varepsilon)
=\frac{\varepsilon\sigma^2}{4}
\left(
e^{\varepsilon+\varepsilon^2\sigma^2/2}
-
e^{-\varepsilon+\varepsilon^2\sigma^2/2}
\right)\).
The continuous personalized-privacy results assume
independent privacy parameters \(\varepsilon_n\sim U[0,1]\),
independent of private signals and the state, with public
beliefs integrating over the unobserved privacy parameters.
The continuous belief-threshold stopping time is
\(\tau_B=\inf\{n\geq1:|\ell_n|\geq B\}\), with \(\ell_0=0\).
Stopping-time asymptotics hold as \(B\to\infty\), with
the model parameters fixed.
The time to the first correct submitted report is
\(\tau_{\mathrm{fc}}=\inf\{n\geq1:x_n=\theta\}\).
If potential agents participate independently with probability
\(\eta(\varepsilon)>0\), independently of their private
signals and reporting randomization, the corresponding
calendar stopping time satisfies
\(\mathbb{E}[T]=\mathbb{E}[\tau]/\eta(\varepsilon)\).
Under personalized privacy, all agents participate,
so calendar and submitted-report stopping times coincide.
}
\label{tab:theory_summary}
\end{table}

Our analysis demonstrates that, contrary to the conventional view that privacy protection necessarily slows down learning, platform-designed privacy-preserving reporting protocols can improve learning efficiency by reshaping the dynamics of information flow over time. Unlike prior literature, which primarily focuses on asymptotic learning without privacy-preserving reporting constraints, we provide a formal convergence-rate and stopping-time analysis for sequential learning under externally specified privacy protocols. These findings reveal how the design of local reporting mechanisms affects report informativeness, participation, belief updating, and final decision making in sequential feedback systems. From a managerial perspective, this highlights that protecting individual privacy need not come at the cost of slower organizational learning. Well-designed privacy-preserving feedback protocols can support broader participation, reduce premature herding, and help platforms make more reliable high-confidence decisions over time.

The remainder of this paper is organized as follows. First, we formally define our problem setup in \Cref{sec:model}, including the sequential decision-making framework, the platform-designed reporting protocol, and the associated privacy and participation constraints. \Cref{sec:binary-model} focuses on the case of binary signals, where we analyze how randomized response affects cascade thresholds, the probability of a correct final decision, and the expected stopping time using a Markov chain representation. In \Cref{sec:continuous-model}, we extend our analysis to continuous Gaussian signals, introduce the smooth randomized response protocol, and characterize how it affects the convergence rate of the public log-likelihood ratio and the expected belief-threshold stopping time. Finally, \Cref{sec:conclusion} summarizes our findings and discusses potential future research directions, including broader implications for privacy-preserving platform design in sequential learning environments.

\vspace{-5pt}
\section{A Model of Sequential Learning with Privacy}
\label{sec:model}
\vspace{-5pt}

We consider a classical sequential learning model. There exists an unknown binary state of nature, denoted by $\theta \in \{-1,+1\}$, with prior probability $1/2$ on each state. Each agent, indexed by \(n \in \{1,2,\ldots\}\), receives a private signal \(s_n \in \mathscr{X}\), drawn independently and identically distributed conditional on the true state, that is, \(s_n \sim F_{\theta}\), where \(F_{\theta}\) depends on \(\theta\). In the absence of privacy concerns, agent \(n\) reports an action \(x_n \in \{-1,+1\}\) based on her private signal \(s_n\) and the history of previous reported actions, denoted by \(h_{n-1}:=\{x_1,\ldots,x_{n-1}\}\). Thus, \(x_n\) represents the publicly observed report of agent \(n\). We introduce the distinction between an agent's intended action and the privacy-protected report later when defining the privacy-preserving reporting mechanism.

In our main motivating example, the agents are employees who sequentially experience a newly introduced AI decision-support tool. The unknown state $\theta$ represents whether the AI tool truly improves workplace performance, with $\theta=+1$ indicating that the tool is beneficial and $\theta=-1$ indicating that it is not. Each agent receives a private signal $s_n$ from her own experience with the tool, such as productivity improvement, time saved, correction effort, workflow disruption, or ease of integration. Based on this signal and the public history of prior reports, the agent forms an intended recommendation $a_n$ about whether the tool should be adopted more broadly.

The private signal may be sensitive because it can reveal information about the agent's reliance on AI, adaptability to new technology, or exposure to automation risk. Hence, raw feedback may discourage participation from privacy-sensitive agents. In the absence of privacy protection, some agents may prefer not to submit feedback at all, leading to an incomplete public record. Privacy protection, therefore, provides a middle ground between raw disclosure and non-disclosure: it allows agents to contribute useful information while limiting inference about their private experience.

To model this setting, we interpret the public action $x_n$ as a privacy-protected report rather than a direct real-world action. Before the sequential process begins, the platform commits to a local privacy-preserving reporting protocol. The protocol is local in the sense that randomization occurs within the agent's reporting interface, before the feedback becomes publicly observable; only the resulting privatized report $x_n$ enters the public history observed by later agents.

To systematically characterize the privacy guarantee provided by the reporting protocol, we adopt the \textit{metric differential privacy} framework \citep{chatzikokolakis2013broadening}. From the platform's perspective, the privacy parameter is a design choice of the feedback system. Before the sequential rollout begins, the platform chooses a privacy level $\varepsilon$ and commits to a local reporting mechanism $\mathcal{M}_{\varepsilon}$. This choice determines how informative each released report is and, in turn, affects the quality of long-run organizational learning. A larger $\varepsilon$ permits more accurate and informative reports, while a smaller $\varepsilon$ provides stronger privacy protection but introduces more noise into the public feedback record. Thus, the platform's design problem is to choose a reporting protocol that balances individual privacy protection with the long-run accuracy of sequential decision making.

Formally, this design choice is represented by requiring the chosen reporting mechanism to satisfy metric differential privacy with privacy parameter $\varepsilon$. Let \( \mathscr{X} \) be a set equipped with a metric \( d_{\mathscr{X}} : \mathscr{X} \times \mathscr{X} \to [0, \infty) \), which measures the distance between two possible private signals. A local reporting mechanism \( \mathcal{M}_{\varepsilon} \) satisfies \((\varepsilon, d_{\mathscr{X}})\)-metric differential privacy if for all \( s_n, s_n' \in \mathscr{X} \), all histories \(h_{n-1}\), and all possible outputs \( x_n \in \{-1, +1\} \), we have
\begin{equation}\label{eq:mdp}
\mathbb{P}(\mathcal{M}_{\varepsilon}(s_n; h_{n-1}) = x_n) 
\leq 
\exp\left(\varepsilon \cdot d_{\mathscr{X}}(s_n, s_n') \right) 
\mathbb{P}(\mathcal{M}_{\varepsilon}(s_n'; h_{n-1}) = x_n).
\end{equation}
In this sense, $\varepsilon$ is not merely a descriptive privacy budget but the parameter that defines the feasible class of reporting protocols available to the platform. Choosing a smaller $\varepsilon$ restricts the platform to mechanisms that make reports less sensitive to the agent's private signal, while choosing a larger $\varepsilon$ allows the platform to use mechanisms that preserve more information from agents' private feedback.

Agents, however, may differ in their willingness to participate in such a feedback system. We allow each agent \(n\) to have a heterogeneous privacy tolerance, denoted by \(\varepsilon_n\). An agent participates only if the privacy loss induced by the platform's reporting protocol is acceptable to her. In particular, we assume that agent \(n\) joins the information collection process if and only if
\(\varepsilon_n \geq \varepsilon.\)

Under this participation rule, a smaller platform-chosen $\varepsilon$ provides stronger privacy protection and therefore attracts a larger set of privacy-sensitive agents, but each submitted report becomes less informative. Conversely, a larger $\varepsilon$ makes each report more informative but may discourage participation from agents with stronger privacy concerns. This captures a central trade-off faced by the platform designer: the reporting protocol affects both the informativeness of each submitted report and the total time required to collect enough feedback for a stopping decision.

This participation interpretation clarifies the role of privacy in the sequential learning process. By choosing the privacy level $\varepsilon$, the platform affects both the informativeness of each released report and the set of agents willing to participate. The platform's objective is therefore to design a privacy protocol that supports accurate long-run learning from a representative stream of participating agents.

Rather than collecting feedback only for internal evaluation, the platform maintains a shared feedback record so that agents who encounter the technology later can learn from the experiences of earlier agents. Such information sharing is particularly relevant during a gradual rollout, when early agents' feedback can help later agents form expectations about the value of the technology without relying solely on their own limited experience. Early participating agents therefore submit privacy-protected feedback that becomes part of the shared record observed by subsequent agents. This transparency facilitates organizational learning, but it also creates the possibility that early reports disproportionately shape later beliefs and generate information cascades. Privacy-preserving reporting can therefore play a dual role: it reduces the information revealed by any individual report while also altering the extent to which early feedback drives subsequent learning.

\vspace{-5pt}
\section{Binary Signals and Information Cascades}
\label{sec:binary-model}
\vspace{-5pt}

In the binary model, we assume that \( s_n \in \mathscr{X} \), where \( \mathscr{X} = \{-1, +1\} \), and that \( s_n = \theta \) with probability \( p > 1/2 \); i.e., a binary symmetric channel where each private signal matches the true state \( \theta \in \{-1, +1\} \) with probability \( p \), and equals \( -\theta \) with probability \( 1 - p \). In this binary signal setting, we endow the signal space \( \mathscr{X} \) with the Hamming distance \( d_h \), which takes the value 0 if two signals are equal and 1 otherwise. Under this choice of metric, metric differential privacy coincides with local differential privacy \citep{chatzikokolakis2013broadening}. The formal definition is as follows:
\vspace{-5pt}

\begin{definition}[$\varepsilon$-Local Differential Privacy for the Binary Model]\label{def:dp-binary}
A local reporting mechanism \( \mathcal{M} \) satisfies \( \varepsilon \)-local differential privacy if for all \( s_n, s_n' \in \{-1, +1\} \), all histories \(h_{n-1} \in \{-1,+1\}^{n-1}\), and all possible outputs \( x_n \in \{-1, +1\} \), we have:
\begin{equation}\label{eq:DPB}
\mathbb{P}(\mathcal{M}(s_n; h_{n-1}) = x_n) 
\leq 
\exp(\varepsilon) \cdot 
\mathbb{P}(\mathcal{M}(s_n'; h_{n-1}) = x_n).
\end{equation}
Here, \( \varepsilon > 0 \) is the privacy level specified by the reporting protocol.
\end{definition}
\vspace{-5pt}

Let \( \mathcal{F}_{\varepsilon} \) denote the set of local reporting mechanisms that satisfy \( \varepsilon \)-local differential privacy. In our motivating example, from the platform's perspective, the feedback system is designed before the sequential rollout begins. The platform chooses a privacy level \( \varepsilon \) and commits to a local reporting mechanism \( \mathcal{M}\in\mathcal{F}_{\varepsilon} \). Conditional on participation, agent \(n\) observes the public history \(h_{n-1}\), receives a private signal \(s_n\), and forms an initial recommendation \(a_n\in\{-1,+1\}\). The reporting interface then releases a privacy-protected public report \(x_n\), generated according to the chosen mechanism \(\mathcal{M}\). Later agents observe only the public history of released reports.

Agents may differ in their willingness to participate in the feedback system. We denote agent \(n\)'s privacy tolerance by \(\varepsilon_n\), interpreted as the maximum privacy loss she is willing to accept. Given a platform-chosen privacy level \(\varepsilon\), agent \(n\) participates only if \(\varepsilon_n \geq \varepsilon\). Therefore, a smaller \(\varepsilon\) provides stronger privacy protection and can encourage broader participation, while a larger \(\varepsilon\) permits more informative reports but may discourage privacy-sensitive agents from contributing feedback.

The platform uses the sequential feedback process to decide whether the AI tool should be adopted more broadly. Before an information cascade occurs, participating agents' recommendations still depend on their private signals, so their reports contribute new information to the public feedback record. Once a cascade occurs, however, subsequent recommendations are driven mainly by the accumulated public history rather than by new private information. Thus, from the platform's perspective, feedback collection is most valuable before the cascade threshold is reached. After that point, the platform can stop collecting feedback and make a final decision based on the accumulated public record. To formalize this stopping decision, we distinguish between the number of submitted reports required to reach the stopping threshold and the corresponding number of potential agent arrivals.

\begin{definition}[Belief-Threshold and Calendar Stopping Times in the Binary Model]
\label{def:binary-stopping-times}
Let $\tau_{\varepsilon,M}$ denote the \emph{belief-threshold stopping time}, defined as the number of submitted reports required for the public history to reach the information-cascade threshold under privacy level $\varepsilon$ and reporting mechanism $M$. The corresponding \emph{calendar stopping time}, denoted by $T_{\varepsilon,M}$, counts the number of potential agents who arrive before the same threshold is reached.
\end{definition}

At this stopping time, the platform makes a final decision \(D_{\varepsilon,\mathcal{M}}\in\{-1,+1\}\), where \(D_{\varepsilon,\mathcal{M}}=+1\) corresponds to adopting the AI tool and \(D_{\varepsilon,\mathcal{M}}=-1\) corresponds to rejecting or discontinuing it. The platform's protocol design problem can be written as
\begin{equation}
\label{eq:company-objective-general}
\max_{\varepsilon>0,\ \mathcal{M}\in\mathcal{F}_{\varepsilon}}
\quad
\mathbb{P}_{\varepsilon,\mathcal{M}}\!\left(D_{\varepsilon,\mathcal{M}}=\theta\right)
-
c\,\mathbb{E}_{\varepsilon,\mathcal{M}}\!\left[T_{\varepsilon,\mathcal{M}}\right],
\end{equation}
where \(c>0\) captures the cost of delayed decision making, continued experimentation, or prolonged feedback collection. 
Here, \(T_{\varepsilon,\mathcal{M}}\) denotes the calendar stopping time, namely the number of potential agents who arrive before the platform collects enough submitted reports to reach the stopping threshold and make a final decision.

This objective captures the platform's privacy-learning trade-off. A mechanism that provides stronger privacy protection can increase participation and may allow the platform to collect feedback from a broader set of agents. However, stronger privacy also makes each released report less informative. Conversely, a less private mechanism may produce more informative reports from participating agents but can reduce participation and slow the accumulation of feedback in calendar time. The platform therefore evaluates a privacy protocol by both the accuracy of the final decision and the time required to reach that decision.

At this stage, we do not impose a particular functional form on \( \mathcal{M} \) beyond the local differential privacy constraint. 
In the main results below, we focus on randomized response as a canonical binary reporting protocol within the feasible class \( \mathcal{F}_{\varepsilon} \). 
This allows us to characterize how the platform's privacy choice affects the informativeness of the reports, cascade formation, final decision accuracy, and stopping time.
\vspace{-5pt}

\subsection{Main Results for Binary Signals}
\label{sec:binary}
\vspace{-5pt}

We now analyze the binary signal model under a platform-designed randomized response reporting protocol. 
The key simplification in the binary setting is that, before an information cascade occurs, an agent's initial recommendation is fully determined by her private signal. Hence, the effect of the reporting protocol can be captured by how much information the released public reports preserve before the stopping threshold is reached.

Before presenting the main result, we first define an information cascade in our sequential feedback setting.
\vspace{-5pt}
\begin{definition}[Information Cascade]\label{def:inf-cascade}
An information cascade occurs when an agent's action becomes independent of her private signal and is determined solely by the public history. 
If the induced action is aligned with the true state, it is referred to as a \textit{correct cascade}. 
In contrast, if the induced action is misaligned with the true state, it is referred to as a \textit{wrong cascade}.
\end{definition}
\vspace{-5pt}

In the present model, an information cascade also serves as a stopping rule for the platform. 
Before a cascade occurs, participating agents' recommendations depend on their private signals and therefore contribute new information to the public feedback record. 
Once the cascade threshold is reached, subsequent recommendations are driven by the public history rather than by new private information. 
Thus, from the platform's perspective, the marginal value of continued feedback collection becomes negligible, and the platform stops collecting feedback and makes its final decision.

The following result states that the randomized response protocol used in the binary model satisfies the local differential privacy constraint. 
For a fixed privacy level \(\varepsilon\), the platform uses a symmetric binary flip mechanism that flips the agent's initial recommendation with probability \(u(\varepsilon)\).

\vspace{-5pt}
\begin{theorem}[Binary Randomized Response Satisfies Differential Privacy]
\label{thm:piece-wise-random-response}
Consider the binary randomized response reporting protocol with privacy level \(\varepsilon\). 
Before an information cascade occurs, the reporting interface flips the agent's initial 
recommendation \(a_n\in\{-1,+1\}\) with probability
\[
u(\varepsilon)=\frac{1}{1+e^{\varepsilon}},
\]
and keeps it unchanged with probability \(1-u(\varepsilon)\). Then this protocol satisfies 
\(\varepsilon\) differential privacy.
\end{theorem}
\vspace{-5pt}

Randomized response provides a natural benchmark in the binary signal setting. Before a cascade occurs, an agent's initial action is in one-to-one correspondence with her private signal: observing the action is equivalent to observing the direction of the signal. Since both the initial action $a_n$ and the released report $x_n$ are binary, a symmetric binary reporting protocol can be represented by the probability that the report differs from the initial action. For a fixed privacy level $\varepsilon$, the randomized response protocol sets this flip probability to
\[
u(\varepsilon)=\frac{1}{1+e^{\varepsilon}},
\]
which is the smallest flip probability compatible with $\varepsilon$-local differential privacy within the symmetric flip class. Thus, a smaller $\varepsilon$ imposes stronger privacy protection and requires more randomization, while a larger $\varepsilon$ relaxes the privacy constraint and allows the released report to preserve more information from the agent's initial recommendation.

In the analysis below, we use this protocol to characterize how privacy affects cascade formation, 
the probability of a correct final decision, and the stopping time of the sequential feedback process. 
This benchmark captures the central privacy-learning trade-off while keeping the belief dynamics 
analytically tractable.

\paragraph{Accuracy of the final decision.} To evaluate the platform's protocol design problem, we first focus on the accuracy of the final decision. The platform stops collecting feedback once the public history reaches an information cascade threshold and then makes a final decision based on the direction of the cascade. Therefore, the probability that the platform makes the correct final decision is equivalent to the probability that the public history first enters a correct information cascade.

To characterize this probability, we first define the information cascade threshold.
\vspace{-5pt}
\begin{definition}[Information Cascade Threshold $(k)$]\label{def:inf-cascade-threshold}
The information cascade threshold \(k\) is defined as the minimum imbalance between the number of \(+1\) and \(-1\) reports in the public history \(h_{n-1}\) that makes the next agent's action independent of her private signal. If the public history contains at least \(k\) more \(+1\) reports than \(-1\) reports, then the next agent chooses \(+1\) regardless of her private signal. Conversely, if the public history contains at least \(k\) more \(-1\) reports than \(+1\) reports, then the next agent chooses \(-1\) regardless of her private signal.
\end{definition}
\vspace{-5pt}

Consider first the non-private benchmark, corresponding to \( \varepsilon=\infty \). In this case, the released report coincides with the agent's initial recommendation. In the standard binary sequential learning model, cascading begins when one action has been observed sufficiently more often than the other. In particular, as shown by \citet{Bikhchandani1992cascades}, the cascade threshold is \(k=2\): if an agent observes two more reports of \(+1\) than \(-1\), she chooses \(+1\) regardless of her private signal; similarly, if she observes two more reports of \(-1\) than \(+1\), she chooses \(-1\) regardless of her private signal.

Under the privacy-protected reporting protocol, the informativeness of each public report is reduced by randomized response. Let \(u(\varepsilon)=1/(1+e^{\varepsilon})\). Before the cascade threshold is reached, an agent's initial recommendation reflects her private signal, but the released report is flipped with probability \(u(\varepsilon)\). Therefore, the effective likelihood ratio conveyed by each public report depends on both the signal accuracy \(p\) and the privacy level \(\varepsilon\). Define
\[
\rho(\varepsilon)
=
\frac{(1-u(\varepsilon))(1-p)+u(\varepsilon)p}
{u(\varepsilon)(1-p)+p(1-u(\varepsilon))}.
\]
The information cascade threshold is then
\begin{equation}
\label{eq:k-binary}
k(\varepsilon)
=
\left\lfloor
\log_{\rho(\varepsilon)}\frac{1-p}{p}
\right\rfloor+1.
\end{equation}

The next theorem characterizes the probability that the platform's final decision is correct. Since the platform stops collecting feedback once the public history reaches the cascade threshold, this probability is exactly the probability that the first cascade reached by the public history is aligned with the true state.
\vspace{-5pt}
\begin{theorem}[Probability of a Correct Cascade]\label{thm:right-cascade-B}
Consider a correct signal probability \(p>1/2\) and a platform-chosen privacy level \(\varepsilon\). Let \(\alpha=(1-p)/p\), \(u(\varepsilon)=1/(1+e^{\varepsilon})\), and \(k(\varepsilon)\) be given by \eqref{eq:k-binary}. For a fixed cascade threshold \(k\), define
\[
v_k
=
\frac{1-\alpha^{\frac{k-2}{k-1}}}
{1-\alpha^{\frac{k-2}{k-1}}+\alpha^{\frac{-1}{k-1}}-\alpha},
\qquad
\varepsilon_k=\log\frac{1-v_k}{v_k}.
\]
For \(\varepsilon\in(\varepsilon_{k+1},\varepsilon_k]\), the cascade threshold remains equal to \(k\), and the probability that the platform reaches a correct final decision increases with \(\varepsilon\).
\end{theorem}
\vspace{-5pt}

This result shows how the privacy level affects the platform's decision accuracy through the amount of information accumulated before stopping. Within any interval \(\varepsilon\in(\varepsilon_{k+1},\varepsilon_k]\), the cascade threshold \(k\) remains fixed. In this region, increasing \(\varepsilon\) reduces the flip probability \(u(\varepsilon)\), so each released report becomes more informative while the stopping threshold remains unchanged. Consequently, the probability of a correct final decision increases.

However, this monotonic improvement holds only within a fixed-threshold interval. Once \(\varepsilon\) increases above a threshold value \(\varepsilon_k\), the cascade threshold decreases from \(k\) to \(k-1\). The platform then stops collecting feedback earlier, so the public history may contain less useful pre-cascade information when the final decision is made. This generates a non-monotonic relationship between the privacy level and the probability of a correct final decision: a less private protocol can make each individual report more informative, but it can also cause the platform to stop the learning process too early.

\begin{figure}[htbp]
    \centering
    \includegraphics[width=0.75\textwidth]{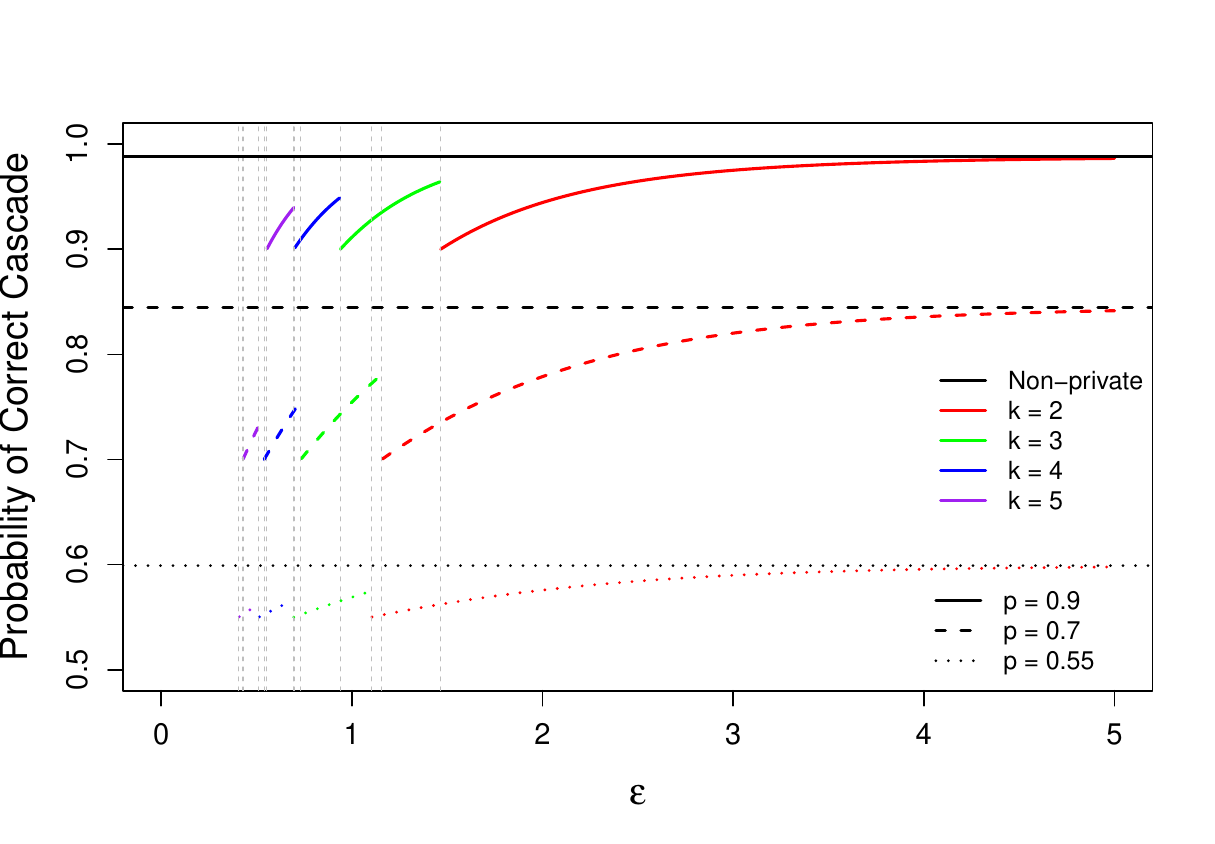}
    \caption{Probability of a correct final decision vs. privacy level \( \varepsilon \) for different cascade thresholds \( k \). Each colored line represents the probability of a correct final decision for a specific threshold \( k \).}
    \label{fig:binary_non_mon}
\end{figure}

\Cref{fig:binary_non_mon} illustrates how the probability of a correct final decision varies with the privacy level \( \varepsilon \) across different cascade thresholds \( k \), considering three values of \(p\): \(p=0.55\), \(p=0.7\), and \(p=0.9\). Each colored line represents a fixed value of \(k\), with vertical dashed lines indicating transition points between thresholds. The probability of a correct final decision does not increase monotonically with \(\varepsilon\). At certain transition points, increasing \(\varepsilon\) reduces the cascade threshold, causing the platform to stop collecting feedback earlier and decreasing the probability of a correct final decision. This effect is more pronounced when private signals are highly accurate. For example, when \(p=0.9\), the drop in the probability of a correct final decision during a threshold transition is more substantial than when \(p=0.55\).

To better understand these non-monotonic behaviors and their limiting cases, we use the analytical expressions governing the cascade threshold and the corresponding probability of a correct final decision. The probability of a correct final decision, derived in \Cref{lem:asymptotic-learning-B}, is
\[
\mathbb{P}(\text{correct final decision})
=
\frac{\rho(\varepsilon)^k-1}{\rho(\varepsilon)^{2k}-1}.
\]

The associated cascade threshold is given in \Cref{lem:expression-k} and restated in \eqref{eq:k-binary}. At the transition point $\varepsilon_k$, we have
\[
\log_{\rho(\varepsilon_k)}
\frac{1-p}{p}
=
k-1.
\]
Hence, as $\varepsilon\to\varepsilon_k^{+}$, the cascade threshold drops from $k$ to $k-1$, and the corresponding probability of a correct final decision satisfies
\[
\lim_{\varepsilon\to\varepsilon_k^{+}}
\mathbb{P}(\text{correct final decision})
=
p.
\]

Intuitively, at the threshold transition, the public evidence required to trigger a cascade is exactly as informative as a single private signal. Since a single private signal is correct with probability $p$, the resulting final decision also has limiting accuracy $p$.

\paragraph{Stopping time and participation.}

We next analyze the second component of the platform's objective: the time required to reach a stopping decision. The privacy level affects this stopping time through two channels. On the one hand, stronger privacy protection makes each submitted report less informative. On the other hand, stronger privacy protection may increase participation in the feedback process.

Let $\varepsilon_n$ denote agent $n$'s privacy tolerance, interpreted as the maximum privacy loss she is willing to accept. We assume that privacy tolerances are independently and identically distributed according to
\[
\varepsilon_n \sim U[0,a],
\qquad a>0,
\]
and are independent of agents' private signals. Thus, privacy tolerances have bounded nonnegative support and no point mass at zero.

Given a platform-chosen privacy level $\varepsilon \in [0,a]$, agent $i$ participates if and only if $\varepsilon_n \geq \varepsilon$. The participation probability is therefore
\[
\eta(\varepsilon)
=
\mathbb{P}(\varepsilon_n \geq \varepsilon)
=
1-\frac{\varepsilon}{a}.
\]
Thus, $\eta(\varepsilon)$ is decreasing in $\varepsilon$: a larger privacy parameter weakens privacy protection and reduces the set of agents willing to participate. Because privacy tolerance is independent of the private signal, participation affects the arrival rate of reports but does not change the signal accuracy among participating agents.

Under the binary local reporting protocol, the flip probability is
\[
u(\varepsilon)=\frac{1}{1+e^\varepsilon}.
\]
Before a cascade occurs, conditional on the true state \(\theta=+1\), each submitted report satisfies
\[
\mathbb{P}(x_m=+1)
=
p(1-u(\varepsilon))+(1-p)u(\varepsilon),
\]
and
\[
\mathbb{P}(x_m=-1)
=
(1-p)(1-u(\varepsilon))+p u(\varepsilon).
\]
The effective drift of the submitted-report process is therefore
\[
\mu(\varepsilon,p)
=
\mathbb{P}(x_m=+1)-\mathbb{P}(x_m=-1)
=
(2p-1)(1-2u(\varepsilon)).
\]
A smaller \(\varepsilon\) increases \(u(\varepsilon)\) and therefore decreases \(\mu(\varepsilon,p)\), making each submitted report less informative about the true state.

Let \(x_m\in\{-1,+1\}\) denote the \(m\)-th submitted privacy-protected report. Under the binary randomized response protocol, the belief-threshold stopping time defined in \Cref{def:binary-stopping-times} takes the form
\[
\tau_{\varepsilon,p}
=
\inf\left\{
m\geq 1:
\left|\sum_{\ell=1}^{m} x_\ell\right|
\geq k(\varepsilon,p)
\right\}.
\]
Thus, \(\tau_{\varepsilon,p}\) counts the number of submitted reports required for the public history to reach the belief threshold associated with either a \(+1\) or \(-1\) information cascade. The corresponding calendar stopping time \(T_{\varepsilon,p}\) counts the number of potential agent arrivals required for the same threshold to be reached.

\begin{theorem}[Stopping Time under Participation]\label{thm:stopping-time-binary}
Suppose privacy tolerances \(\varepsilon_n\) are independent of private signals, and each potential agent participates independently with probability \(\eta(\varepsilon)=\mathbb{P}(\varepsilon_n\geq\varepsilon)\). Under the binary local reporting protocol with flip probability \(u(\varepsilon)=1/(1+e^\varepsilon)\), the expected belief-threshold stopping time is
\begin{equation}
\label{eq:report-stopping-time}
\mathbb{E}[\tau_{\varepsilon,p}]
=
\frac{k(\varepsilon,p)}{\mu(\varepsilon,p)}
\cdot
\frac{1-\rho(\varepsilon,p)^{k(\varepsilon,p)}}{1+\rho(\varepsilon,p)^{k(\varepsilon,p)}},
\end{equation}
where
\[
\mu(\varepsilon,p)
=
(2p-1)(1-2u(\varepsilon)),
\]
and
\[
\rho(\varepsilon,p)
=
\frac{(1-p)(1-u(\varepsilon))+p u(\varepsilon)}
{p(1-u(\varepsilon))+(1-p)u(\varepsilon)}.
\]
Moreover, the expected calendar stopping time is
\begin{equation}
\label{eq:calendar-stopping-time}
\mathbb{E}[T_{\varepsilon,p}]
=
\frac{\mathbb{E}[\tau_{\varepsilon,p}]}{\eta(\varepsilon)}
=
\frac{1}{\eta(\varepsilon)}
\cdot
\frac{k(\varepsilon,p)}{\mu(\varepsilon,p)}
\cdot
\frac{1-\rho(\varepsilon,p)^{k(\varepsilon,p)}}{1+\rho(\varepsilon,p)^{k(\varepsilon,p)}}.
\end{equation}
\end{theorem}

The theorem decomposes the expected calendar stopping time into two components. The first component,
\[
\mathbb{E}[\tau_{\varepsilon,p}]
=
\frac{k(\varepsilon,p)}{\mu(\varepsilon,p)}
\cdot
\frac{1-\rho(\varepsilon,p)^{k(\varepsilon,p)}}{1+\rho(\varepsilon,p)^{k(\varepsilon,p)}},
\]
is the expected belief-threshold stopping time. This term captures the information-quality channel. Stronger privacy protection reduces \(\varepsilon\), increases the flip probability \(u(\varepsilon)\), and decreases the effective drift \(\mu(\varepsilon,p)\). As a result, each submitted report contains less information about the true state, so more reports are typically required for the public belief to reach the stopping threshold.

The second component, \(1/\eta(\varepsilon)\), captures the participation channel. Stronger privacy protection increases the participation probability \(\eta(\varepsilon)\), because more agents are willing to participate when the privacy loss is smaller. This reduces the expected number of potential agent arrivals required to obtain each submitted report.

Therefore, the expected calendar stopping time can be non-monotonic in the privacy level. A smaller \(\varepsilon\) makes each report noisier and tends to increase the number of submitted reports required for a stopping decision, but it also increases participation and can accelerate the arrival of reports. Conversely, a larger \(\varepsilon\) makes each submitted report more informative but may reduce the participation rate. The platform must therefore choose the privacy level by balancing final decision accuracy against the expected time required to collect enough feedback for a stopping decision.

To illustrate the stopping-time trade-off, we numerically evaluate both the belief-threshold stopping time and the calendar stopping time under the binary randomized response protocol. 
The former measures how many received reports are needed to reach an information cascade, whereas the latter also accounts for the participation probability induced by the platform's privacy choice. 
This comparison highlights the two channels through which privacy affects decision speed: report informativeness and report arrival.

Figure~\ref{fig:binary_stopping_time} compares these two notions of stopping time. The left panel plots the expected belief-threshold stopping time, while the right panel plots the expected calendar stopping time. In both panels, colors indicate the information cascade threshold \(k\), line styles correspond to different signal accuracies \(p\), and the horizontal black lines show the corresponding non-private benchmarks.

\begin{figure}[htbp]
    \centering
    \includegraphics[width=0.48\textwidth]{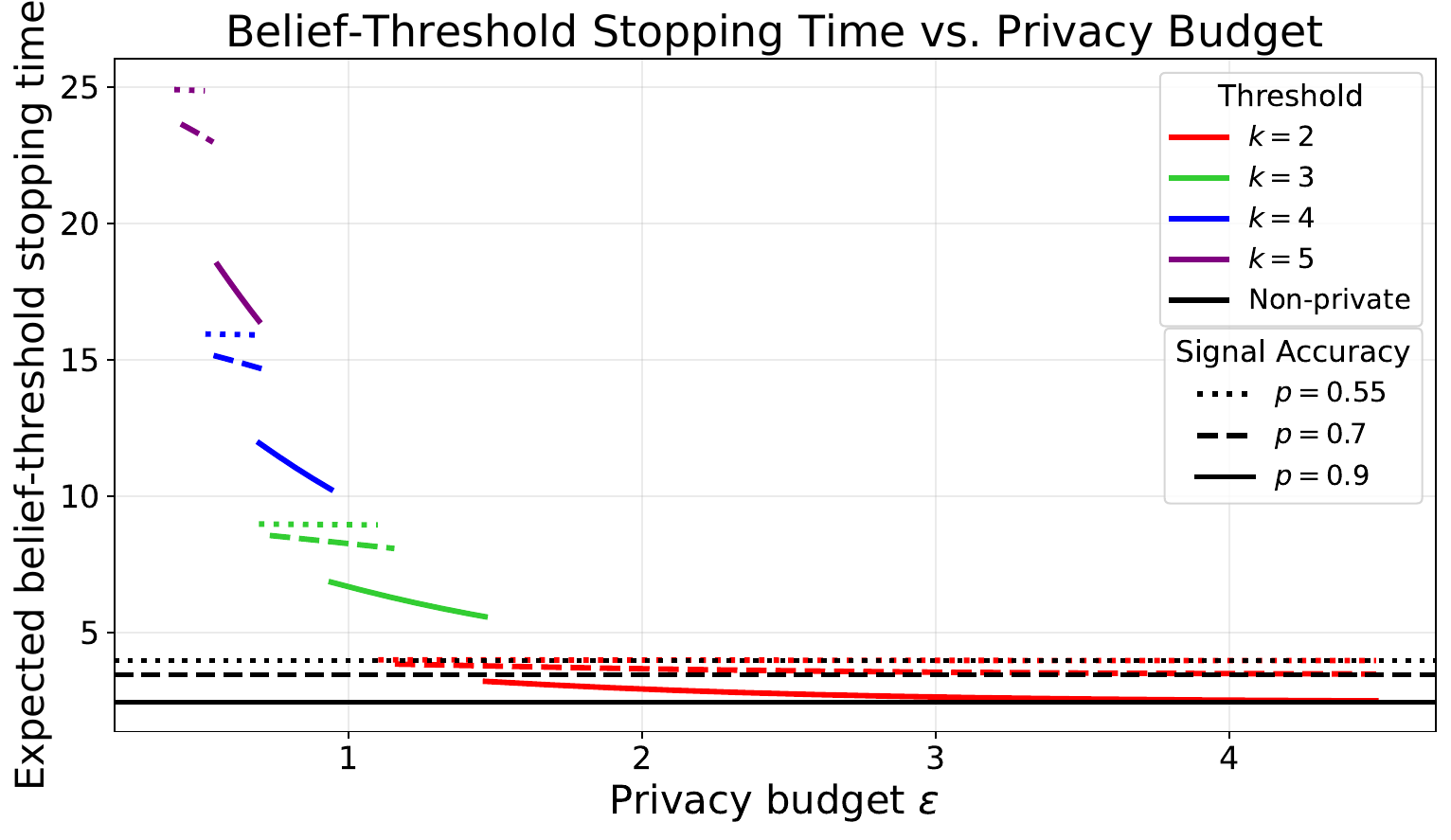}
    \hfill
    \includegraphics[width=0.48\textwidth]{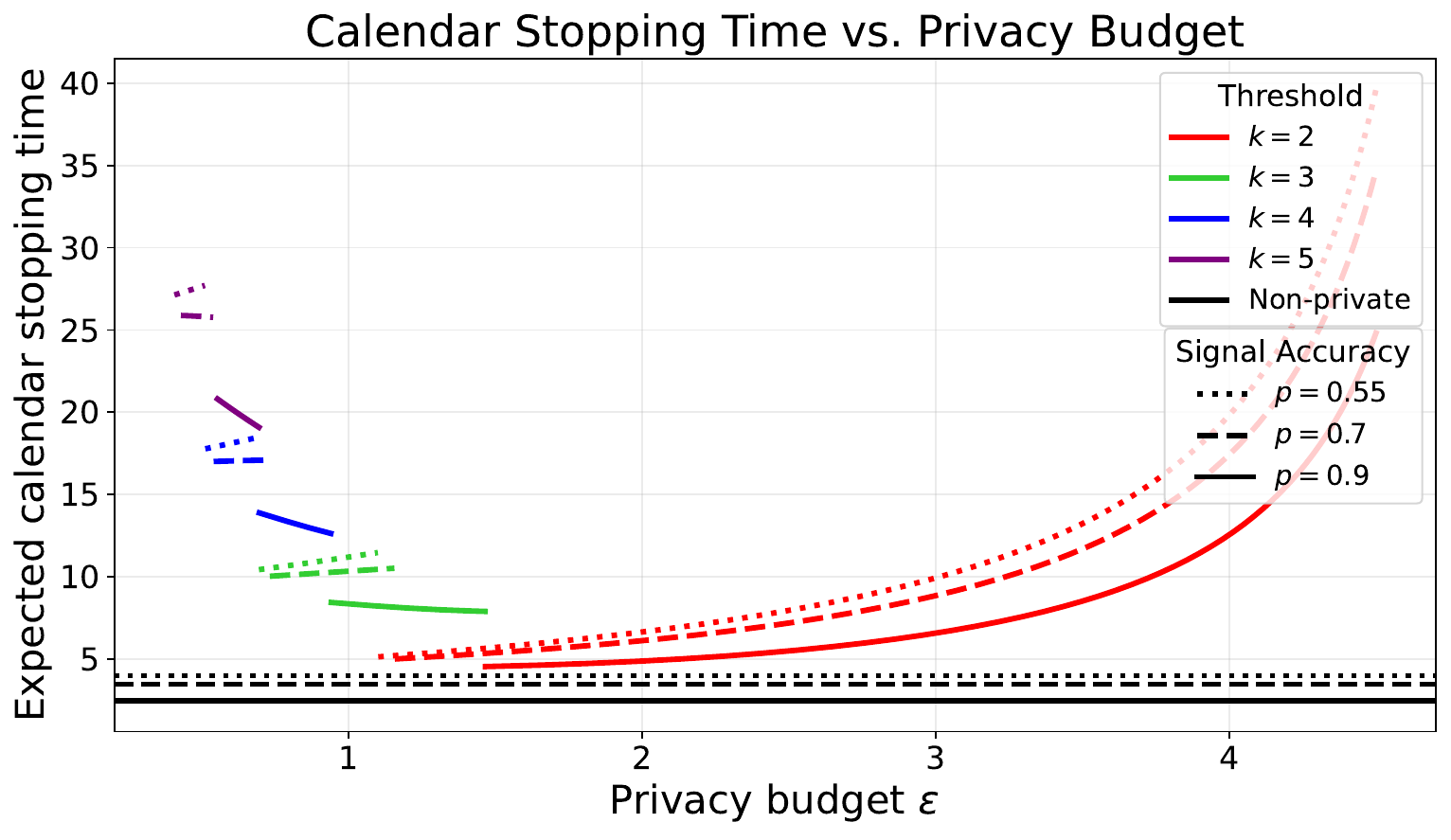}
    \caption{Expected stopping time as a function of the privacy budget \(\varepsilon\) in the binary signal model. Left: belief-threshold stopping time, which counts the number of received reports needed to reach an information cascade. Right: calendar stopping time, which adjusts the belief-threshold stopping time by the participation probability \(\eta(\varepsilon)\). Colors indicate the information cascade threshold \(k\), and line styles correspond to different signal accuracies \(p\).}
    \label{fig:binary_stopping_time}
\end{figure}

The figure shows that stopping time is generally non-monotone in the privacy budget. When \(\varepsilon\) is small, reports are heavily randomized, so each submitted report contains less information and the platform needs more submitted reports to reach a cascade. As \(\varepsilon\) increases, reports become more informative, which can reduce the belief-threshold stopping time within a fixed-threshold region. However, changes in the cascade threshold create discontinuities and can reverse this local trend. The calendar-time panel converts these report-level dynamics into the platform's actual waiting time by accounting for participation. Thus, the numerical results show that the platform's privacy choice affects decision speed through both the informativeness of submitted reports and the rate at which such reports arrive.

We next examine the platform's overall objective, which combines final decision accuracy and the cost of delayed stopping. For a fixed signal accuracy \(p\) and delay cost \(c\), we evaluate
\[
J(\varepsilon)
=
\mathbb{P}_{\varepsilon,\mathcal M}(D_{\varepsilon,\mathcal M}=\theta)
-
c\mathbb{E}_{\varepsilon,\mathcal M}[T_{\varepsilon,\mathcal M}]
\]
under the binary randomized response protocol. Figure~\ref{fig:binary_platform_objective} plots the platform objective as a function of the privacy budget \(\varepsilon\) for three levels of signal accuracy, \(p\in\{0.55,0.7,0.9\}\). We fix the stopping cost at \(c=0.01\) and assume heterogeneous privacy tolerances \(\varepsilon_n\sim U[0,5]\).

Across all three values of \(p\), the platform objective is non-monotone in the privacy budget. When \(\varepsilon\) is very small, reports are heavily randomized, so the platform requires more time to accumulate useful information and the stopping-cost term lowers the objective. As \(\varepsilon\) increases, reports become more informative and the objective initially improves. When \(\varepsilon\) becomes too large, however, weaker privacy protection reduces participation, increasing the expected calendar stopping time and eventually lowering the objective. Thus, the platform may prefer an interior privacy level that balances report informativeness, participation, and decision speed.

The comparison across signal accuracies also shows that the platform achieves a higher objective when private signals are more informative. A larger \(p\) improves the quality of the underlying information available to the reporting process, raising the attainable objective over most of the privacy-budget range. At the same time, the non-monotone relationship with \(\varepsilon\) remains present across different signal accuracies, indicating that the basic privacy--participation trade-off is robust to changes in signal quality.

\begin{figure}[htbp]
    \centering
    \includegraphics[width=0.85\textwidth]{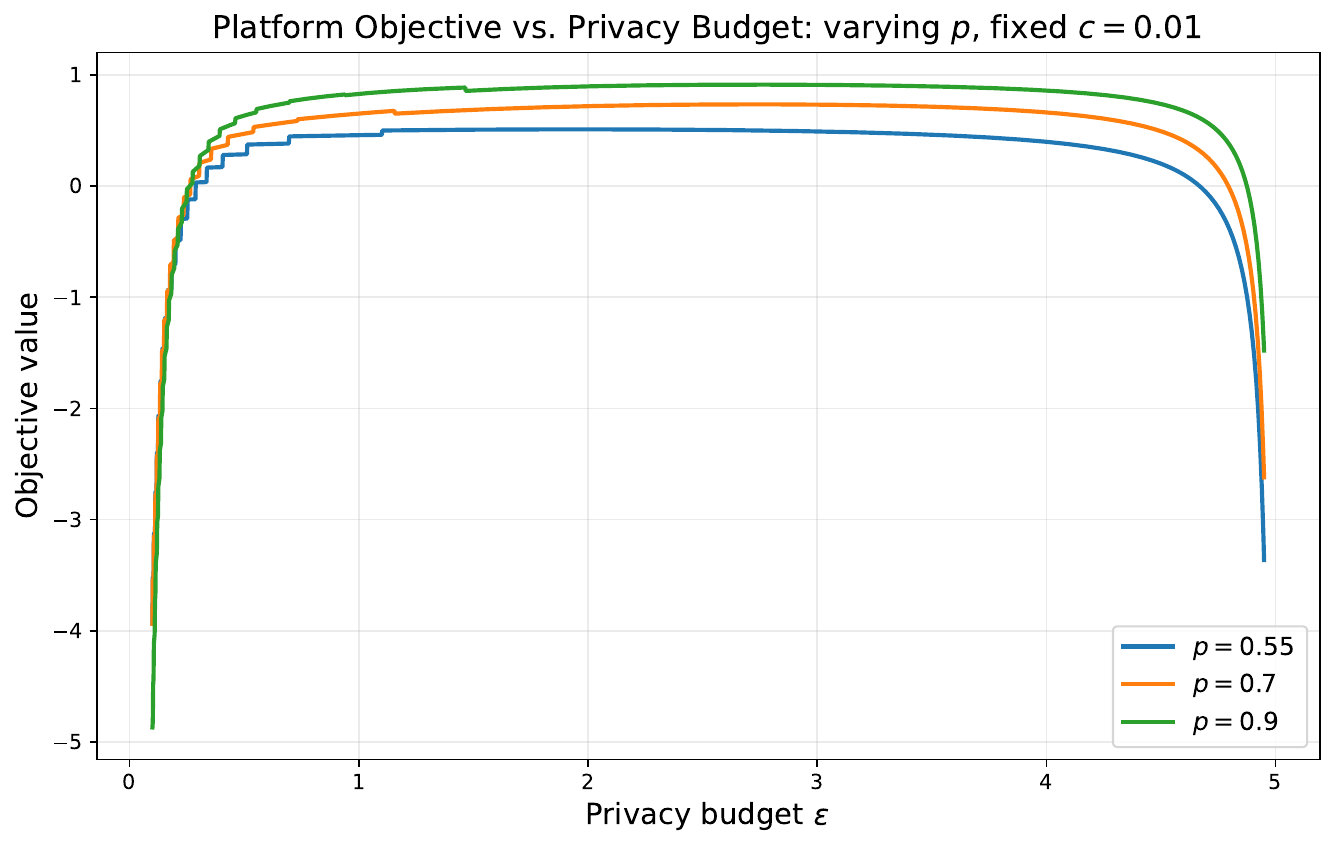}
    \caption{Platform objective as a function of the privacy budget \(\varepsilon\) in the binary signal model for different levels of signal accuracy \(p\). The objective combines the probability of a correct final decision and the expected calendar stopping-time cost. We set \(p\in\{0.55,0.7,0.9\}\), \(c=0.01\), and \(\varepsilon_n\sim U[0,5]\).}
    \label{fig:binary_platform_objective}
\end{figure}

\subsection*{The Binary Signals Model with Personalized Privacy Mechanisms}
\label{sec:heterogeneous-privacy-binary}

In the preceding analysis, the platform applies a common privacy level $\varepsilon$ to all participating agents. We now consider a personalized privacy mechanism in which the platform tailors the reporting protocol to each agent's privacy tolerance. Specifically, agent $n$ has an individual privacy tolerance $\varepsilon_n$, interpreted as the maximum privacy loss she is willing to accept. The platform applies an $\varepsilon_n$-local randomized response mechanism to agent $n$'s report. Under this personalized mechanism, every agent is offered a privacy guarantee that is acceptable to her, and therefore all agents participate in the data collection process.

In the binary signal setting, personalized privacy induces heterogeneous flipping probabilities. Agent \(n\)'s privatized report is generated by randomized response with
\[
u_n
=
u(\varepsilon_n)
=
\frac{1}{1+e^{\varepsilon_n}}.
\]
Thus, agents with stronger privacy requirements have smaller \(\varepsilon_n\) and larger flipping probabilities, while agents with weaker privacy requirements have larger \(\varepsilon_n\) and smaller flipping probabilities.

Since the privacy level varies across agents, the informativeness of reports is no longer governed by a single flip probability. Let
\[
\bar{u}
=
\mathbb{E}_{\varepsilon_n}\left[
\frac{1}{1+e^{\varepsilon_n}}
\right]
\]
denote the average flipping probability induced by the personalized privacy mechanism. Before an information cascade occurs, each agent's action is determined by her private signal, and the released report is a randomized version of this action. Hence, from the perspective of later agents and the platform, each report is informative on average according to \(\bar{u}\).

Under this personalized mechanism, the public log-likelihood ratio is obtained by replacing the homogeneous flipping probability \(u(\varepsilon)\) with the average flipping probability \(\bar{u}\). If the public history contains \(k\) more \(+1\) reports than \(-1\) reports, the public belief is
\[
l_n
=
\log
\left[
\frac{(1-\bar{u})p+\bar{u}(1-p)}
{(1-\bar{u})(1-p)+\bar{u}p}
\right]^k.
\]
Accordingly, the cascade threshold is
\[
k(\bar{u},p)
=
\left\lfloor
\log_{
\frac{(1-\bar{u})(1-p)+\bar{u}p}
{\bar{u}(1-p)+p(1-\bar{u})}
}
\frac{1-p}{p}
\right\rfloor+1.
\]

\begin{theorem}[Probability of Correct Cascade under Personalized Privacy]\label{thm:right-cascade-B-hete}
Consider a correct signal probability \(p>1/2\). Under the personalized privacy mechanism, each agent \(n\) is assigned an individual privacy level \(\varepsilon_n\) and reports through randomized response with flip probability \(u_n=1/(1+e^{\varepsilon_n})\). Let
\[
\bar{u}
=
\mathbb{E}_{\varepsilon_n}
\left[
\frac{1}{1+e^{\varepsilon_n}}
\right],
\qquad
\alpha=\frac{1-p}{p}.
\]
For a fixed cascade threshold \(k=k(\bar{u},p)\), define
\[
\tilde{v}_k
=
\frac{1-\alpha^{\frac{k-2}{k-1}}}
{1-\alpha^{\frac{k-2}{k-1}}+\alpha^{\frac{-1}{k-1}}-\alpha}.
\]
For \(\bar{u}\in[\tilde{v}_k,\tilde{v}_{k+1})\), the cascade threshold remains equal to \(k\), and the probability of a correct cascade decreases with \(\bar{u}\).
\end{theorem}

The personalized privacy mechanism eliminates the participation channel in the stopping-time analysis. Because each agent receives a privacy guarantee tailored to her own tolerance, all agents are willing to participate in the data collection process. Therefore, the expected calendar stopping time coincides with the expected number of submitted reports needed to reach the stopping threshold. The only remaining effect of privacy on this quantity is through the informativeness of the submitted reports.

Specifically, before a cascade occurs, the probability that a submitted report agrees with the true state is
\[
p(1-\bar{u})+(1-p)\bar{u},
\]
while the probability that it disagrees with the true state is
\[
(1-p)(1-\bar{u})+p\bar{u}.
\]
The effective drift of the submitted-report process is therefore
\[
\mu(\bar{u},p)
=
(2p-1)(1-2\bar{u}).
\]
The expected stopping time is
\[
\mathbb{E}[T_{\bar{u},p}]
=
\frac{k(\bar{u},p)}{\mu(\bar{u},p)}
\cdot
\frac{1-\rho(\bar{u},p)^{k(\bar{u},p)}}
{1+\rho(\bar{u},p)^{k(\bar{u},p)}},
\]
where
\[
\rho(\bar{u},p)
=
\frac{(1-p)(1-\bar{u})+p\bar{u}}
{p(1-\bar{u})+(1-p)\bar{u}}.
\]

This expression shows that, under personalized privacy, the data collection time is governed only by report informativeness. Stronger average privacy protection increases \(\bar{u}\), reduces the drift \(\mu(\bar{u},p)\), and therefore increases the number of reports needed to reach an information cascade. Unlike the uniform privacy mechanism, there is no participation-rate penalty or benefit, because all agents participate by construction.

\section{Learning Efficiency with Continuous Signals}
\label{sec:continuous-model}

Building on the binary model, we now consider a continuous signal environment. The unknown state is still binary, denoted by \(\theta\in\{-1,+1\}\), with prior probability \(1/2\) on each state. Each agent \(n\) receives a private signal \(s_n\in\mathbb{R}\), drawn independently conditional on the true state according to
\[
s_n \sim \mathcal{N}(\theta,\sigma^2).
\]
Thus, larger signals provide stronger evidence in favor of state \(+1\), while smaller signals provide stronger evidence in favor of state \(-1\).

As in the binary model, each agent forms an initial action based on her private signal and the public history. The platform's privacy-preserving mechanism transforms this action into a privacy-protected public report \(x_n\in\{-1,+1\}\). The sequence of released reports forms the public history and is used to update the platform's belief about the underlying state.

Let \(l_n\) denote the public log-likelihood ratio before agent \(n\) acts:
\begin{equation}
\label{eq:public-LLR-continuous}
l_n
=
\log
\frac{
\mathbb{P}(\theta=+1\mid x_1,\ldots,x_{n-1})
}{
\mathbb{P}(\theta=-1\mid x_1,\ldots,x_{n-1})
}.
\end{equation}
A positive value of \(l_n\) indicates that the public history favors state \(+1\), while a negative value indicates that it favors state \(-1\).

Unlike the binary signal model, the continuous Gaussian model does not generally lead to an irreversible information cascade. Even when the public belief strongly favors one state, a sufficiently strong private signal can still overturn the current public belief. We therefore characterize the platform's stopping decision using a belief threshold.

The platform specifies a belief threshold \(B>0\). Once the absolute public log-likelihood ratio reaches this threshold, the platform stops collecting feedback and makes a final decision. To formalize this stopping decision, we distinguish between the number of submitted reports required to reach the belief threshold and the corresponding number of potential agent arrivals.

\begin{definition}[Belief-Threshold and Calendar Stopping Times in the Continuous Model]
\label{def:continuous-stopping-times}
For a belief threshold \(B>0\), the \emph{belief-threshold stopping time} is
\begin{equation}
\label{eq:belief-threshold-stopping-time}
\tau_B
=
\inf\{m\geq 1: |l_m|\geq B\}.
\end{equation}
Thus, \(\tau_B\) counts the number of submitted reports required for the public log-likelihood ratio to reach the belief threshold \(B\).

The corresponding \emph{calendar stopping time}, denoted by \(T_B\), counts the number of potential agents who arrive before the same threshold is reached. If each potential agent participates independently with probability \(\eta(\varepsilon)\), then
\[
\mathbb{E}[T_B]
=
\frac{\mathbb{E}[\tau_B]}{\eta(\varepsilon)}.
\]
\end{definition}

At the belief-threshold stopping time, the platform makes the final decision
\begin{equation}
\label{eq:company-final-decision-continuous}
D_B=\operatorname{sign}(l_{\tau_B}).
\end{equation}
Thus, in our motivating example, \(D_B=+1\) corresponds to adopting or expanding the AI tool, while \(D_B=-1\) corresponds to rejecting, discontinuing, or substantially revising it.

The platform's objective is to choose a privacy-preserving reporting protocol that balances final decision accuracy against the time required to reach a decision. Let \(\mathcal{F}_{\varepsilon,d_{\mathbb{R}}}\) denote the set of local reporting mechanisms satisfying \((\varepsilon,d_{\mathbb{R}})\)-metric differential privacy on the continuous signal space. For a given belief threshold \(B\), the platform's protocol design problem in \eqref{eq:company-objective-general} can be rewritten as
\begin{equation}
\label{eq:company-objective-continuous}
\max_{\varepsilon>0,\ \mathcal{M}\in\mathcal{F}_{\varepsilon,d_{\mathbb{R}}}}
\quad
\mathbb{P}_{\varepsilon,\mathcal{M}}(D_B=\theta)
-
c\,\mathbb{E}_{\varepsilon,\mathcal{M}}[T_B],
\end{equation}
where \(c>0\) captures the cost of delayed decision making, continued experimentation, or prolonged feedback collection. The first term measures the accuracy of the platform's final decision, while the second term captures the expected calendar time needed for the platform to reach the belief threshold.

Metric differential privacy is particularly suitable for the continuous signal model because it allows the privacy guarantee to vary with the distance between private signals. In the Gaussian setting, signals are unbounded and differ in how strongly they support one state over the other. For any public belief \(l_n\), there exists a decision threshold \(t(l_n)\) such that an agent's initial recommendation switches from \(-1\) to \(+1\) when her private signal \(s_n\) crosses this threshold. In the Gaussian model, this threshold is
\[
t(l_n)=-\frac{\sigma^2 l_n}{2}.
\]
Signals close to \(t(l_n)\) are locally sensitive, since a small perturbation may change the agent's recommendation. Signals far from the threshold, by contrast, provide robust evidence for one state and are less sensitive to small perturbations.

\begin{figure}[htbp]
    \centering
    \includegraphics[width=0.5\textwidth]{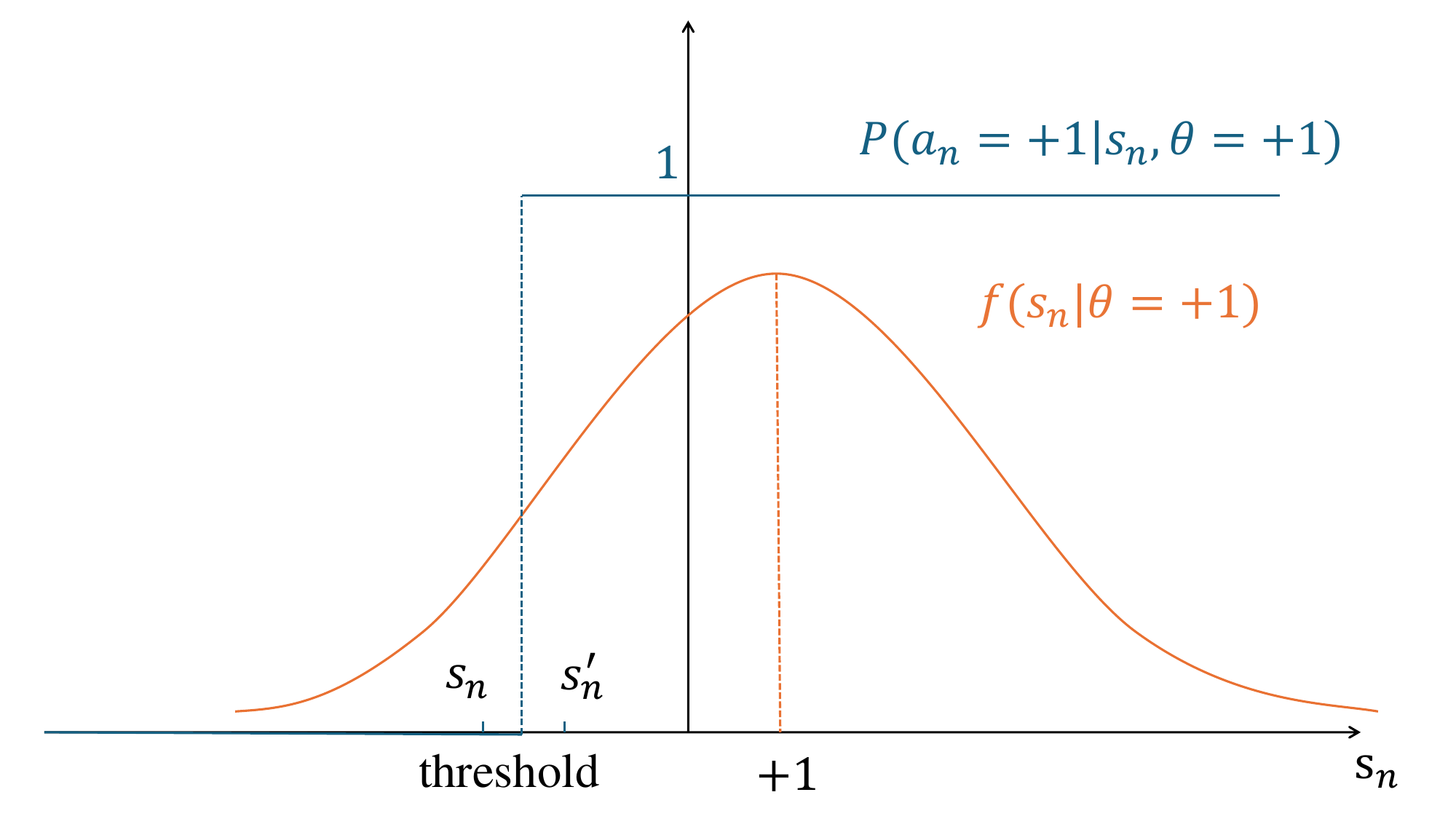}
    \caption{The blue step function shows the probability of action \(+1\) given signal \(s_n\), and the orange Gaussian curve shows the signal distribution under \(\theta=+1\). The $\pm 1$ actions are most sensitive to signal perturbations near the threshold \(t(l_n)\).}
    \label{fig:action_and_normal}
\end{figure}

Standard differential privacy treats all signal pairs as equally sensitive and imposes a uniform indistinguishability requirement over the entire signal space. This is too restrictive for unbounded Gaussian signals: it forces the reporting protocol to maintain persistent randomization even for signals far from the decision threshold, where the initial recommendation is highly robust. Such persistent randomization prevents released reports from becoming asymptotically correct: under fixed-probability randomized response, their error probability converges to \(u(\varepsilon)>0\), even as public beliefs learn the true state (Appendix~\ref{app:failure-asymptotic}). 

By contrast, metric differential privacy requires nearby signals to induce similar report distributions while allowing more distant signals to induce more distinguishable reports. This matches the structure of continuous sequential learning: the protocol protects signals near the decision threshold, where small changes may alter the recommendation, without imposing excessive distortion on signals that provide strong evidence for one state (far from the threshold \(t(l_n)\)). We also consider Pufferfish privacy as an alternative privacy framework in Appendix~\ref{app:pufferfish}, where a staircase randomized-response mechanism yields analogous asymptotic learning and stopping-time results.

We evaluate metric-DP reporting protocols under the platform objective in \eqref{eq:company-objective-continuous} through the evolution of the public log-likelihood ratio \(l_n\). The first component of the objective, the probability of a correct final decision, depends on whether the public belief reaches the boundary corresponding to the true state. We also evaluate the limiting accuracy of individual released reports using the following asymptotic learning criterion. This criterion refers to the underlying reporting process continued indefinitely.

\begin{definition}[Asymptotic Learning]\label{def:AL}
Let \(p_n=\mathbb{P}(x_n=\theta)\) denote the probability
that the \(n\)-th released report matches the true state.
We say that asymptotic learning occurs if and only if
\[
\liminf_{n\to\infty}p_n=1.
\]
\end{definition}

Asymptotic learning concerns the limiting correctness of individual released reports. Under the smooth randomized response mechanisms studied below, we also establish that public beliefs converge to the true state. This convergence allows the platform to make a correct final decision with probability approaching one as the belief threshold increases, providing a high-confidence correctness guarantee for the first term in the platform objective \eqref{eq:company-objective-continuous}.

The second component of the objective is the expected stopping time, which depends on how quickly the public belief reaches the threshold \(B\). We analyze this term using the convergence rate of the public log-likelihood ratio.

\begin{definition}[Convergence Rate of the Public Log-Likelihood Ratio]\label{def:public-llr-rate}
Let \(f(n)\) be a positive increasing function. We say that \(f(n)\) characterizes the convergence rate of the public log-likelihood ratio if, conditional on the true state $(\theta=+1)$,
\[
\lim_{n\to\infty}\frac{l_n}{f(n)}=1,
\quad \text{with probability one}.
\]
\end{definition}

This definition connects to the stopping-time term in the platform objective. For the growth rates established below, conditional on \(\theta=+1\), the belief-threshold stopping time satisfies
\[
f(\tau_B)\sim B,
\qquad \text{a.s. as }B\to\infty.
\]
For logarithmic belief growth, this relation identifies the leading order of \(\log\tau_B\), but does not imply \(\tau_B\sim f^{-1}(B)\). The expected stopping time and its leading constant require a separate analysis of the stopped process. A faster-growing \(f(n)\) allows the platform to reach the same belief threshold with fewer sequential reports, reducing the second term in \eqref{eq:company-objective-continuous}.

\subsection{Main Results for Continuous Signals}
\label{sec:continuous}

We now analyze the platform objective in the continuous signal model. Recall from \Cref{def:continuous-stopping-times} that the platform stops when the public log-likelihood ratio reaches the belief threshold \(B\), and then makes the final decision according to \eqref{eq:company-final-decision-continuous}.

To interpret the role of \(B\) in decision accuracy, we use the threshold value to obtain a conservative bound, allowing for boundary overshoot induced by discrete report updates. Accordingly, the posterior probability assigned to the chosen state is at least
\[
\frac{e^B}{1+e^B}.
\]
This lower bound specifies the guaranteed posterior confidence level associated with \(B\). 

Thus, a common threshold \(B\) guarantees the same minimum posterior confidence across reporting protocols, although actual decision accuracy can differ because of boundary overshoot. 
The reporting protocol mainly affects how quickly the public belief reaches this threshold. Therefore, the central object for comparing reporting protocols is the expected belief-threshold stopping time \(\mathbb{E}[\tau_B]\), which measures how many released reports are needed for the public belief to reach the threshold \(B\). Under a homogeneous privacy protocol, the corresponding expected calendar stopping time is
\[
\mathbb{E}[T_B]
=
\frac{\mathbb{E}[\tau_B]}{\eta(\varepsilon)}.
\]

We analyze this stopping-time component in two complementary ways. First, for a finite belief threshold \(B\), we characterize the exact belief threshold stopping-time problem through a Bellman equation for the public log-likelihood ratio. This provides a general framework for comparing different metric-DP reporting protocols at the level of received reports. Second, for high-confidence decisions where \(B\) is large, we characterize the asymptotic growth rate of the public log-likelihood ratio and separately analyze the expected belief-threshold stopping time using the stopped process. The finite-\(B\) analysis characterizes how many submitted reports are required for the public belief to reach the belief threshold. The corresponding platform-relevant calendar stopping time is then obtained by adjusting for participation.

\subsubsection{Finite-Threshold Stopping-Time Characterization}
\label{sec:finite-threshold-stopping}
We begin with the finite-threshold stopping problem at the submitted-report level.
Throughout this subsection, \(l\) denotes the \emph{public} log-likelihood ratio.
We consider a belief-Markov reporting protocol \(\mathcal M\), under which the dependence
of the reporting rule on the public history is summarized by the current public belief \(l\).

Let
\[
Q^{\mathcal M}(x\mid s,l)
=
\mathbb{P}_{\mathcal M}
\bigl(x_n=x\mid s_n=s,l_n=l\bigr),
\qquad x\in\{-1,+1\},
\]
denote the probability of releasing report \(x\) conditional on the agent's private signal
\(s\) and the current public log-likelihood ratio \(l\). Since the private signal is not
observed by the platform or by subsequent agents, it is integrated out when updating the
public belief. Thus, under state \(\theta\),
\[
P_\theta^{\mathcal M}(x\mid l)
=
\int_{\mathbb R}
Q^{\mathcal M}(x\mid s,l) f_\theta(s)\,ds,
\qquad x\in\{-1,+1\},
\]
where \(f_\theta\) denotes the density of the private signal conditional on state \(\theta\).
Hence, \(P_\theta^{\mathcal M}(x\mid l)\) is the conditional probability of observing
report \(x\) given the current public belief and the underlying state.

After observing the released report \(x\), Bayes' rule gives the public
log-likelihood-ratio update
\[
l'
=
l+\Delta_{\mathcal M}(l,x),
\qquad
\Delta_{\mathcal M}(l,x)
=
\log
\frac{P_{+}^{\mathcal M}(x\mid l)}
     {P_{-}^{\mathcal M}(x\mid l)}.
\]
Thus, the effect of the agent's private signal on public belief updating is captured
through the state-dependent report probabilities
\(P_\theta^{\mathcal M}(x\mid l)\), rather than through an additional private-signal
term in the public log-likelihood ratio.

For the symmetric Gaussian environment and reporting protocols considered below,
it is sufficient to condition on the true state \(\theta=+1\).
Starting from public belief \(L_0=l\), define
\[
\tau_B(l)
=
\inf\{m\geq 0: |L_m|\geq B\},
\]
where
\[
L_{m+1}
=
L_m+\Delta_{\mathcal M}(L_m,X_{m+1}).
\]
Let
\[
H_{\mathcal M}(l;B)
=
\mathbb{E}_{+}^{\mathcal M}[\tau_B(l)]
\]
denote the expected number of additional submitted reports required for the public
log-likelihood ratio to reach the belief threshold \(B\), starting from public belief \(l\).

For \(|l|<B\), \(H_{\mathcal M}\) satisfies the Bellman equation
\begin{equation}
\label{eq:bellman-stopping-continuous}
H_{\mathcal M}(l;B)
=
1+
\sum_{x\in\{-1,+1\}}
P_{+}^{\mathcal M}(x\mid l)
H_{\mathcal M}
\bigl(l+\Delta_{\mathcal M}(l,x);B\bigr),
\end{equation}
with boundary condition
\[
H_{\mathcal M}(l;B)=0,
\qquad |l|\geq B.
\]

Thus, from the initial public belief \(l=0\), the expected number of submitted reports
required to reach the belief threshold is \(H_{\mathcal M}(0;B)\). If each potential
agent participates independently with probability \(\eta(\varepsilon)\), the
corresponding expected calendar stopping time is
\[
\mathbb{E}_{+}^{\mathcal M}[T_B]
=
\frac{H_{\mathcal M}(0;B)}{\eta(\varepsilon)}.
\]

This Bellman equation provides an exact finite-\(B\) characterization of the
belief-threshold stopping-time component for belief-Markov reporting protocols.
In general, it does not admit a closed-form solution because both the conditional
report probabilities \(P_\theta^{\mathcal M}(x\mid l)\) and the likelihood-ratio
increments \(\Delta_{\mathcal M}(l,x)\) depend on the current public belief and on
the chosen reporting protocol. However, it can be solved numerically and used to
compare different metric-DP mechanisms under the same belief threshold. The
resulting belief-threshold stopping time can then be converted into the
platform-relevant calendar stopping time through the participation adjustment above.

To compare different metric-DP reporting protocols, we consider a general class of threshold-based flip mechanisms. Given public belief \(l_n\), the Gaussian decision threshold is
\[
t(l_n)=-\frac{\sigma^2 l_n}{2}.
\]
The initial recommendation is \(a_n=+1\) if \(s_n>t(l_n)\), and \(a_n=-1\) otherwise. A threshold-based flip protocol is described by a function \(\psi:[0,\infty)\to[0,1/2]\), where \(\psi(d)\) is the probability of flipping the initial recommendation when the signal is at distance \(d=|s_n-t(l_n)|\) from the decision threshold:
\[
\mathbb{P}_{\mathcal M}(x_n=+1\mid s_n,l_n)
=
\begin{cases}
1-\psi(|s_n-t(l_n)|), & s_n>t(l_n),\\
\psi(|s_n-t(l_n)|), & s_n<t(l_n).
\end{cases}
\]
We assume \(\psi(0)=1/2\) and that \(\psi(d)\) is non-increasing in \(d\). This captures the idea that reports are maximally randomized near the decision threshold and become more stable as the private signal moves farther away from the threshold.

This class allows us to separate the privacy constraint from the specific functional form
of the reporting protocol. Different choices of \(\psi\) may all satisfy metric differential
privacy, but they can induce different conditional report probabilities
\(P_\theta^{\mathcal M}(x\mid l)\), different log-likelihood updates
\(\Delta_{\mathcal M}(l,x)\), and therefore different finite-\(B\) stopping times through
the Bellman equation in \eqref{eq:bellman-stopping-continuous}. Thus, the choice of
\(\psi\) is part of the reporting protocol design problem.

One canonical example is the smooth randomized response protocol, which corresponds to the exponential flip function
\[
\psi_{\varepsilon}(d)
=
\frac{1}{2}e^{-\varepsilon d}.
\]
Under this protocol, the reporting interface flips the agent's initial recommendation with probability
\[
P(x_n\neq a_n\mid s_n,l_n)
=
\frac{1}{2}e^{-\varepsilon |s_n-t(l_n)|}.
\]
Thus, reports are most randomized when the private signal is close to the current decision threshold and become more stable as the signal moves farther away from the threshold. This mechanism satisfies \((\varepsilon,d_{\mathbb R})\)-metric differential privacy.

\begin{definition}[Smooth Randomized Response]\label{def:SRR}
A local reporting mechanism \(\mathcal{M}_{\varepsilon}(s_n; h_{n-1})\) is called a smooth randomized response protocol with privacy parameter \(\varepsilon\) if, given an initial recommendation \(a_n\), the released report \(x_n\) is obtained by flipping \(a_n\) with probability
\[
u_n(s_n)
=
\frac{1}{2}e^{-\varepsilon |s_n-t(l_n)|}.
\]
Equivalently,
\[
\mathbb{P}_{\mathcal M_{\varepsilon}}(x_n\neq a_n\mid s_n,l_n)
=
\frac{1}{2}e^{-\varepsilon |s_n-t(l_n)|}.
\]
\end{definition}

In the main text, we focus on smooth randomized response as a tractable benchmark for analyzing belief dynamics and stopping time under metric differential privacy. This protocol is parameterized by the platform-chosen privacy level \(\varepsilon\) and satisfies the \((\varepsilon,d_{\mathbb R})\)-metric differential privacy constraint. Other threshold-based flip functions may also satisfy the same metric-DP constraint and may lead to different finite-\(B\) stopping-time behavior. 
In Appendix~\ref{app:alternative-mdp-mechanisms}, we analyze polynomial-decay flip functions, establishing their metric-DP feasibility and characterizing the resulting learning rates and stopping times.

Before analyzing the induced learning dynamics, we first record the privacy feasibility of smooth randomized response.

\begin{theorem}[Metric-DP Feasibility of Smooth Randomized Response]
\label{thm:smoothly-random-response-C}
In the sequential learning model with Gaussian signals, the smooth randomized response protocol with privacy parameter \(\varepsilon\) satisfies \((\varepsilon,d_{\mathbb{R}})\)-metric differential privacy.
\end{theorem}

Unlike conventional randomized response, the smooth randomized response protocol makes the amount of randomization depend on the distance between the received signal and the decision threshold. For signals \(s_n\) far from the threshold, the protocol adds less noise, thereby avoiding excessive distortion of highly informative reports. For signals close to the threshold, the protocol provides stronger randomization because small changes in the private signal may alter the initial recommendation. This structure makes smooth randomized response a useful benchmark for studying how metric-DP reporting affects both information accumulation and stopping time.

The finite-\(B\) stopping-time behavior of this protocol can be evaluated through the Bellman equation in \eqref{eq:bellman-stopping-continuous}. However, because the Bellman equation generally has no closed-form solution, we next complement the finite-threshold characterization with an asymptotic analysis. This allows us to obtain explicit theoretical insight into the high-confidence regime, where the platform requires a large belief threshold \(B\) before stopping.

\subsubsection{Asymptotic Analysis of the High-Threshold Regime}
\label{sec:continuous-asymptotic}

In the high-threshold regime, the relevant object is the growth rate of the public log-likelihood ratio. We therefore first establish that the public belief continues to concentrate on the true state under smooth randomized response, and then characterize the corresponding convergence rate. Recall from \cref{def:AL} that asymptotic learning occurs if the probability of a correct released report converges to one. Under smooth randomized response, we also establish that the public log-likelihood ratio diverges in the direction of the true state almost surely. This public-belief convergence allows the platform to make arbitrarily reliable decisions as \(B\to\infty\). Moreover, the logarithmic learning rate under homogeneous smooth randomized response identifies the exponential scale of the belief-threshold stopping time. Its expectation and leading constant require a separate analysis of the stopped process.

Thus, convergence of public beliefs to the true state ensures that the platform's decision becomes asymptotically correct as the belief threshold increases, while the convergence rate of the public belief determines the corresponding large-\(B\) scaling of the belief-threshold stopping time. The expected calendar stopping time is then obtained by adjusting the expected belief-threshold stopping time by the participation probability \(\eta(\varepsilon)\).

For the smooth randomized response benchmark, we first establish that asymptotic learning is preserved. This result shows that the probability of a correct released report converges to one despite privacy-preserving randomization.

\begin{theorem}[Asymptotic Learning under Smooth Randomized Response]
\label{thm:smoothly-asymptotic-learning-C}
Consider the Gaussian sequential learning model with binary states. Under the smooth randomized response protocol \(\mathcal{M}_{\varepsilon}\) with any fixed privacy parameter \(\varepsilon\in(0,\infty)\), asymptotic learning occurs.
\end{theorem}

\Cref{thm:smoothly-asymptotic-learning-C} shows that smooth randomized response preserves asymptotic learning. We next characterize the rate at which the public belief grows, which determines the large-\(B\) scaling of the belief-threshold stopping time.

Therefore, the smooth randomized response benchmark guarantees the limiting correctness of the platform's final decision as the belief threshold becomes large. Having established asymptotic correctness, we next characterize the speed at which the public belief reaches high confidence, and then analyze the expected belief-threshold stopping time separately. The following theorem gives this rate under smooth randomized response.

\begin{theorem}[Asymptotic Learning Rate under Smooth Randomized Response]
\label{thm:smoothly-learning-speed}
Consider the Gaussian sequential learning model under the smooth randomized response protocol \(\mathcal{M}_{\varepsilon}\) with any fixed privacy parameter \(\varepsilon\in(0,\infty)\).
Conditional on \(\theta=+1\), the public log-likelihood ratio satisfies
\[
l_n
\sim
\frac{2}{\varepsilon\sigma^2}\log n,
\qquad \text{a.s.}
\]
Equivalently, the convergence rate is
\[
f(n)
=
\frac{2}{\varepsilon\sigma^2}\log n
=
\Theta_{\varepsilon}(\log n).
\]
\end{theorem}

The $\Theta_{\varepsilon}(\log(n))$ asymptotic learning rate in \Cref{thm:smoothly-learning-speed} represents an improvement by an order of $\sqrt{\log(n)}$ over the $\Theta(\sqrt{\log(n)})$ asymptotic learning speed established by \citet{Hann-Caruthers2018speed} in the non-private regime. To better understand this improvement, assume that $\theta = +1$. For Gaussian signals, asymptotic learning is guaranteed, and the conditional probability of a submitted report $+1$ tends to one as the public belief grows. To protect their private signals, agents use the smooth randomized response to add noise to their actions. Importantly, this noise is asymmetric: In the true state, where signals are more aligned with the majority, agents require less distortion to maintain privacy; under the false state, where signals are more likely to contradict the majority, agents need to introduce more noise. This asymmetry arises because, when the public belief strongly favors the true state, signals under the false state are on average closer to the decision threshold, where the mechanism applies more randomization. As a result, reports agree with the prevailing belief more often under the true state than under the false state. This difference makes reports informative for later agents, thereby accelerating learning.
\begin{figure}[htbp]
    \centering
    \includegraphics[width=\textwidth]{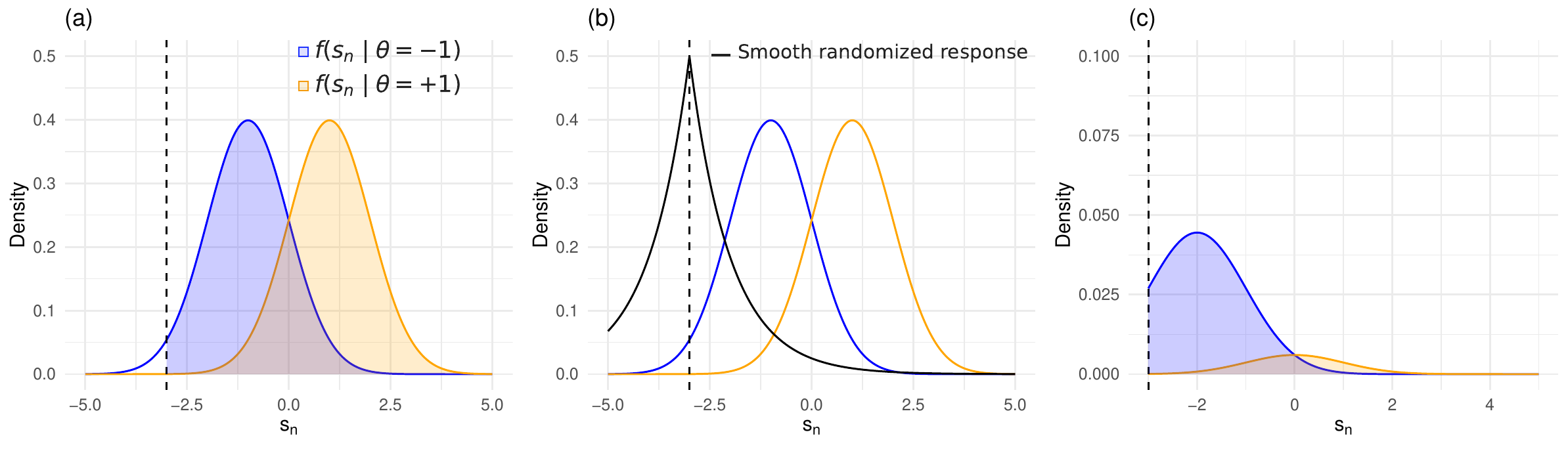}
    \caption{Smooth randomized response and the signal likelihoods.
(a) shows the signal distribution under two different states ($\theta = -1$ in blue and $\theta = +1$ in orange), with the shaded area representing the probability of choosing $a_{n} = +1$ in each case. The vertical dotted line indicates the threshold location. (b) shows the smooth randomized response function as it varies with $s_{n}$ (black). (c) shows the probability of action changing from $+1$ to $-1$ under the influence of smooth randomized response, represented by the shaded area.}
    \label{fig:noise_and_two_normal}
\end{figure}

This intuition is further illustrated in \Cref{fig:noise_and_two_normal}. When the public belief strongly favors $+1$, smooth randomized response reduces the probability of releasing a report $+1$ in both states relative to the intended action, with a more pronounced reduction when $\theta = -1$, due to the shape of the signal distribution and its distance from the decision threshold. Specifically, from the perspective of the next agent, the probability that agent \(n\)'s initial recommendation is flipped by the reporting mechanism is higher under the false state than under the true state. Because agent \(n+1\) does not observe \(s_n\), this flipping probability is inferred by averaging over the possible private signals under each state. This leads to the log-likelihood update:
\begin{equation}
    \begin{aligned}
    l_{n+1} &= l_n + D_{x_n}(l_n), \\
    D_x(l) &= \log \frac{\mathbb{P}(x_n=x \mid l_n=l,\theta=+1)}
    {\mathbb{P}(x_n=x \mid l_n=l,\theta=-1)},
    \qquad x\in\{-1,+1\}.
    \end{aligned}
    \label{eq:log_likelihood_update}
\end{equation}

For a positive report in the high-belief regime, both probabilities in the likelihood ratio decrease due to noise, but the denominator decreases more sharply. Negative reports produce negative updates and must also be included in the analysis. The logarithmic learning rate follows by analyzing the conditional mean increment over both report outcomes and controlling the accumulated martingale fluctuations; it does not require reports to become permanently correct.

It is worth mentioning that the mechanism we propose differs from existing studies. There are two distinct ways to accelerate learning speed. One approach is to modify the tail properties of the signal distribution. \citet{rosenberg2019efficiency} and \citet{Hann-Caruthers2018speed} observe that Gaussian signals lead to slow learning speeds due to their thin tails, which reduce the probability of reversing an incorrect trend. This occurs because the likelihood of receiving a strong enough private signal to counteract an incorrect cascade is very low. In contrast, heavy-tailed distributions allow for rare but highly informative signals that can quickly correct mistaken herding beliefs. Another way to accelerate learning speed is to increase the weight of private signals. \citet{arieli2025hazards} show that overconfidence accelerates learning, as agents place greater reliance on their private signals rather than following public consensus. However, these approaches either depend on fixed signal properties or introduce systematic biases. In our paper, we leverage the asymmetry of the smooth randomized mechanism, which reshapes the informativeness of released reports while maintaining privacy constraints. Unlike a standard randomized response, which introduces noise indiscriminately, our mechanism preserves useful information in decision making, leading to more efficient information aggregation and improved learning speed.

Having established the convergence rate of the public belief under smooth randomized response, we now analyze the platform's expected belief-threshold stopping time. The logarithmic learning rate identifies the exponential scale of the stopping time, while a separate analysis of the stopped process determines its expectation and leading constant. The following theorem characterizes the corresponding expected belief-threshold stopping time.

\begin{theorem}[Expected Stopping Times under Smooth Randomized Response]
\label{thm:smoothly-stopping-time}
Under the same conditions as \Cref{thm:smoothly-learning-speed},
consider the belief-threshold stopping rule
\[
\tau_B
=
\inf\{n\geq 1:|l_n|\geq B\}.
\]
For either state \(\theta\in\{-1,+1\}\), as \(B\to\infty\),
the expected belief-threshold stopping time satisfies
\[
\mathbb{E}[\tau_B\mid\theta]
\sim
\frac{1}{K(\varepsilon,\sigma)}
\exp\left(\frac{\varepsilon\sigma^2}{2}B\right),
\]
where
\[
K(\varepsilon,\sigma)
=
C(\varepsilon)
-
\frac{1}{2}
e^{-\varepsilon+\frac{\varepsilon^2\sigma^2}{2}}
\left(1-e^{-\varepsilon^2\sigma^2}\right)
>0,
\]
and
\[
C(\varepsilon)
=
\frac{\varepsilon\sigma^2}{4}
\left(
e^{\varepsilon+\frac{\varepsilon^2\sigma^2}{2}}
-
e^{-\varepsilon+\frac{\varepsilon^2\sigma^2}{2}}
\right).
\]
Consequently, if each potential agent participates independently
with probability \(\eta(\varepsilon)>0\), independently of
private signals and reporting randomization, then the expected
calendar stopping time satisfies
\[
\mathbb{E}[T_B\mid\theta]
=
\frac{\mathbb{E}[\tau_B\mid\theta]}{\eta(\varepsilon)}
\sim
\frac{1}{\eta(\varepsilon)K(\varepsilon,\sigma)}
\exp\left(\frac{\varepsilon\sigma^2}{2}B\right).
\]
\end{theorem}

The theorem shows that, for fixed \(\varepsilon>0\) and \(\sigma>0\),
the expected number of submitted reports required for the public
log-likelihood ratio to reach either boundary, \(+B\) or \(-B\),
grows exponentially in \(B\). The privacy parameter \(\varepsilon\)
affects this expected stopping time through both the exponential
term \(\exp(\varepsilon\sigma^2B/2)\) and the multiplicative
constant \(1/K(\varepsilon,\sigma)\). The expected calendar
stopping time additionally incorporates participation through
the factor \(1/\eta(\varepsilon)\).

To complement the theoretical stopping-time characterization, we numerically compare three approaches for evaluating belief-threshold and calendar stopping times in the continuous Gaussian signal model: the Bellman equation, Monte Carlo simulation, and the asymptotic approximation in \Cref{thm:smoothly-stopping-time}. 

Figure~\ref{fig:continuous_stopping_time_comparison} reports both quantities as functions of the platform-chosen privacy budget \(\varepsilon\). Panel~\ref{fig:continuous_report_stopping_time} shows the belief-threshold stopping time, while Panel~\ref{fig:continuous_calendar_stopping_time} shows the corresponding calendar stopping time after incorporating participation. In both panels, the Bellman and Monte Carlo curves closely track each other, validating the finite-threshold numerical characterization. The asymptotic approximation captures the overall trend but can deviate when finite-threshold effects are pronounced. The comparison illustrates that the expected stopping time can be minimized at an intermediate privacy budget. When \(\varepsilon\) is too small, reports are heavily randomized and each submitted report carries limited information, slowing belief updating. 
When \(\varepsilon\) is too large, the reporting protocol approaches the non-private sequential-learning benchmark, in which public beliefs aggregate information only very slowly in the Gaussian setting. 
As a result, the platform again requires more submitted reports to reach the same belief threshold. Thus, the numerical results suggest that the platform's optimal privacy budget need not lie at either extreme, but can instead balance report informativeness, participation, and decision speed.

\begin{figure}[htbp]
    \centering

    \begin{subfigure}[t]{0.48\textwidth}
        \centering
        \includegraphics[width=\textwidth]{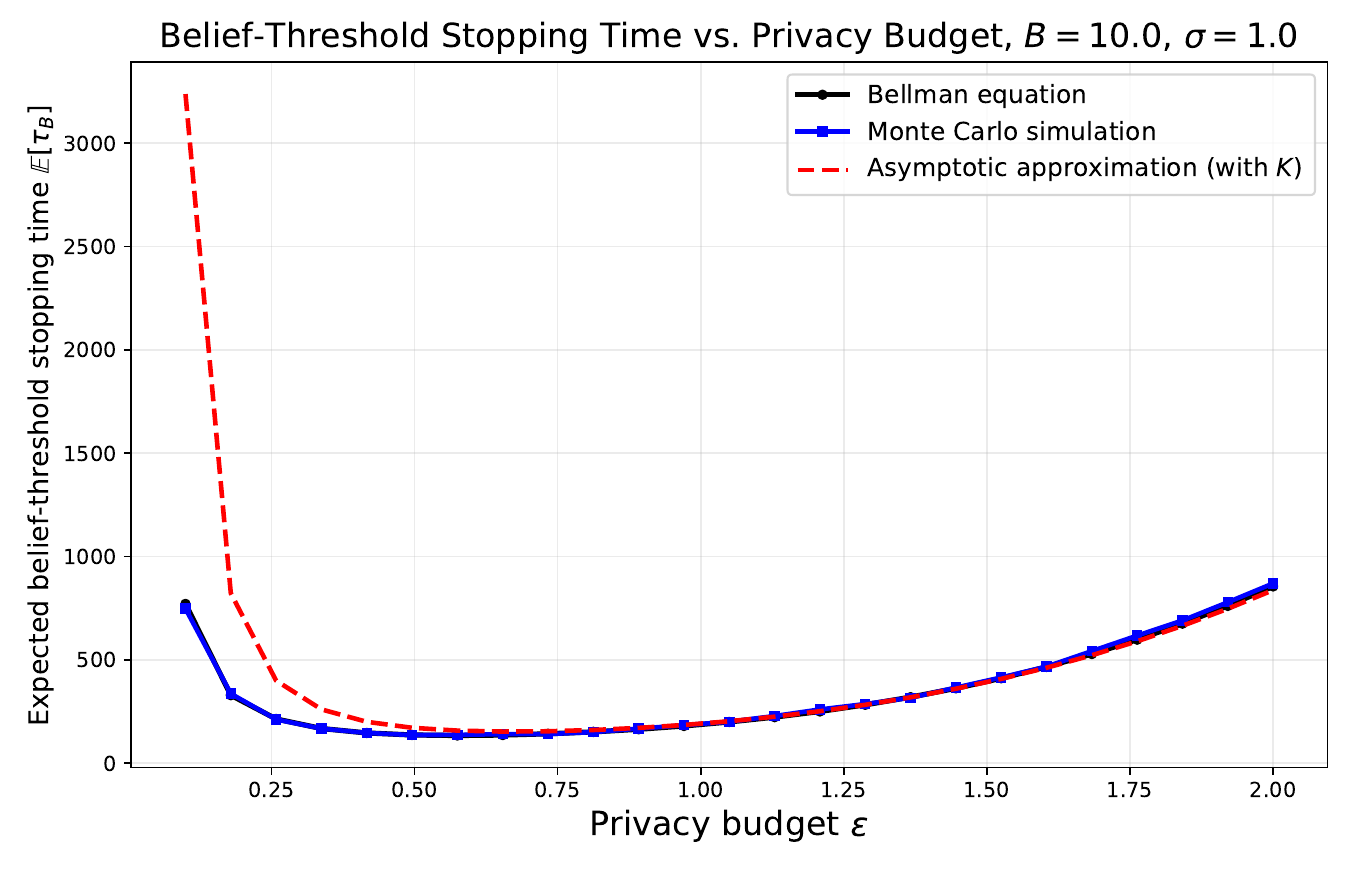}
        \caption{belief-threshold stopping time $(\tau_B)$}
        \label{fig:continuous_report_stopping_time}
    \end{subfigure}
    \hfill
    \begin{subfigure}[t]{0.48\textwidth}
        \centering
        \includegraphics[width=\textwidth]{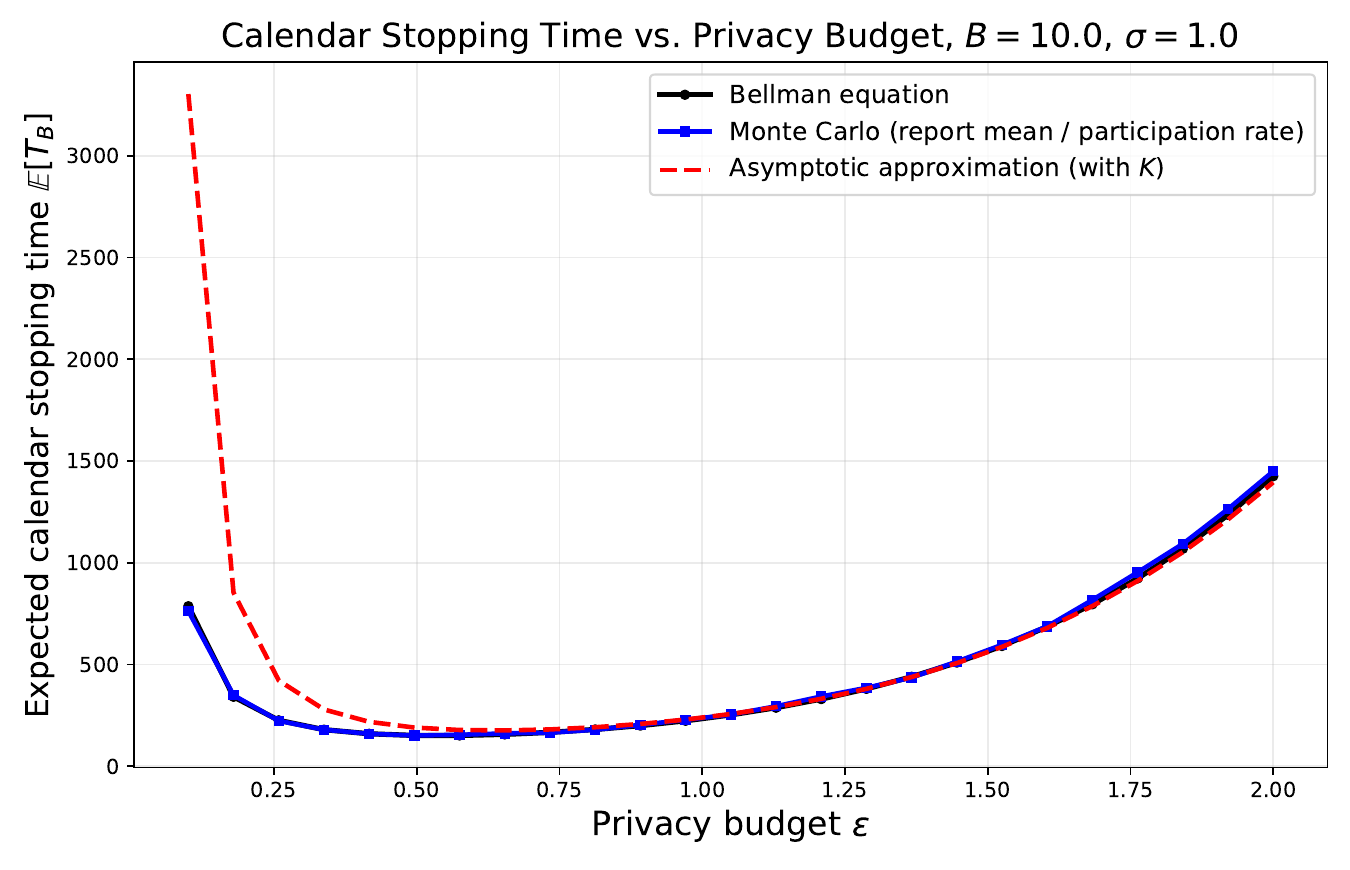}
        \caption{Calendar stopping time $(T_B)$}
        \label{fig:continuous_calendar_stopping_time}
    \end{subfigure}

    \caption{Expected stopping time as a function of the privacy budget \(\varepsilon\) in the continuous Gaussian signal model. Panel (a) reports the expected belief-threshold stopping time \(\mathbb{E}[\tau_B]\), which counts the number of received reports needed to reach the belief threshold $B$. Panel (b) reports the expected calendar stopping time \(\mathbb{E}[T_B]\), which adjusts the belief-threshold stopping time by the participation probability \(\eta(\varepsilon)\). Both panels compare the Bellman equation, Monte Carlo simulation, and the asymptotic approximation. Parameters are \(B=10\) and \(\sigma=1\).}
    \label{fig:continuous_stopping_time_comparison}
\end{figure}

\subsection*{The Continuous Signals Model with Personalized Privacy Mechanisms}
\label{sec:heterogeneous-privacy-continuous}

In the homogeneous setting, the reporting protocol uses a common positive privacy parameter \(\varepsilon\) for all agents. We now extend the model to a personalized privacy mechanism setting, in which the platform adapts the local reporting protocol to agents' heterogeneous privacy tolerances. Specifically, each agent \(n\) has an individual privacy parameter \(\varepsilon_n\), and the platform applies a corresponding local reporting mechanism with privacy parameter \(\varepsilon_n\). Thus, heterogeneity enters through the informativeness of the released reports rather than through an agent-side choice of the privacy mechanism.

To capture heterogeneous privacy tolerances, we focus on the normalized case
\[
\varepsilon_n \sim U[0,1].
\]
The key feature of this distribution is that it has positive density near zero, so the population includes agents who receive very strong privacy protection. For notational simplicity, we present the main results under \(U[0,1]\). The same analysis extends to \(\varepsilon_n\sim U[0,\kappa]\) for any fixed \(\kappa>0\). Additional simulation results for different choices of privacy parameter intervals and their corresponding learning performance are provided in Figure~\ref{fig:heterogeneous_llr}. We first establish that personalized smooth randomized response still preserves asymptotic learning, as shown in Theorem~\ref{thm:smoothly-asymptotic-learning-C-hete}.

\begin{theorem}[Asymptotic Learning under Personalized Privacy Mechanisms]
\label{thm:smoothly-asymptotic-learning-C-hete}
Consider a sequential learning model with binary states, Gaussian private signals, and agent-specific privacy parameters \(\varepsilon_n\) drawn independently from \(U[0,1]\). Under personalized smooth randomized response, asymptotic learning occurs.
\end{theorem}

Theorem~\ref{thm:smoothly-asymptotic-learning-C-hete} shows that asymptotic learning still occurs when privacy protection is personalized across agents. This result confirms that personalized smooth randomized response preserves the long-run accuracy of sequential learning despite the presence of agents who receive very strong privacy protection. In the next theorem, we go beyond asymptotic learning and characterize the speed at which public belief converges. Specifically, we quantify how heterogeneity in personalized privacy parameters affects the convergence rate of the public log-likelihood ratio.

\begin{theorem}[Learning Rate under Personalized Privacy Mechanisms]
\label{thm:smoothly-learning-speed-hete}
Consider the Gaussian sequential learning model with agent-specific privacy parameters \(\varepsilon_n\sim U[0,1]\). Under personalized smooth randomized response, conditional on \(\theta=+1\), the public log-likelihood ratio satisfies
\[
l_n
\sim
\left(
\frac{32}{\sigma^6}n
\right)^{1/4}.
\]
In particular, the convergence rate is
\[
f(n)=\Theta(n^{1/4}).
\]
\end{theorem}

Due to the positive density of privacy parameters near zero, personalized privacy mechanisms can generate polynomial growth of the public belief. Theorem~\ref{thm:smoothly-learning-speed-hete} shows that, under the normalized uniform distribution \(U[0,1]\), the convergence rate is of order \(n^{1/4}\).  This naturally raises the question of how fast learning can be under other distributions of personalized privacy parameters. The following theorem provides a general bound.

\begin{theorem}[Learning Rate Bounds under Personalized Privacy Mechanisms]
\label{thm:smoothly-learning-speed-upperbound}
Consider a sequential learning model with Gaussian private signals and personalized smooth randomized response, where each agent's privacy parameter \(\varepsilon_n\) is drawn independently from a distribution \(G\) supported on \([0,1]\). Suppose that \(G\) assigns zero measure to the singleton \(\{0\}\), i.e., \(G(\{0\})=0\). Then asymptotic learning occurs, and the public log-likelihood ratio grows at most on the order of \(\sqrt n\) and at least on the order of \(\log n\):
\[
l_n=O(\sqrt n),
\qquad
l_n=\Omega(\log n).
\]
\end{theorem}

Theorem~\ref{thm:smoothly-learning-speed-upperbound} highlights a general limitation in personalized privacy-aware sequential learning. The behavior of the privacy-parameter distribution near zero plays an important role in determining the convergence rate. In the high-belief regime, the likelihood-ratio contribution associated with any fixed positive privacy parameter decays rapidly as the public belief becomes stronger, whereas reports generated under smaller positive privacy parameters retain a comparatively larger contribution. As a result, when the distribution places sufficient probability mass near zero, these agents increasingly influence the asymptotic evolution of public belief. To place this result in context, it is useful to compare it with two benchmarks. In the ideal case where the planner directly observes \(n\) i.i.d. Gaussian private signals, the aggregated log-likelihood ratio grows at rate \(\Theta(n)\), corresponding to the classical law of large numbers. In the standard non-private sequential learning setting, where agents observe only the actions of their predecessors, the convergence rate is reduced to \(\Theta(\sqrt{\log n})\) due to the information loss from action-based observation \citep{Hann-Caruthers2018speed,rosenberg2019efficiency}.

Our result shows that personalized smooth randomized response can substantially change the asymptotic rate at which public belief accumulates information. Under the uniform distribution \(U[0,1]\), the convergence rate becomes \(\Theta(n^{1/4})\), which is faster than the logarithmic rate obtained under any fixed positive privacy parameter. The key mechanism is the positive density of small but strictly positive privacy parameters near zero, combined with the state-dependent asymmetry induced by smooth randomized response. These agents become increasingly important in the high-belief regime because their likelihood-ratio contributions decay more slowly than those associated with fixed positive privacy parameters. Nonetheless, Theorem~\ref{thm:smoothly-learning-speed-upperbound} shows that, within the personalized smooth randomized response framework, no distribution of privacy parameters supported on \([0,1]\) can lead to a convergence rate faster than order \(\sqrt n\).

\begin{theorem}[Expected Stopping Times under Personalized Privacy]
\label{thm:personalized-stopping-time}
Consider the Gaussian sequential learning model with
\(s_n\mid\theta\sim\mathcal{N}(\theta,\sigma^2)\),
a symmetric prior, and personalized smooth randomized response.
Suppose the privacy parameters \(\varepsilon_n\) are drawn
independently from \(U[0,1]\), independently of the private
signals and the state, and are not publicly observed.
Public beliefs are updated from the released reports by
integrating over the privacy parameters.

Consider the belief-threshold stopping rule
\[
\tau_B
=
\inf\{n\geq 1:|l_n|\geq B\}.
\]
For either state \(\theta\in\{-1,+1\}\), as \(B\to\infty\),
the expected belief-threshold stopping time satisfies
\[
\mathbb{E}[\tau_B\mid\theta]
\sim
\frac{\sigma^6}{32}B^4.
\]
Under the personalized privacy mechanism, all agents participate
because each agent is assigned a privacy parameter compatible
with her own tolerance. Hence \(T_B=\tau_B\), and
\[
\mathbb{E}[T_B\mid\theta]
=
\mathbb{E}[\tau_B\mid\theta]
\sim
\frac{\sigma^6}{32}B^4.
\]
\end{theorem}

The theorem shows that, under personalized smooth randomized response with \(\varepsilon_n\sim U[0,1]\), the public log-likelihood ratio grows at a polynomial rate in the number of submitted reports. Therefore, the number of reports required to reach a high belief threshold \(B\) grows polynomially in \(B\), rather than exponentially as in the fixed-privacy-parameter case in \cref{thm:smoothly-stopping-time}. Since personalized privacy removes the participation loss by tailoring the reporting protocol to each agent's privacy tolerance, the same polynomial scaling also characterizes the platform's calendar stopping time.

\paragraph{Simulation under Heterogeneous Privacy Budgets.}
While our theory shows that stronger privacy concerns yield faster asymptotic learning (as $n\to\infty$), we ask how learning behaves in finite populations. Following \cite{alaggan2015heterogeneous, acharya2024personalized}, we classify agents by privacy: \emph{conservatives} ($\varepsilon=0.1$), \emph{pragmatists} ($\varepsilon=0.5$), and \emph{liberals} ($\varepsilon=1$). We study five regimes: three homogeneous settings with fixed $\varepsilon\in{0.1,0.5,1}$, a heterogeneous setting with $\varepsilon\sim U(0,1)$, and a non-private baseline. For each regime we run five independent simulations; signals are normal with standard deviation $\sigma=1$. This design ties behavioral types to budgets and tests whether asymptotic predictions already manifest in finite samples.

\begin{figure}[htbp]
    \centering
    \includegraphics[width=0.5\textwidth]{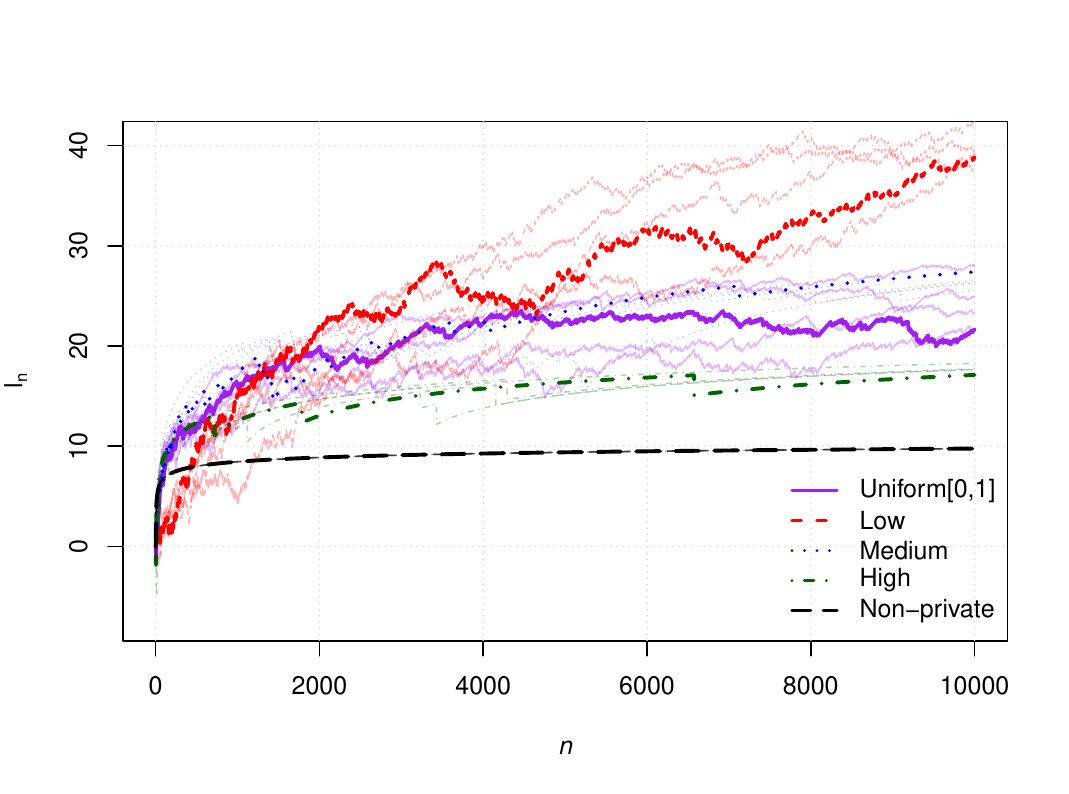}
\caption{Log-likelihood ratio dynamics under five privacy regimes. Homogeneous settings: conservatives $(\varepsilon=0.1)$, pragmatists $(\varepsilon=0.5)$, and liberals $(\varepsilon=1)$; heterogeneous setting: $\varepsilon \sim U(0,1)$; and a non-private baseline.  Signals are normally distributed with $\sigma=1$.}
    \label{fig:heterogeneous_llr}
\end{figure}

Figure~\ref{fig:heterogeneous_llr} illustrates the trajectory of the log-likelihood ratio $l_n$ as the number of agents increases.  We observe that, consistent with the theoretical result, the log-likelihood ratio grows more rapidly under stronger privacy concerns (lower budgets). Intuitively, when privacy concern is stronger, agents introduce more randomization in their actions. When the public belief strongly favors the true state, this randomization creates an asymmetry between report distributions under the two states: reports agree with the prevailing belief more often under the true state than under the false state. This asymmetry amplifies the log-likelihood ratio between states, thereby accelerating information aggregation.

\subsection{Connection to Classical Social-Learning Efficiency Metrics}
\label{sec:classical-learning-efficiency}

The belief-threshold stopping rule studied above is the primary performance measure for the platform. 
The platform stops collecting feedback once the public log-likelihood ratio reaches a prescribed belief threshold \(B\), so the platform's operational objective concerns the time required to reach this threshold.

To connect our analysis with the classical sequential-learning literature, we additionally consider two standard measures of learning efficiency: the time to the first correct submitted report and the total number of incorrect submitted reports. These measures do not represent alternative stopping rules for the platform. Rather, they describe the efficiency of the underlying sequential learning process if feedback collection were allowed to continue, and provide a useful benchmark for understanding how the belief-growth improvements identified above translate into classical notions of learning performance.

Throughout this subsection, we condition on the true state \(\theta=+1\). Hence, a submitted report \(x_n=+1\) is correct, whereas \(x_n=-1\) is incorrect. Define
\[
    \tau
    :=
    \inf\{n\geq 1:x_n=+1\},
\]
as the time to the first correct submitted report, and let
\[
    W
    :=
    \sum_{n=1}^{\infty}
    \mathbf{1}\{x_n=-1\}
\]
denote the total number of incorrect submitted reports.

Under the standard non-private Gaussian sequential-learning benchmark, the public log-likelihood ratio grows only at rate \(\Theta(\sqrt{\log n})\), and both classical efficiency measures have infinite expectations. We next show that homogeneous smooth randomized response makes the expected time to the first correct submitted report finite when privacy protection is sufficiently strong. However, continued randomization leads to infinitely many incorrect submitted reports almost surely if feedback collection continues indefinitely.

\begin{theorem}[Time to First Correct Submitted Report under Homogeneous Privacy]
\label{thm:finite-expectation-homogeneous}
Consider the Gaussian sequential learning model with a fixed privacy parameter \(\varepsilon\in(0,\infty)\). Suppose agents follow homogeneous smooth randomized response. Then there exists a positive constant \(C_1\), which may depend on \(\varepsilon\) and \(\sigma\) but is independent of \(n\), such that
\[
    \mathbb{E}_{+}[\tau]
    \leq
    C_1
    \sum_{n=1}^{\infty}
    n^{-\frac{2}{\varepsilon\sigma^2}}.
\]
Consequently,
\[
    \mathbb{E}_{+}[\tau]<\infty
    \qquad
    \text{whenever }
    \varepsilon<\frac{2}{\sigma^2}.
\]
\end{theorem}

\begin{theorem}[Total Number of Incorrect Submitted Reports under Homogeneous Privacy]
\label{thm:finite-incorrect-actions-homogeneous}
Consider the Gaussian sequential learning model with a fixed privacy parameter \(\varepsilon\in(0,\infty)\). Suppose agents follow homogeneous smooth randomized response. Then
\[
    \mathbb{P}_{+}(W=\infty)=1.
\]
Consequently,
\[
    \mathbb{E}_{+}[W]=\infty
    \qquad
    \text{for every }\varepsilon\in(0,\infty).
\]
\end{theorem}

These results provide a complementary connection between the platform-oriented stopping-time analysis and classical social-learning efficiency. Sufficiently strong homogeneous privacy can make the expected time to the first correct submitted report finite, while the total number of incorrect submitted reports remains almost surely infinite. Faster public-belief growth and shorter high-threshold stopping times therefore coexist with infinitely many incorrect reports in the indefinitely continued reporting process.

At the same time, these measures should be interpreted as long-run learning benchmarks rather than as the platform's operational stopping rule. The platform's decision remains governed by the belief threshold \(B\), whereas \(\tau\) and \(W\) describe properties of the sequential reporting process if it is allowed to continue. The corresponding results under personalized privacy mechanisms are presented in Appendix~\ref{app:personalized-classical-efficiency}.

\vspace{-15pt}
\section{Conclusions}
\label{sec:conclusion}

We investigated the effects of platform-designed privacy-preserving reporting protocols on sequential learning, considering both binary and continuous signal environments under metric differential privacy frameworks. Our analysis shows that privacy protection affects learning not only by reducing the informativeness of individual reports, but also by changing cascade thresholds, participation incentives, stopping decisions, and the speed at which public beliefs accumulate information.

In the binary signal case, we model the sequential feedback process as a Markov chain with absorbing states and characterize the probability that the platform reaches a correct final decision. Under randomized response, a smaller privacy budget introduces more noise into each released report, reducing the informativeness of individual feedback. However, the relationship between the privacy budget and final decision accuracy is not monotonic. Within a fixed cascade-threshold region, increasing the privacy budget makes reports more informative and improves the probability of a correct final decision. At threshold transition points, however, a larger privacy budget can reduce the cascade threshold, causing the platform to stop collecting feedback earlier and potentially lowering decision accuracy. These findings highlight that privacy-preserving reporting protocols affect binary sequential learning through both report informativeness and stopping behavior.

For continuous signals, we study a platform-designed smooth randomized response protocol that randomizes reports according to the distance between an agent's private signal and the current decision threshold. This protocol provides metric differential privacy while avoiding excessive distortion for signals that give strong evidence for one state. We show that asymptotic learning is preserved under this protocol and that the public log-likelihood ratio grows at rate \( \Theta_{\varepsilon}(\log n) \), compared with the classical non-private rate \( \Theta(\sqrt{\log n}) \). This acceleration arises from the asymmetric effect of smooth randomized response across the two states: under the true state, reports are flipped with lower average probability than under the false state, making released reports more informative for public inference. Consequently, the platform can reach high belief thresholds faster than in the standard non-private sequential learning benchmark.

We further analyze the platform's expected belief-threshold stopping time. In the continuous signal model, the platform stops collecting feedback once the absolute public log-likelihood ratio reaches a prescribed belief threshold \(B\). Our stopping-time analysis shows that this expected stopping time grows exponentially in \(B^2\) in the non-private benchmark, but only exponentially in \(B\) under the homogeneous smooth randomized response protocol. Thus, platform-designed privacy-preserving reporting can reduce the time required to reach a high-confidence decision, even though each individual report is privatized.

We also consider a personalized privacy setting in which the platform adapts the reporting protocol to heterogeneous privacy tolerances. When individual privacy parameters satisfy \(\varepsilon_n \sim U[0,1]\), the public log-likelihood ratio grows at rate \(\Theta(n^{1/4})\), and our stopping-time analysis shows that the expected high-threshold stopping time grows polynomially in the belief threshold \(B\). This result shows that heterogeneous privacy guarantees can further accelerate information aggregation relative to both the non-private benchmark and the homogeneous privacy setting. At the same time, sequential learning remains fundamentally slower than the fully transparent benchmark in which the platform directly observes i.i.d. private signals, where belief accumulation would occur at the linear rate \(\Theta(n)\).

Our findings offer practical guidance for organizations that must balance privacy protection, participation, and timely decision making. In settings such as organizational AI adoption, employee feedback systems, technology rollouts, and other sequential platform decisions, privacy-preserving reporting should not be viewed only as a constraint on information collection. Properly designed reporting protocols can broaden participation, reduce premature herding, preserve useful uncertainty, and help the platform reach reliable high-confidence decisions more efficiently. These results suggest that privacy-preserving platform design can support both ethical data collection and effective organizational learning.

\ACKNOWLEDGMENT{%
We thank Juba Ziani, Krishna Dasaratha, David Hirshleifer and Jie Gao for their valuable feedback and discussions. An earlier version of this work appeared as an extended abstract at FORC 2025 \citep{liu2025forc}. The present paper develops a platform-design formulation incorporating participation and stopping costs, and adds finite-threshold and asymptotic stopping-time analyses.
}

\section*{Code Availability}
All code used to generate the figures and reproduce the simulation results in this paper is available  at
\url{https://github.com/YuxinLiu1997/Privacy-Preserving-Sequential-Learning}.

\bibliographystyle{informs2014} 
\bibliography{reference}

@String{Computing = "Computing" }

@String{Computer = "{IEEE} Computer" }

@String{Springer = "Springer-Verlag" }

@article{kim2025behavior,
  title={Behavior-Based Pricing Under Informed Privacy Consent: Unraveling Autonomy Paradox},
  author={Kim, Yunhyoung and Cui, Haitao and Zhu, Yi},
  journal={Marketing Science},
  year={2025},
  publisher={INFORMS}
}

@article{alaggan2015heterogeneous,
  title   = {Heterogeneous Differential Privacy},
  author  = {Alaggan, Mohammad and Gambs, S{\'e}bastien and Kermarrec, Anne-Marie},
  journal = {arXiv preprint arXiv:1504.06998},
  year    = {2015}
}

@article{acharya2024personalized,
  title   = {Personalized Differential Privacy for Ridge Regression},
  author  = {Acharya, Krishna and Boenisch, Franziska and Naidu, Rakshit and Ziani, Juba},
  journal = {Naval Research Logistics},
  note = {forthcoming, arXiv preprint arXiv:2401.17127},  
  year    = {2025}
}

@article{acquisti2005privacy,
  title   = {Privacy and Rationality in Individual Decision Making},
  author  = {Acquisti, Alessandro and Grossklags, Jens},
  journal = {IEEE Security \& Privacy},
  volume  = {3},
  number  = {1},
  pages   = {26--33},
  year    = {2005},
  publisher = {IEEE}
}

@article{casadesus2015competing,
  title={Competing with privacy},
  author={Casadesus-Masanell, Ramon and Hervas-Drane, Andres},
  journal={Management Science},
  volume={61},
  number={1},
  pages={229--246},
  year={2015},
  publisher={INFORMS}
}

@article{adjerid2019choice,
  title={Choice architecture, framing, and cascaded privacy choices},
  author={Adjerid, Idris and Acquisti, Alessandro and Loewenstein, George},
  journal={Management science},
  volume={65},
  number={5},
  pages={2267--2290},
  year={2019},
  publisher={INFORMS}
}

@article{montes2019value,
  title={The value of personal information in online markets with endogenous privacy},
  author={Montes, Rodrigo and Sand-Zantman, Wilfried and Valletti, Tommaso},
  journal={Management Science},
  volume={65},
  number={3},
  pages={1342--1362},
  year={2019},
  publisher={INFORMS}
}

@article{ke2023privacy,
  title={Privacy rights and data security: {GDPR} and personal data markets},
  author={Ke, T Tony and Sudhir, K},
  journal={Management Science},
  volume={69},
  number={8},
  pages={4389--4412},
  year={2023},
  publisher={INFORMS}
}

@inproceedings{he2014blowfish,
  title={Blowfish privacy: Tuning privacy-utility trade-offs using policies},
  author={He, Xi and Machanavajjhala, Ashwin and Ding, Bolin},
  booktitle={Proceedings of the 2014 ACM SIGMOD International Conference on Management of Data},
  pages={1447--1458},
  year={2014}
}

@inproceedings{bun2016concentrated,
  title={Concentrated differential privacy: Simplifications, extensions, and lower bounds},
  author={Bun, Mark and Steinke, Thomas},
  booktitle={Theory of cryptography conference},
  pages={635--658},
  year={2016},
  organization={Springer}
}

@inproceedings{song2017pufferfish,
  title={Pufferfish privacy mechanisms for correlated data},
  author={Song, Shuang and Wang, Yizhen and Chaudhuri, Kamalika},
  booktitle={Proceedings of the 2017 ACM International Conference on Management of Data},
  pages={1291--1306},
  year={2017}
}

@article{tao2023distributed,
  title={A distributed privacy-preserving learning dynamics in general social networks},
  author={Tao, Youming and Chen, Shuzhen and Li, Feng and Yu, Dongxiao and Yu, Jiguo and Sheng, Hao},
  journal={IEEE Transactions on Knowledge and Data Engineering},
  volume={35},
  number={9},
  pages={9547--9561},
  year={2023},
  publisher={IEEE}
}

@article{kifer2014pufferfish,
  title={Pufferfish: A framework for mathematical privacy definitions},
  author={Kifer, Daniel and Machanavajjhala, Ashwin},
  journal={ACM Transactions on Database Systems (TODS)},
  volume={39},
  number={1},
  pages={1--36},
  year={2014},
  publisher={ACM New York, NY, USA}
}

@article{peres2020fragile,
  title={How fragile are information cascades?},
  author={Peres, Yuval and R{\'a}cz, Mikl{\'o}s Z and Sly, Allan and Stuhl, Izabella},
  journal={The Annals of Applied Probability},
  volume={30},
  number={6},
  pages={2796--2814},
  year={2020},
  publisher={JSTOR}
}

@article{rizk2023enforcing,
  title={Enforcing privacy in distributed learning with performance guarantees},
  author={Rizk, Elsa and Vlaski, Stefan and Sayed, Ali H},
  journal={IEEE Transactions on Signal Processing},
  volume={71},
  pages={3385--3398},
  year={2023},
  publisher={IEEE}
}

@article{xu2018query,
  title={Query complexity of Bayesian private learning},
  author={Xu, Kuang},
  journal={Advances in Neural Information Processing Systems},
  volume={31},
  year={2018}
}

@article{smith2013rational,
  title={Rational social learning by random sampling},
  author={Smith, Lones and Sorensen, Peter Norman},
  journal={Available at SSRN 1138095},
  year={2013}
}

@article{tsitsiklis2021private,
  title={Private Sequential Learning},
  author={Tsitsiklis, John N and Xu, Kuang and Xu, Zhi},
  journal={Operations Research},
  volume={69},
  number={5},
  pages={1575--1590},
  year={2021},
  publisher={INFORMS}
}

@article{rosenberg2019efficiency,
  title={On the efficiency of social learning},
  author={Rosenberg, Dinah and Vieille, Nicolas},
  journal={Econometrica},
  volume={87},
  number={6},
  pages={2141--2168},
  year={2019},
  publisher={Wiley Online Library}
}

@book{chamley2004rational,
  title={Rational herds: Economic models of social learning},
  author={Chamley, Christophe},
  year={2004},
  publisher={Cambridge University Press}
}

@article{bohren2016informational,
  title={Informational herding with model misspecification},
  author={Bohren, J Aislinn},
  journal={Journal of Economic Theory},
  volume={163},
  pages={222--247},
  year={2016},
  publisher={Elsevier}
}

@article{frick2023belief,
  title={Belief convergence under misspecified learning: A martingale approach},
  author={Frick, Mira and Iijima, Ryota and Ishii, Yuhta},
  journal={The Review of Economic Studies},
  volume={90},
  number={2},
  pages={781--814},
  year={2023},
  publisher={Oxford University Press US}
}

@article{papachristou2025differentially,
  title={Differentially private distributed estimation and learning},
  author={Papachristou, Marios and Rahimian, M Amin},
  journal={IISE Transactions},
  volume={57},
  number={7},
  pages={756--772},
  year={2025},
  publisher={Taylor \& Francis}
}

@article{bikhchandani2024information,
  title   = {Information Cascades and Social Learning},
  author  = {Bikhchandani, Sushil and Hirshleifer, David and Tamuz, Omer and Welch, Ivo},
  journal = {Journal of Economic Literature},
  volume  = {62},
  number  = {3},
  pages   = {1040--1093},
  year    = {2024}
}

@article{Smith2000pathological,
author = {Smith, Lones and Sørensen, Peter},
title = {Pathological Outcomes of Observational Learning},
journal = {Econometrica},
volume = {68},
number = {2},
pages = {371-398},
year = {2000}
}

@article{arieli2025hazards,
  title   = {The Hazards and Benefits of Condescension in Social Learning},
  author  = {Arieli, Itai and Babichenko, Yakov and M{\"u}ller, Stephan and Pourbabaee, Farzad and Tamuz, Omer},
  journal = {Theoretical Economics},
  volume  = {20},
  number  = {1},
  pages   = {27--56},
  year    = {2025}
}

@article{Wang2021privacy_sociallearning,
  author={Akbay, Abdullah Basar and Wang, Weina and Zhang, Junshan},
  journal={IEEE Journal on Selected Areas in Information Theory}, 
  title={Impact of Social Learning on Privacy-Preserving Data Collection}, 
  year={2021},
  volume={2},
  number={1},
  pages={268-282},
}

@article{Ali2018cost,
title = {Herding with costly information},
journal = {Journal of Economic Theory},
volume = {175},
pages = {713-729},
year = {2018},
author = {S. Nageeb Ali},
}

@article{Le2017noise,
  author={Le, Tho Ngoc and Subramanian, Vijay G. and Berry, Randall A.},
  journal={IEEE Transactions on Signal and Information Processing over Networks}, 
  title={Information Cascades With Noise}, 
  year={2017},
  volume={3},
  number={2},
  pages={239-251},
}

@article{Hann-Caruthers2018speed,
title = {The speed of sequential asymptotic learning},
journal = {Journal of Economic Theory},
volume = {173},
pages = {383-409},
year = {2018},
author = {Wade Hann-Caruthers and Vadim V. Martynov and Omer Tamuz},

}

@article{Acemoglu2011network,
    author = {Acemoglu, Daron and Dahleh, Munther A. and Lobel, Ilan and Ozdaglar, Asuman},
    title = {Bayesian Learning in Social Networks},
    journal = {The Review of Economic Studies},
    volume = {78},
    number = {4},
    pages = {1201-1236},
    year = {2011},
}

@article{Bikhchandani1992cascades,
author = {Bikhchandani, Sushil and Hirshleifer, David and Welch, Ivo},
title = {A Theory of Fads, Fashion, Custom, and Cultural Change as Informational Cascades},
journal = {Journal of Political Economy},
volume = {100},
number = {5},
pages = {992-1026},
year = {1992},
}

@article{Banerjee1992herd,
    author = {Banerjee, Abhijit V.},
    title = "{A Simple Model of Herd Behavior}",
    journal = {The Quarterly Journal of Economics},
    volume = {107},
    number = {3},
    pages = {797-817},
    year = {1992},
}

@article{dwork2014algorithmic,
  title={The algorithmic foundations of differential privacy.},
  author={Dwork, Cynthia and Roth, Aaron},
  journal={Found. Trends Theor. Comput. Sci.},
  volume={9},
  number={3-4},
  pages={211--407},
  year={2014}
}

@inproceedings{fan2019practical,
  title={Practical image obfuscation with provable privacy},
  author={Fan, Liyue},
  booktitle={2019 IEEE International Conference on Multimedia and Expo (ICME)},
  pages={784--789},
  year={2019},
  organization={IEEE}
}

@inproceedings{han2020voice,
  title={Voice-indistinguishability: Protecting voiceprint in privacy-preserving speech data release},
  author={Han, Yaowei and Li, Sheng and Cao, Yang and Ma, Qiang and Yoshikawa, Masatoshi},
  booktitle={2020 IEEE International Conference on Multimedia and Expo (ICME)},
  pages={1--6},
  year={2020},
  organization={IEEE}
}

@inproceedings{andres2013geo,
  title={Geo-indistinguishability: Differential privacy for location-based systems},
  author={Andr{\'e}s, Miguel E and Bordenabe, Nicol{\'a}s E and Chatzikokolakis, Konstantinos and Palamidessi, Catuscia},
  booktitle={Proceedings of the 2013 ACM SIGSAC conference on Computer \& communications security},
  pages={901--914},
  year={2013}
}

@inproceedings{chatzikokolakis2013broadening,
  title={Broadening the scope of differential privacy using metrics},
  author={Chatzikokolakis, Konstantinos and Andr{\'e}s, Miguel E and Bordenabe, Nicol{\'a}s Emilio and Palamidessi, Catuscia},
  booktitle={International Symposium on Privacy Enhancing Technologies},
  pages={82--102},
  year={2013},
  organization={Springer}
}

@inproceedings{liu2025forc,
  author = {Liu, Yuxin and Rahimian, M. Amin},
  title = {Differentially Private Sequential Learning},
  booktitle = {6th Symposium on Foundations of Responsible
               Computing (FORC 2025)},
  series = {Leibniz International Proceedings in Informatics
            (LIPIcs)},
  volume = {329},
  pages = {18:1--18:6},
  year = {2025},
  editor = {Bun, Mark},
  publisher = {Schloss Dagstuhl -- Leibniz-Zentrum
               f{\"u}r Informatik},
  doi = {10.4230/LIPIcs.FORC.2025.18},
  note = {{E}xtended Abstract},
}


\begin{APPENDICES}
\clearpage %
\renewcommand{\theHsection}{appendix.\Alph{section}}
\renewcommand\thefigure{\thesection.\arabic{figure}}
\setcounter{figure}{0}
\setcounter{page}{1}
\renewcommand{\thepage}{A-\arabic{page}}

{\noindent \Large \bf Online Appendix}

\noindent
This appendix provides additional related work, proofs, and extensions of the main results. The proofs cover privacy guarantees, learning rates, stopping times, and classical learning-efficiency measures under homogeneous and personalized privacy. We also clarify the distinction between public-belief learning and report accuracy under standard differential privacy, and extend the analysis to Pufferfish privacy and polynomial-decay reporting mechanisms under metric differential privacy.

\makeatletter
\let\OA@originaladdcontentsline\addcontentsline
\renewcommand{\addcontentsline}[3]{%
  \def\OA@type{#2}%
  \def\OA@section{section}%
  \ifx\OA@type\OA@section
    \OA@originaladdcontentsline{apxtoc}{#2}{#3}%
  \else
    \OA@originaladdcontentsline{#1}{#2}{#3}%
  \fi
}

\medskip

\noindent{\bf Appendix Table of Contents}

\medskip

\makeatletter
\begingroup
  \hypersetup{hidelinks,linktoc=page}
  \setcounter{tocdepth}{1}

  \renewcommand{\@pnumwidth}{3em}
  \renewcommand{\@tocrmarg}{4em}
  \renewcommand*{\l@section}{\@dottedtocline{1}{0em}{3em}}

  \renewcommand*{\numberline}[1]{%
    \hb@xt@\@tempdima{%
      \hyperlink{\Hy@tocdestname}{#1}\hfil
    }%
  }

  \@starttoc{apxtoc}
\endgroup
\makeatother

\clearpage

\section{Additional Related Work}
\label{app:related_work}

\paragraph{Social learning under model misspecification.}
Our research relates to the literature on social learning under model misspecification, which typically emphasizes the adverse consequences of imperfect models on learning outcomes. For example, \citet{frick2023belief} show that even minor misspecifications can substantially hinder belief convergence, while \citet{Le2017noise} demonstrate that exogenous observation errors increase the likelihood of incorrect cascades as noise grows. Similarly, \citet{bohren2016informational} highlight how mistaken assumptions about correlation structures among others' actions can generate persistent inefficiencies. These studies collectively suggest that misspecification and observation errors often erode the reliability of information aggregation in sequential environments. 

In contrast, our model studies privacy-preserving noise as part of a platform-designed reporting protocol rather than as an unintended misspecification or exogenous observation error. The platform commits to a local privacy-preserving mechanism before feedback collection begins, and this mechanism transforms agents' private information into public reports observed by later agents. We show that such structured randomization can have effects that differ sharply from generic observation noise. In the continuous-signal setting, smooth randomized response under metric differential privacy can preserve asymptotic learning and accelerate the growth of the public log-likelihood ratio. This faster belief growth, in turn, reduces the belief-threshold stopping time needed for the platform to reach a decision.

\paragraph{Privacy concern and information disclosure.}  
A central theme in the privacy literature concerns how individuals make disclosure decisions when facing privacy risks or privacy-related incentives. \citet{adjerid2019choice} demonstrate that choice architecture and framing effects can lead to cascaded patterns of privacy disclosure, suggesting that even subtle interventions systematically shift disclosure behaviors. \citet{ke2023privacy} study how regulatory regimes such as GDPR reshape individual incentives and the structure of personal data markets, while \citet{kim2025behavior} highlight the autonomy paradox that arises when people consent to behavior-based pricing, balancing perceived control with potential welfare losses. Similarly, \citet{casadesus2015competing} and \citet{montes2019value} analyze how the value of personal information and endogenous privacy preferences shape disclosure incentives in market settings. These studies establish that privacy concerns fundamentally alter disclosure behavior, but they largely focus on individual decision-making or static market equilibria. 

Our work emphasizes a sequential and collective dimension of privacy-preserving disclosure. In our model, the platform specifies a local reporting protocol that determines how private signals are transformed into public reports. Agents may differ in their privacy tolerances, so the platform's privacy choice affects not only the informativeness of each submitted report but also the set of agents willing to participate. As a result, privacy protection changes the dynamics of social learning through both an information-quality channel and a participation channel. Rather than focusing only on whether individuals disclose information, we study how platform-designed privacy protection affects belief updating, information cascades, and the stopping time required for the platform to make a final high-confidence decision.

\paragraph{Privacy protection and utility tradeoffs.}  
Our work also relates to the broader literature on the trade-off between privacy guarantees and utility, particularly in the context of relaxed definitions of differential privacy. Standard differential privacy, while widely adopted, often suffers from limitations in settings with continuous or unbounded domains, as it treats all neighboring datasets as equally distinct and may require injecting excessive noise, leading to severe utility loss \citep{dwork2014algorithmic, bun2016concentrated}. To address these issues, several relaxed notions have been proposed. \textit{Metric differential privacy} generalizes differential privacy by replacing binary adjacency with distance-based similarity, allowing privacy guarantees to vary smoothly with input proximity and enabling finer privacy--utility tradeoffs in applications such as geo-location, speech, and image data \citep{andres2013geo, han2020voice, fan2019practical}. \textit{Pufferfish privacy} provides another flexible framework that explicitly specifies which secrets should be indistinguishable, thereby avoiding unnecessary noise \citep{kifer2014pufferfish}; variants such as Blowfish further demonstrate improved utility for tasks such as histograms and clustering \citep{he2014blowfish, song2017pufferfish}. 

Building on this line of work, we apply the metric differential privacy framework to sequential learning with continuous signals. Metric differential privacy is particularly suitable in our setting because the informativeness and sensitivity of a private signal depend on its distance from the current decision threshold. Signals close to the threshold are more sensitive because small perturbations may change the induced recommendation, whereas signals far from the threshold provide more robust evidence for one state. The platform-designed smooth randomized response protocol exploits this structure by introducing more randomization near the decision threshold and preserving more information farther away from it. We show that this distance-sensitive privacy design not only balances privacy and utility at the report level, but also changes the aggregate dynamics of public belief formation and the stopping time required for high-confidence platform decisions.


\section{Proof of \texorpdfstring{\Cref{thm:piece-wise-random-response}}{thmpiecewiserandomresponse}: Randomized Response Strategy for Binary Model}
\label{app:proof:thm:piece-wise-random-response}

Before the information cascade, agents take actions entirely based on their private signals. Thus, the probability expressions become:
\begin{equation*}
    \mathbb{P}(a_n = +1 \mid s_n, x_1, \dots, x_{n-1}) = \mathbb{I} \{ s_n = +1 \}, \mbox{ and }
    \mathbb{P}(a_n = -1 \mid s_n, x_1, \dots, x_{n-1}) = \mathbb{I} \{ s_n = -1 \}.
\end{equation*}

Now, using \Cref{eq:DPB}, we can rewrite:
\begin{equation*}
    \begin{aligned}
        \frac{\mathbb{P}(\mathcal{M}(s_n; h_{n-1}) = x_n)}{\mathbb{P}(\mathcal{M}(s_n'; h_{n-1}) = x_n)}
        &= \frac{\mathbb{P}(x_n \mid s_n, x_1, x_2, \dots, x_{n-1})}{\mathbb{P}(x_n \mid s_n', x_1, x_2, \dots, x_{n-1})} \\
        &\leq \max_{\substack{x_n \in \{-1, +1\} \\ s_n,\tilde{s}_n  \in \{-1, +1\}}}  
        \frac{\mathbb{P}(x_n \mid s_n, x_1, x_2, \dots, x_{n-1})}{\mathbb{P}(x_n \mid \tilde{s}_n, x_1, x_2, \dots, x_{n-1})} \\
        &= \max_{\substack{x_n \in \{-1, +1\} \\ s_n,\tilde{s}_n  \in \{-1, +1\}}}  
        \frac{\mathbb{P}_{\mathcal{M}_n}(x_n \mid a_n=+1) \mathbb{I} \{ s_n = +1 \} + \mathbb{P}_{\mathcal{M}_n}(x_n \mid a_n=-1) \mathbb{I} \{ s_n = -1 \}}{\mathbb{P}_{\mathcal{M}_n}(x_n \mid a_n=+1) \mathbb{I} \{ \tilde{s}_n = +1 \} + \mathbb{P}_{\mathcal{M}_n}(x_n \mid a_n=-1) \mathbb{I} \{ \tilde{s}_n = -1 \}} \\
        &= \frac{1 - u_n}{u_n} = \exp(\varepsilon),
    \end{aligned}
\end{equation*}

After the information cascade, agents base their decisions solely on public history rather than their private signals. As a result, randomized response mechanisms for protecting private signals become unnecessary because private signals are no longer used in decision making. Finally, given the binary action space, the randomized response mechanism introduces the minimal perturbation required to satisfy the privacy constraint. As such, it is an optimal strategy among all mechanisms that preserve privacy.

\endproof

\section{Proof of \texorpdfstring{\Cref{thm:right-cascade-B}}{thmrightcascadeB}: Probability of the Correct Cascade}\label{app:proof:thm:right-cascade-B}
Building on our observations in \cref{app:proof:thm:piece-wise-random-response}, we analyze the threshold of the information cascade. Here, we define $k$ as the difference between the occurrences of action $+1$ and action $-1$ in the observation history up to the agent \(n\). With this in mind, we establish the following lemma.

\begin{lemma}\label{lem:expression-k}
    If the information cascade starts, the threshold $k$ is given by:
\begin{equation}\label{eq:threshold-k}
  k=\lfloor \log_{ \frac{(1-u(\varepsilon))(1-p) + u(\varepsilon)p }{u(\varepsilon)(1-p) + p(1-u(\varepsilon))} }\frac{1-p}{p} \rfloor+1.
\end{equation}
\end{lemma}

\proof 
First, since the prior of the state is symmetric, the log-likelihood ratio of the public belief for agent $n$ based on the reports from agent $1$ to agent $n-1$ is defined as follows:
\begin{equation*}
    l_{n}= \log\frac{\mathbb{P}(\theta=+1| x_1,\ldots,x_{n-1})}{\mathbb{P}(\theta=-1| x_1,\ldots,x_{n-1})}=\log\frac{\mathbb{P}(x_{1}, \ldots, x_{n-1} |\theta=+1)}{\mathbb{P}(x_{1}, \ldots, x_{n-1}  |\theta=-1)}.
\end{equation*}

We begin the analysis by scrutinizing the action of agent 3. Suppose that before agent 3 takes action, the information cascade has not yet begun, and we examine the log-likelihood ratio of public belief for agent 3. Specifically, if $x_1=x_2=+1$, then,

\begin{equation*}
    l_{3} = \log\frac{\mathbb{P}(x_{1}=+1|\theta=+1)}{\mathbb{P}(x_{1}=+1|\theta=-1)}\frac{\mathbb{P}(x_{2}=+1|\theta=+1, x_{1}=+1)}{\mathbb{P}(x_{2}=+1|\theta=-1, x_{1}=+1)}=
\log \left[\frac{u(\varepsilon)(1-p) + p(1-u(\varepsilon))}{(1-u(\varepsilon))(1-p) + u(\varepsilon)p} \right]^{2}. 
\end{equation*}
If $x_1=x_2=-1$, then
\begin{equation*}
l_{3} = \log \frac{\mathbb{P}(x_{1}=-1|\theta=+1)}{\mathbb{P}(x_{1}=-1|\theta=-1)}\frac{\mathbb{P}(x_{2}=-1|\theta=+1, x_{1}=-1)}{\mathbb{P}(x_{2}=-1|\theta=+1, x_{1}=-1)} =\log \left[\frac{(1-u(\varepsilon))(1-p) + u(\varepsilon)p}{u(\varepsilon)(1-p) + p(1-u(\varepsilon))} \right]^{2}.
\end{equation*}
If $x_1=+1, x_2=-1$ or $x_1=-1, x_2=+1$, then $l_{3}=1$. For agent 3, in making decisions, it needs to compare the public log-likelihood ratio with the private log-likelihood ratio $L_3$.

\begin{equation*}
    L_{3}=\log{\frac{\mathbb{P}(\theta=+1|s_{3})}{\mathbb{P}(\theta=-1|s_{3})}}=\log{\frac{\mathbb{P}(s_{3}|\theta=+1)}{\mathbb{P}(s_{3}|\theta=-1)}}.
\end{equation*}
 If \(s_3=+1\), then \(L_{3}=\log\frac{p}{1-p}\); otherwise, \(L_{3}=\log\frac{1-p}{p}\). In equilibrium, agent 3 selects \(a_{3}=+1\) if and only if \(l_{3}+L_{3}>0\). If an information cascade occurs, agent 3 will choose action $-1$ after observing \(x_1=x_2=-1\) even if \(s_{3}=+1\), or agent 3 will choose action +1 after observing \(x_1=x_2=+1\) even if \(s_{3}=-1\). This implies the following expression.

\begin{equation*}
   \log \left[\frac{(1-u(\varepsilon))(1-p) + u(\varepsilon)p}{u(\varepsilon)(1-p) + p(1-u(\varepsilon))} \right]^{2} + \log \frac{p}{1-p} < 0, \log \left[\frac{u(\varepsilon)(1-p) + p(1-u(\varepsilon))}{(1-u(\varepsilon))(1-p) + u(\varepsilon)p} \right]^{2}+\log\frac{1-p}{p} > 0.
\end{equation*}
For agent $n$, by induction we find that if the information cascade occurs, then

\begin{equation*}
   \log \left[\frac{(1-u(\varepsilon))(1-p) + u(\varepsilon)p}{u(\varepsilon)(1-p) + p(1-u(\varepsilon))} \right]^{k} + \log \frac{p}{1-p} < 0, \mbox{ and } \log \left[\frac{u(\varepsilon)(1-p) + p(1-u(\varepsilon))}{(1-u(\varepsilon))(1-p) + u(\varepsilon)p} \right]^{k}+\log\frac{1-p}{p} > 0.
\end{equation*}

In total, if $k$ is the threshold for the start of the information cascade, then $  k=\lfloor \log_{ \frac{(1-u(\varepsilon))(1-p) + u(\varepsilon)p }{u(\varepsilon)(1-p) + p(1-u(\varepsilon))} }\frac{1-p}{p} \rfloor+1$.

\endproof

\begin{lemma}\label{lem:asymptotic-learning-B}
When $\tilde{u}(\varepsilon) \neq \frac{1}{2}$, the probability of a correct cascade is given by:
\begin{equation*}
\mathbb{P}(\text{correct cascade}) = \frac{\rho(\varepsilon)^k - 1}{\rho(\varepsilon)^{2k} - 1}.
\end{equation*}
where $\rho(\varepsilon) = \frac{1 - \tilde{u}(\varepsilon)}{\tilde{u}(\varepsilon)}$ and $\tilde{u}(\varepsilon) = u(\varepsilon)(1-p) + p(1-u(\varepsilon))$. When $\tilde{u}(\varepsilon) = \frac{1}{2}$, the information cascade does not occur.
\end{lemma}

\proof 
For the case $\tilde{u}(\varepsilon) =\frac{1}{2}$, it is trivial. However, for $\tilde{u}(\varepsilon) \neq \frac{1}{2}$, we can reframe the problem as a Markov chain featuring two absorbing states: $k$ and $-k$. The transition probabilities are defined as follows: $\mathbb{P}(z_{n}=m+1|z_{n-1}=m)=\tilde{u}(\varepsilon)$ and $\mathbb{P}(z_{n}=m-1|z_{n-1}=m)=1-\tilde{u}(\varepsilon)$, where $z_{n-1}$ is the difference between the occurrences of action $+1$ and action $-1$ in the observation history up to agent $n$. This formulation aligns with the classical gambler's ruin problem.

Given an initial capital of $k$ defined in \Cref{eq:threshold-k}, our objective is to calculate the probabilities associated with losing all capital $k$ or achieving a goal of $2k$. Without loss of generality, we assume $\theta = +1$, then the probability of achieving a goal of $2k$ is the same as the probability of the correct cascade. Let, 
\begin{equation*}
S_n=X_1+\cdots+X_n, \quad S_0=0 .
\end{equation*}
The event
\begin{equation*}
M_{k, n}=\left[k+S_n=2k\right] \cap \bigcap_{m=1}^{n-1}\left[0<k+S_m<2k\right],
\end{equation*}
represents earning $2k$ at time $n$. If $s_{2k}(k)$ denotes the probability of achieving this goal ultimately, then
\begin{equation*}
s_{2k}(k)=P\left(\bigcup_{n=1}^{\infty} M_{k, n}\right)=\sum_{n=1}^{\infty} P\left(M_{k, n}\right).
\end{equation*}
By convention, given the sample space $\Omega$ and the empty set $\emptyset$, we set $M_{k,0}=\emptyset$, $M_{2k,0}=\Omega$, and $M_{2k,n}=\emptyset$. Then we get $s_{2k}(0)=0$ and $s_{2k}(2k)=1$. Based on the problem setting, we could derive the following difference equation,
\begin{equation*}
s_{2k}(k)=\tilde{u}(\varepsilon) s_{2k}(k+1)+(1-\tilde{u}(\varepsilon))s_{2k}(k-1).
\end{equation*}
By solving the above difference equation with side conditions $s_{2k}(0)=0$ and $s_{2k}(2k)=1$. We get the probability that the gambler can successfully achieve his goal before ruin as follows:
\begin{equation*}
s_{2k}(k)=\frac{\rho(\varepsilon)^k-1}{\rho(\varepsilon)^{2k}-1}.
\end{equation*}
where $\rho(\varepsilon)=\frac{1-\tilde{u}(\varepsilon)}{\tilde{u}(\varepsilon)}$ and $\tilde{u}(\varepsilon)=u(\varepsilon)(1-p) + p(1-u(\varepsilon))$.
\endproof

Based on the above lemmas, we now present the formal proof of Theorem \ref{thm:right-cascade-B}. First, we derive the interval $I_{k}$ in which the threshold is $k$. By \Cref{lem:expression-k}, after the simple calculation, for any $k \geq 2$
\begin{equation*}
        v_{k}=\frac{1-\alpha^{\frac{k-2}{k-1}}}{1-\alpha^{\frac{k-2}{k-1}}+\alpha^{\frac{-1}{k-1}}-\alpha},
\end{equation*}
where $\alpha=(1-p)/p$. Thus, $I_{k} = (\varepsilon_{k+1}, \varepsilon_{k}]$, and $\varepsilon_{k}=\log\frac{1-v_{k}}{v_{k}}$. By lemma \ref{lem:asymptotic-learning-B} take the first derivative of $s_{2k}(k)$ with respect to $\rho(\varepsilon)$ for $k \geq 2$ and $\rho(\varepsilon) \in (0,1)$:

\begin{equation*}
    \frac{ds_{2k}(k)}{d\rho}=\frac{3k\rho(\varepsilon)^{3k-2}-k\rho(\varepsilon)^{k-1}-2k\rho(\varepsilon)^{2k-1}}{(\rho(\varepsilon)^{2k}-1)^2}<0.
\end{equation*}
Next, we take the first derivative of $\rho(\varepsilon)$ and $u(\varepsilon)$ with respect to $u$ and $\varepsilon$, respectively. Then $\frac{d\rho(\varepsilon)}{du(\varepsilon)}>0$ and $\frac{du(\varepsilon)}{d\varepsilon}<0$. Consequently. $\frac{ds_{2k}(k)}{d\varepsilon}>0$, With an increase in the differential privacy budget, the probability of a correct cascade also increases.
\endproof

\section{Proof of Theorem~\ref{thm:stopping-time-binary}: Stopping Time under Participation}
\label{app:proof-stopping-time-binary}

\proof{Proof of Theorem~\ref{thm:stopping-time-binary}.}
Without loss of generality, condition on the true state \(\theta=+1\). Conditional on participation and before an information cascade occurs, each submitted privacy-protected report \(x_m\in\{-1,+1\}\) follows the binary randomized response protocol. Hence,
\[
\mathbb{P}(x_m=+1)
=
p(1-u(\varepsilon))+(1-p)u(\varepsilon),
\]
and
\[
\mathbb{P}(x_m=-1)
=
(1-p)(1-u(\varepsilon))+p u(\varepsilon).
\]
The drift of the submitted-report process is therefore
\[
\mu(\varepsilon,p)
=
\mathbb{P}(x_m=+1)-\mathbb{P}(x_m=-1)
=
(2p-1)(1-2u(\varepsilon)).
\]
Since \(p>1/2\) and \(u(\varepsilon)<1/2\) for \(\varepsilon>0\), we have \(\mu(\varepsilon,p)>0\).

Let
\[
S_m=\sum_{\ell=1}^m x_\ell.
\]
The report-level stopping time is the first time this biased random walk reaches either boundary:
\[
\tau_{\varepsilon,p}
=
\inf\left\{
m\geq 1:
|S_m|\geq k(\varepsilon,p)
\right\}.
\]
Thus, \(\tau_{\varepsilon,p}\) is the first hitting time of the interval boundaries
\[
-k(\varepsilon,p)
\quad \text{and} \quad
k(\varepsilon,p).
\]

Define
\[
q_{\varepsilon,p}
=
p(1-u(\varepsilon))+(1-p)u(\varepsilon),
\]
so that \(1-q_{\varepsilon,p}=(1-p)(1-u(\varepsilon))+pu(\varepsilon)\). Also define
\[
\rho(\varepsilon,p)
=
\frac{1-q_{\varepsilon,p}}{q_{\varepsilon,p}}
=
\frac{(1-p)(1-u(\varepsilon))+p u(\varepsilon)}
{p(1-u(\varepsilon))+(1-p)u(\varepsilon)}.
\]
For a biased random walk starting from \(S_0=0\), with upward probability \(q_{\varepsilon,p}>1/2\), downward probability \(1-q_{\varepsilon,p}\), drift \(\mu(\varepsilon,p)=2q_{\varepsilon,p}-1\), and absorbing boundaries \(\{-k(\varepsilon,p),k(\varepsilon,p)\}\), the expected first hitting time is
\[
\mathbb{E}[\tau_{\varepsilon,p}]
=
\frac{k(\varepsilon,p)}{\mu(\varepsilon,p)}
\cdot
\frac{1-\rho(\varepsilon,p)^{k(\varepsilon,p)}}
{1+\rho(\varepsilon,p)^{k(\varepsilon,p)}}.
\]
This gives the expression for the expected number of submitted reports required to reach the cascade threshold.

It remains to relate the report-level stopping time to the calendar stopping time. Each potential agent participates independently with probability
\[
\eta(\varepsilon)=\mathbb{P}(\varepsilon_n\geq \varepsilon).
\]
Therefore, the number of potential arrivals required to obtain one submitted report has mean \(1/\eta(\varepsilon)\). Since \(\tau_{\varepsilon,p}\) submitted reports are required before the cascade threshold is reached, the expected calendar stopping time is
\[
\mathbb{E}[T_{\varepsilon,p}]
=
\frac{\mathbb{E}[\tau_{\varepsilon,p}]}{\eta(\varepsilon)}.
\]
Substituting the expression for \(\mathbb{E}[\tau_{\varepsilon,p}]\) yields
\[
\mathbb{E}[T_{\varepsilon,p}]
=
\frac{1}{\eta(\varepsilon)}
\cdot
\frac{k(\varepsilon,p)}{\mu(\varepsilon,p)}
\cdot
\frac{1-\rho(\varepsilon,p)^{k(\varepsilon,p)}}
{1+\rho(\varepsilon,p)^{k(\varepsilon,p)}}.
\]
This completes the proof.
\endproof

\section{Failure of Asymptotic Learning under Uniform Randomized Response}
\label{app:failure-asymptotic}

In this appendix, we show that uniform randomized response
does not achieve asymptotic learning as defined in
\Cref{def:AL}. For Gaussian signals and a fixed positive
privacy parameter, public beliefs learn the true state,
but the probability of a correct released report does not
converge to one.

\begin{theorem}[Randomized Response Strategy for Continuous Model]
\label{app:thm:random-response-C}
In the sequential learning model with Gaussian signals,
consider randomized response applied to the initial
recommendation with fixed flip probability
\[
u_n=u(\varepsilon)=\frac{1}{1+e^\varepsilon},
\qquad \varepsilon\in(0,\infty).
\]
This reporting mechanism satisfies
\(\varepsilon\)-local differential privacy.
\end{theorem}

\proof
The initial recommendation is \(a_{n+1}=+1\) if and only if
\[
l_n+\frac{2s_{n+1}}{\sigma^2}>0,
\qquad\text{or equivalently}\qquad
s_{n+1}>-\frac{\sigma^2l_n}{2}.
\]
For either output, the conditional report probability
given the private signal is either \(u(\varepsilon)\)
or \(1-u(\varepsilon)\). Therefore its likelihood ratio
between any two signals is at most
\[
\frac{1-u(\varepsilon)}{u(\varepsilon)}=e^\varepsilon.
\]
This verifies the privacy guarantee for the specified
mechanism; the privacy constraint alone does not uniquely
prescribe a reporting mechanism.
\endproof

Because the flipping probability remains positive,
uniform randomized response introduces persistent errors
into released reports. The following theorem shows that
these errors prevent asymptotic learning as defined in
\Cref{def:AL}, even though public beliefs converge to
the true state.

\begin{theorem}[Failure of Asymptotic Learning]
\label{app:thm:no-asymptotic-learning-C}
Under uniform randomized response with
\(\varepsilon\in(0,\infty)\), Gaussian signals, and a
nondegenerate prior, the probability of a correct
released report satisfies
\[
\lim_{n\to\infty}\mathbb{P}(x_n=\theta)
=
1-u(\varepsilon)<1.
\]
Hence, asymptotic learning as defined in \Cref{def:AL}
does not occur.

Nevertheless, the public log-likelihood ratio satisfies
\[
l_n\longrightarrow+\infty
\quad\text{under }\theta=+1,
\qquad
l_n\longrightarrow-\infty
\quad\text{under }\theta=-1,
\qquad\text{a.s.}
\]
If the reporting process continues indefinitely,
infinitely many incorrect reports are released
almost surely.
\end{theorem}

\proof
Index \(l_n\) by the number of released reports.
Let \(\Phi\) be the standard normal distribution function
and put \(u=u(\varepsilon)\). The probability of a
positive report is
\[
p_\theta(l)
=
u+(1-2u)\Phi\left(
\frac{\theta+\sigma^2l/2}{\sigma}
\right).
\]
For every finite \(l\), these probabilities are continuous
and satisfy
\[
0<p_{-1}(l)<p_{+1}(l)<1.
\]
Under \(\theta=-1\), the likelihood ratio
\(R_n=\exp(l_n-l_0)\) is a nonnegative martingale with
mean one, so it converges to a finite limit.
On the event \(R_\infty>0\), \(l_n\) would converge
to a finite value. However, at any finite belief
the two possible updates
\[
D_+(l)
=
\log\frac{p_{+1}(l)}{p_{-1}(l)}>0,
\qquad
D_-(l)
=
\log\frac{1-p_{+1}(l)}{1-p_{-1}(l)}<0
\]
are both bounded away from zero in absolute value
on a sufficiently small neighborhood of that belief.
Thus the increments could not tend to zero.
This contradiction shows that \(R_\infty=0\),
hence \(l_n\to-\infty\).
Applying the same argument to
\(\exp(-(l_n-l_0))\) under \(\theta=+1\)
proves the other limit.

Substituting these limits into \(p_\theta(l)\), we obtain
\[
\mathbb{P}(x_{n+1}=\theta\mid l_n,\theta)
\longrightarrow
1-u(\varepsilon)
\qquad\text{a.s.}
\]
Since these conditional probabilities are bounded by one,
the dominated convergence theorem yields
\[
\lim_{n\to\infty}
\mathbb{P}(x_{n+1}=\theta\mid\theta)
=
1-u(\varepsilon)
\]
under either state. Averaging over the prior therefore gives
\[
\lim_{n\to\infty}\mathbb{P}(x_n=\theta)
=
1-u(\varepsilon)<1.
\]
Thus, asymptotic learning as defined in \Cref{def:AL}
does not occur, despite convergence of public beliefs
to the true state.

Moreover, at every history, the conditional probability
of an incorrect report is at least \(u>0\).
These conditional error probabilities have a divergent
sum. The conditional Borel--Cantelli lemma therefore
implies that infinitely many incorrect reports are
released almost surely.
\endproof

\section{Proof of \texorpdfstring{\Cref{thm:smoothly-random-response-C}: Smooth Randomized Response Strategy}{thmsmoothlyrandomresponseC}}\label{app:proof:thm:smoothly-random-response-C}
To obtain an upper bound for the smooth randomized response strategy, we analyze the evolution of the public log-likelihood ratio. Without privacy concerns, agent \( n \) chooses \( a_n = +1 \) if and only if
\begin{equation*}
    l_n + L_n > 0.
\end{equation*}

Thus, there exists a threshold \( t(l_n) \) such that the agent chooses \( a_n = +1 \) if \( s_n > t(l_n) \), and \( a_n = -1 \) otherwise. 

Suppose an agent receives a signal \( s_n > t(l_n) \), and another signal \( s_n' > t(l_n) \) is also above the threshold. Then, by the smooth randomized response mechanism, we have
\begin{equation*}
    \frac{\mathbb{P}(\mathcal{M}(s_n; h_{n-1}) = x_n)}{\mathbb{P}(\mathcal{M}(s_n'; h_{n-1}) = x_n)} \leq e^{\varepsilon d_{\mathbb{R}}(s_n, s_n')}.
\end{equation*}

Now consider the case where \( s_n > t(l_n) \) but \( s_n' < t(l_n) \). Then,
\begin{equation*}
    \frac{\mathbb{P}(\mathcal{M}(s_n; h_{n-1}) = x_n)}{\mathbb{P}(\mathcal{M}(s_n'; h_{n-1}) = x_n)}
    = \max_{\substack{x_n \in \{-1, +1\} \\ s_n > t(l_n),\; s_n' < t(l_n)}} 
    \frac{\mathbb{P}(x_n \mid s_n, h_{n-1})}{\mathbb{P}(x_n \mid s_n', h_{n-1})}.
\end{equation*}

If \( x_n = +1 \), then
\begin{align*}
    \frac{\mathbb{P}(x_n = +1 \mid s_n, h_{n-1})}{\mathbb{P}(x_n = +1 \mid s_n', h_{n-1})} 
    &= \frac{1 - \frac{1}{2} e^{-\varepsilon(s_n - t(l_n))}}{\frac{1}{2} e^{\varepsilon(s_n' - t(l_n))}} \\
    &= e^{\varepsilon(s_n - s_n')} \left( 2 e^{\varepsilon t(l_n) - \varepsilon s_n} - e^{2\varepsilon t(l_n) - 2\varepsilon s_n} \right) \\
    &\leq e^{\varepsilon(s_n - s_n')}.
\end{align*}

If \( x_n = -1 \), then
\begin{align*}
    \frac{\mathbb{P}(x_n = -1 \mid s_n, h_{n-1})}{\mathbb{P}(x_n = -1 \mid s_n', h_{n-1})} 
    &= \frac{\frac{1}{2} e^{-\varepsilon(s_n' - t(l_n))}}{1 - \frac{1}{2} e^{\varepsilon(s_n - t(l_n))}} \\
    &= e^{\varepsilon(s_n - s_n')} \cdot \frac{e^{-2\varepsilon s_n + \varepsilon t(l_n) + \varepsilon s_n'}}{2 - e^{\varepsilon(s_n' + t(l_n))}} \\
    &\leq e^{\varepsilon(s_n - s_n')}.
\end{align*}

Now suppose \( s_n < t(l_n) \) and \( s_n' < t(l_n) \). Then, it is straightforward that
\begin{equation*}
    \frac{\mathbb{P}(\mathcal{M}(s_n; h_{n-1}) = x_n)}{\mathbb{P}(\mathcal{M}(s_n'; h_{n-1}) = x_n)} \leq e^{\varepsilon d_{\mathbb{R}}(s_n, s_n')}.
\end{equation*}

If \( x_n = +1 \), we have
\begin{equation*}
    \frac{\mathbb{P}(x_n = +1 \mid s_n, h_{n-1})}{\mathbb{P}(x_n = +1 \mid s_n', h_{n-1})} 
    = \frac{\frac{1}{2} e^{\varepsilon(s_n' - t(l_n))}}{1 - \frac{1}{2} e^{-\varepsilon(s_n - t(l_n))}} 
    \leq e^{\varepsilon d_{\mathbb{R}}(s_n, s_n')}.
\end{equation*}

If \( x_n = -1 \), we similarly have
\begin{equation*}
    \frac{\mathbb{P}(x_n = -1 \mid s_n, h_{n-1})}{\mathbb{P}(x_n = -1 \mid s_n', h_{n-1})} 
    \leq e^{\varepsilon(s_n - s_n')}.
\end{equation*}

This completes the derivation, showing that the output distribution of the mechanism satisfies the upper bound
\[
\frac{\mathbb{P}(\mathcal{M}(s_n; h_{n-1}) = x_n)}{\mathbb{P}(\mathcal{M}(s_n'; h_{n-1}) = x_n)} \leq e^{\varepsilon d_{\mathbb{R}}(s_n, s_n')}
\]
in all possible cases. 

\endproof

\section{Proof of Theorem \ref{thm:smoothly-asymptotic-learning-C}: Smoothly Asymptotic Learning}\label{app:proof:thm:smoothly-asymptotic-learning-C}

\proof
Index \(l_n\) by the number of submitted reports, and let \(\mathcal F_n=\sigma(x_1,\ldots,x_n)\). Write
\[
q_\theta(l)=\mathbb P(x_{n+1}=-1\mid l_n=l,\theta),
\qquad
D_+(l)=\log\frac{1-q_{+1}(l)}{1-q_{-1}(l)},
\qquad
D_-(l)=\log\frac{q_{+1}(l)}{q_{-1}(l)}.
\]
The actual update is \(l_{n+1}=l_n+D_{x_{n+1}}(l_n)\), so both reports enter the public-belief process.

Let \(z=\sigma^2l/2\), and let \(\phi_\theta\) denote the density of \(N(\theta,\sigma^2)\). Under smooth randomized response,
\[
\begin{aligned}
q_\theta(l)
={}&\int_{-\infty}^{-z}\phi_\theta(s)
       \left(1-\frac12e^{\varepsilon(s+z)}\right)ds\\
&+\frac12\int_{-z}^{\infty}\phi_\theta(s)e^{-\varepsilon(s+z)}ds.
\end{aligned}
\]
The conditional probability of a negative report is a continuous, strictly decreasing function of the private signal. Coupling the signals as \(S_{+1}=S_{-1}+2\) therefore gives
\begin{equation}\label{eq:prob_min_prob_plus}
0<q_{+1}(l)<q_{-1}(l)<1
\qquad\text{for every finite }l.
\end{equation}
These report probabilities are continuous in \(l\).

Without loss of generality, suppose \(\theta=-1\). The likelihood ratio \(R_n=\exp(l_n-l_0)\) satisfies
\[
\mathbb E_{-1}[R_{n+1}\mid\mathcal F_n]
=R_n\sum_{x\in\{-1,+1\}}
\frac{\mathbb P_{+1}(x_{n+1}=x\mid\mathcal F_n)}
     {\mathbb P_{-1}(x_{n+1}=x\mid\mathcal F_n)}
\mathbb P_{-1}(x_{n+1}=x\mid\mathcal F_n)
=R_n.
\]
Hence \(R_n\) converges almost surely to a finite nonnegative limit. If this limit were positive, \(l_n\) would converge to a finite value. By \eqref{eq:prob_min_prob_plus}, both \(D_+\) and \(D_-\) have nonzero limits there, contradicting convergence of \(l_n\). Thus \(R_n\to0\) and \(l_n\to-\infty\) almost surely. Under \(\theta=+1\), the same argument applied to the reciprocal likelihood ratio yields \(l_n\to+\infty\).

Finally, as the belief diverges in the correct direction, dominated convergence in the integral defining \(q_\theta\) gives
\[
\mathbb P(x_{n+1}=\theta\mid\mathcal F_n,\theta)
\longrightarrow 1
\qquad\text{almost surely under either state}.
\]
Since these conditional probabilities are bounded by one, taking expectations and averaging over the prior yields
\[
\lim_{n\to\infty}\mathbb P(x_n=\theta)=1.
\]
Thus, asymptotic learning occurs in the sense of \cref{def:AL}. This convergence of report correctness probabilities does not assert that every sufficiently late report is correct.
\endproof

\section{Proof of \texorpdfstring{\Cref{thm:smoothly-learning-speed}}{thmsmoothlylearningspeed}: Learning Rate under Smooth Randomized Response}\label{app:proof:thm:smoothly-learning-speed}

Before providing a formal proof of this result, we introduce a series of lemmas that will serve as foundational support for the main theorem.

\begin{lemma}\label{lem:tail-bound}
    Let $X$  be a standard normal random variable, $\mathcal{N}(0, 1) $. For \( x > 0 \), the tail probability \( \mathbb{P}(X \geq x) \) is bounded as follows:
    \begin{equation*}
        \frac{e^{-x^2 / 2}}{\sqrt{2 \pi}} \left(\frac{1}{x} - \frac{1}{x^3}\right) \leq \mathbb{P}(X \geq x) \leq \frac{e^{-x^2 / 2}}{x \sqrt{2 \pi}}.
    \end{equation*}
\end{lemma}

\proof 
We proceed by using integration by parts. First, we have
\begin{equation*}
   \int_{x}^{\infty}e^{-\frac{t^{2}}{2}}dt=\left. \frac{-e^{-t^2/2}}{t} \right|_x^{\infty}-\int_x^{\infty} \frac{e^{-t^2/2}}{t^2} dt.
\end{equation*}
Since the second term on the right-hand side is non-positive, we obtain the upper bound
\begin{equation*}
\int_{x}^{\infty} e^{-\frac{t^{2}}{2}} \, dt \leq \frac{e^{-x^2/2}}{x}.
\end{equation*}
Thus,
\begin{equation*}
\mathbb{P}(X \geq x) \leq \frac{e^{-x^2/2}}{x \sqrt{2\pi}}.
\end{equation*}

Next, we apply integration by parts to refine the lower bound:
\begin{equation*}
    \int_x^{\infty} \frac{1}{\sqrt{2\pi}} e^{-t^2 / 2} \, dt = \left. \frac{-e^{-t^2/2}}{t \sqrt{2\pi}} \right|_x^{\infty} + \left. \frac{e^{-t^2/2}}{t^3 \sqrt{2\pi}} \right|_x^{\infty} + 3\int_x^{\infty} \frac{e^{-t^2/2}}{t^4 \sqrt{2\pi}} \, dt.
\end{equation*}
Dropping the last integral, which is non-negative, we obtain
\[
\mathbb{P}(X \geq x) \geq \frac{e^{-x^2 / 2}}{\sqrt{2\pi}} \left( \frac{1}{x} - \frac{1}{x^3} \right).
\]
Combining these bounds, we conclude that
\begin{equation*}
    \frac{e^{-x^2 / 2}}{\sqrt{2 \pi}} \left(\frac{1}{x} - \frac{1}{x^3}\right) \leq \mathbb{P}(X \geq x) \leq \frac{e^{-x^2 / 2}}{x \sqrt{2 \pi}}.
\end{equation*} 
\endproof

\begin{lemma}\label{lem:simplify_diff_equation}
For \(\varepsilon\in(0,\infty)\),
\[
\lim_{l\to\infty}\frac{D_+(l)}{G_-(-l)}=1,
\]
where
\[
D_+(l)=\log\frac{1-q_{+1}(l)}{1-q_{-1}(l)},
\qquad
G_-(-l)=\frac12
\left(e^{\varepsilon+\varepsilon^2\sigma^2/2}
-e^{-\varepsilon+\varepsilon^2\sigma^2/2}\right)
e^{-\varepsilon\sigma^2l/2},
\]
and \(q_\theta(l)=\mathbb P(x_{n+1}=-1\mid l_n=l,\theta)\).
\end{lemma}

\proof
Put \(z=\sigma^2l/2\). The negative-report probability is the sum of a Gaussian left-tail term and
\[
\frac12e^{-\varepsilon z}
\int_{-z}^{\infty}e^{-\varepsilon s}\phi_\theta(s)\,ds
=\frac12e^{-\theta\varepsilon+\varepsilon^2\sigma^2/2}
 e^{-\varepsilon z}
 \Phi\left(\frac{z+\theta-\varepsilon\sigma^2}{\sigma}\right).
\]
The left-tail term is bounded above by \(\Phi((-z-\theta)/\sigma)\). Gaussian tail bounds therefore imply, for some \(c>0\),
\[
q_\theta(l)
=\frac12e^{-\theta\varepsilon+\varepsilon^2\sigma^2/2}
 e^{-\varepsilon z}+O(e^{-c z^2}).
\]
In particular, both probabilities tend to zero, and their difference is asymptotic to \(G_-(-l)\). Expanding \(\log(1-q)\) gives
\[
D_+(l)=q_{-1}(l)-q_{+1}(l)
+O\left(q_{-1}(l)^2+q_{+1}(l)^2\right)
\sim G_-(-l).
\]
This characterizes the positive-report update, not the conditional expected update.
\endproof

\begin{lemma}\label{lem: A/B=1} 
    Let $A, B: \mathbb{R}_{>0} \to \mathbb{R}_{>0}$ be continuous functions that are eventually monotonically decreasing and tend to zero. Consider the sequences $(a_t)$ and $(b_t)$ defined by the recurrence relations:
    \[
    a_{t+1} = a_t + A(a_t) \quad \text{and} \quad b_{t+1} = b_t + B(b_t).
    \]
    Suppose that
    \[
    \lim_{x \to \infty} \frac{A(x)}{B(x)} = 1.
    \]
    Then it follows that
    \[
    \lim_{t \to \infty} \frac{a_t}{b_t} = 1.
    \]
     This is proved by \citet[Lemma 9]{Hann-Caruthers2018speed}.
\end{lemma}

\begin{lemma}\label{lem:discrete and continuous} 
    Assume that $A: \mathbb{R}_{>0} \to \mathbb{R}_{>0}$ is a continuous function with a convex differentiable tail and that $A(x)$ approaches zero as $x$ tends to infinity. Let $(a_n)$ be a sequence satisfying the recurrence relation:
    \[
    a_{n+1} = a_n + A(a_n).
    \]
    Moreover, suppose there exists a function $f: \mathbb{R}_{>0} \to \mathbb{R}_{>0}$ such that $f'(t) = A(f(t))$ for sufficiently large $t$. Then we have
    \[
    \lim_{n \to \infty} \frac{f(n)}{a_n} = 1.
    \]
    This proved by \citet[Lemma 10]{Hann-Caruthers2018speed}.
\end{lemma}

\begin{lemma}[Accumulated drift of a bounded log-likelihood process]
\label{lem:llr-accumulated-drift}
Suppose that under \(\theta=+1\), \(l_n\to+\infty\) almost surely, the log-likelihood increments \(\Delta_n=l_{n+1}-l_n\) satisfy \(|\Delta_n|\le b\), and
\[
h(l_n)=\mathbb E_{+1}[\Delta_n\mid\mathcal F_n]>0.
\]
Let \(A_0=0\) and \(A_n=\sum_{k=0}^{n-1}h(l_k)\). Then
\[
A_n\to\infty,
\qquad \frac{l_n}{A_n}\to1
\qquad\text{a.s.}
\]
If \(h(l)\sim H l^{-p}\) for \(H>0\) and \(p>0\), then
\[
l_n\sim((p+1)Hn)^{1/(p+1)}\qquad\text{a.s.}
\]
If instead \(h(l)=\Theta(e^{-\beta l})\) for \(\beta>0\), then
\[
l_n\sim\beta^{-1}\log n\qquad\text{a.s.}
\]
\end{lemma}

\proof
The likelihood-ratio identity gives
\[
\mathbb E_{+1}[e^{-\Delta_n}\mid\mathcal F_n]=1.
\]
For \(|d|\le b\), Taylor's theorem implies
\(e^{-d}-1+d\ge\frac12e^{-b}d^2\). Consequently,
\[
\mathbb E_{+1}[\Delta_n^2\mid\mathcal F_n]
\le2e^b h(l_n).
\]
Write \(M_n=l_n-l_0-A_n\) and \(\xi_{n+1}=M_{n+1}-M_n\). If \(A_\infty<\infty\), the predictable quadratic variation of \(M_n\) is finite and the martingale converges almost surely on this event, contradicting \(l_n\to\infty\). Thus \(A_n\to\infty\). Moreover,
\[
\begin{aligned}
\sum_{n=0}^{\infty}
\mathbb E_{+1}\left[
\left.\frac{\xi_{n+1}^2}{(1+A_{n+1})^2}\right|\mathcal F_n\right]
&\le2e^b\sum_{n=0}^{\infty}
\frac{A_{n+1}-A_n}{(1+A_{n+1})^2}\\
&\le2e^b\int_0^{\infty}\frac{du}{(1+u)^2}<\infty.
\end{aligned}
\]
The associated martingale series converges almost surely. Kronecker's lemma yields \(M_n/A_n\to0\), proving \(l_n/A_n\to1\).

In the power-law case,
\[
A_{n+1}-A_n\sim H A_n^{-p},
\qquad
A_{n+1}^{p+1}-A_n^{p+1}\longrightarrow(p+1)H.
\]
Summing and using \(l_n/A_n\to1\) proves the stated rate.

In the exponential case, for any \(\delta\in(0,1)\), eventually
\[
c_1e^{-\beta(1+\delta)A_n}
\le A_{n+1}-A_n
\le c_2e^{-\beta(1-\delta)A_n}
\]
for positive constants \(c_1,c_2\). Applying the mean value theorem to \(e^{\beta(1+\delta)A_n}\) for the lower bound and to \(e^{\beta(1-\delta)A_n}\) for the upper bound gives
\[
\frac1{\beta(1+\delta)}
\le\liminf_{n\to\infty}\frac{A_n}{\log n}
\le\limsup_{n\to\infty}\frac{A_n}{\log n}
\le\frac1{\beta(1-\delta)}.
\]
Here boundedness of \(A_{n+1}-A_n\le b\) controls the mean-value remainder. Letting \(\delta\downarrow0\) proves the result.
\endproof

Given these lemmas, we now prove \Cref{thm:smoothly-learning-speed} using the full conditional increment.

\proof
Let \(\beta=\varepsilon\sigma^2/2\). By \Cref{lem:simplify_diff_equation},
\[
\begin{aligned}
q_{+1}(l)&\sim\frac12e^{-\varepsilon+\varepsilon^2\sigma^2/2}e^{-\beta l},\\
D_+(l)&\sim\frac12e^{-\varepsilon+\varepsilon^2\sigma^2/2}
 (e^{2\varepsilon}-1)e^{-\beta l},\\
D_-(l)&=\log\frac{q_{+1}(l)}{q_{-1}(l)}\longrightarrow-2\varepsilon.
\end{aligned}
\]
Thus the full expected increment is
\[
\begin{aligned}
h_+(l)
&=(1-q_{+1}(l))D_+(l)+q_{+1}(l)D_-(l)\\
&\sim\frac12e^{-\varepsilon+\varepsilon^2\sigma^2/2}
 (e^{2\varepsilon}-1-2\varepsilon)e^{-\beta l}.
\end{aligned}
\]
Its coefficient is positive because \(e^{2\varepsilon}>1+2\varepsilon\). For every finite \(l\), \(h_+(l)\) is the strictly positive KL divergence between the two conditional report distributions.

Metric differential privacy bounds the ratio of report probabilities at signals separated by \(2\) by \(e^{2\varepsilon}\). Since the two Gaussian signal laws are translations by \(2\), integrating this bound gives \(|D_x(l)|\le2\varepsilon\) for both reports. Asymptotic learning was proved in \Cref{app:proof:thm:smoothly-asymptotic-learning-C}. Applying \Cref{lem:llr-accumulated-drift} therefore yields
\[
l_n\sim\frac{2}{\varepsilon\sigma^2}\log n
\qquad\text{a.s. under }\theta=+1.
\]
The proof allows negative reports at arbitrarily late times. It establishes the leading logarithmic rate and does not identify a multiplicative constant in \(e^{\beta l_n}\).
\endproof

\section{Proof of \texorpdfstring{\Cref{thm:smoothly-stopping-time}}{thmsmoothlystoppingtime}: Expected Stopping Times under Smooth Randomized Response}
\label{app:proof:thm:smoothly-stopping-time}

We first give a stopped-process argument that will also be used for the alternative reporting mechanisms. Throughout, \(l_0=0\), \(l_n\) is indexed by submitted reports, and
\[
\tau_B=\inf\{n\ge1:|l_n|\ge B\}.
\]

\begin{lemma}[Expected two-sided stopping time from a potential]
\label{lem:two-sided-stopped-potential}
Suppose the reporting model is symmetric under \((\theta,l,x)\mapsto(-\theta,-l,-x)\), the increments satisfy \(|l_{n+1}-l_n|\le b\), and the true-state drift \(h_+(l)\) is continuous and strictly positive at every finite \(l\). Let \(\mathbb P=\frac12\mathbb P_{+1}+\frac12\mathbb P_{-1}\) be the prior mixture, and let \(V\) be an even, continuous potential. Suppose its mixture drift
\[
g(l)=\mathbb E[V(l_{n+1})-V(l_n)\mid l_n=l]
\]
is continuous. If \(g(l)\to c>0\) as \(|l|\to\infty\), and \(V(l_{\tau_B})/a(B)\to1\) uniformly over possible exit locations, where \(B=o(a(B))\), then
\[
\mathbb E_{+1}[\tau_B]=\mathbb E_{-1}[\tau_B]
\sim\frac{a(B)}{c}.
\]
If instead \(g\) is bounded above and below by positive constants outside a compact interval and \(V(l_{\tau_B})=\Theta(a(B))\) uniformly, then \(\mathbb E_{+1}[\tau_B]=\Theta(a(B))\).
\end{lemma}

\proof
For fixed \(B\), let \(m_B=\inf_{|l|\le B}h_+(l)>0\). Stopping the drift decomposition at \(\tau_B\wedge N\) gives
\[
m_B\mathbb E_{+1}[\tau_B\wedge N]
\le\mathbb E_{+1}[l_{\tau_B\wedge N}]
\le B+b.
\]
Thus \(\mathbb E_{+1}[\tau_B]<\infty\); symmetry gives the same conclusion under the other state and the mixture. For any fixed \(R\), the same decomposition, using positivity of the drift everywhere, yields
\[
\mathbb E_{+1}\!\left[\sum_{n<\tau_B}\mathbf1_{\{|l_n|\le R\}}\right]
\le\frac{B+b}{m_R}=O_R(B).
\]
Reflection symmetry gives this bound under the mixture as well.

Since the stopped process stays in \([-B-b,B+b]\), continuity and integrability justify the stopped drift identity
\[
\mathbb E[V(l_{\tau_B})]-V(0)
=\mathbb E\!\left[\sum_{n<\tau_B}g(l_n)\right].
\]
For any \(\delta\in(0,c)\), choose \(R\) so that \(|g(l)-c|\le\delta\) outside \([-R,R]\). The compact occupation bound implies
\[
(c-\delta)\mathbb E[\tau_B]-O_R(B)
\le\mathbb E[V(l_{\tau_B})]-V(0)
\le(c+\delta)\mathbb E[\tau_B]+O_R(B).
\]
Divide by \(a(B)\), use \(B=o(a(B))\), and let \(\delta\downarrow0\). This proves the first assertion. The same inequalities with fixed positive upper and lower bounds prove the second. Finally, symmetry gives \(\mathbb E[\tau_B]=\mathbb E_{+1}[\tau_B]=\mathbb E_{-1}[\tau_B]\).
\endproof

\proof
Let
\[
\beta=\frac{\varepsilon\sigma^2}{2},
\qquad
C(\varepsilon)=\frac{\varepsilon\sigma^2}{4}
\left(e^{\varepsilon+\varepsilon^2\sigma^2/2}
-e^{-\varepsilon+\varepsilon^2\sigma^2/2}\right).
\]
By \Cref{lem:simplify_diff_equation}, as \(l\to\infty\),
\[
q_{+1}(l)\sim\tfrac12e^{-\varepsilon+\varepsilon^2\sigma^2/2}e^{-\beta l},
\quad D_+(l)\sim\frac{C(\varepsilon)}{\beta}e^{-\beta l},
\quad D_-(l)\to-2\varepsilon.
\]
Consequently, averaging over both reports gives
\[
\begin{aligned}
\mathbb E_{+1}[e^{\beta l_{n+1}}-e^{\beta l_n}\mid l_n=l]
&=e^{\beta l}\left[(1-q_{+1}(l))(e^{\beta D_+(l)}-1)
+q_{+1}(l)(e^{\beta D_-(l)}-1)\right]\\
&\longrightarrow K(\varepsilon,\sigma),
\end{aligned}
\]
where
\[
K(\varepsilon,\sigma)
=C(\varepsilon)
-\frac12e^{-\varepsilon+\varepsilon^2\sigma^2/2}
\left(1-e^{-\varepsilon^2\sigma^2}\right).
\]
This constant is positive: after factoring out \(\frac12e^{-\varepsilon+\varepsilon^2\sigma^2/2}\), its remaining factor is
\(\beta(e^{2\varepsilon}-1)-(1-e^{-2\varepsilon\beta})>0\).

To treat both boundaries, set
\[
p(l)=\frac{e^l}{1+e^l},
\qquad
V(l)=p(l)e^{\beta l}+(1-p(l))e^{-\beta l}
=\frac{\cosh((\beta+1/2)l)}{\cosh(l/2)}.
\]
The posterior identity \(p(l_n)=\mathbb P(\theta=+1\mid\mathcal F_n)\) implies
\[
g(l)=p(l)\mathbb E_{+1}[\Delta e^{\beta l_n}\mid l_n=l]
+(1-p(l))\mathbb E_{-1}[\Delta e^{-\beta l_n}\mid l_n=l].
\]
At \(+\infty\), the first conditional increment tends to \(K(\varepsilon,\sigma)\), while the second is \(O(e^{-2\beta l})\). Hence \(g(l)\to K(\varepsilon,\sigma)\); reflection gives the same limit at \(-\infty\).

The increments are bounded by \(2\varepsilon\), and the true-state drift is continuous and strictly positive. Furthermore, outward crossings of \(+B\) use \(D_+(l)\to0\), and outward crossings of \(-B\) have the same property by symmetry. Thus
\[
B\le|l_{\tau_B}|\le B+o(1),
\qquad
V(l_{\tau_B})\sim e^{\beta B}
\]
uniformly over exit locations. Applying \Cref{lem:two-sided-stopped-potential} gives
\[
\mathbb E[\tau_B\mid\theta=+1]
\sim\frac{1}{K(\varepsilon,\sigma)}
\exp\left(\frac{\varepsilon\sigma^2}{2}B\right).
\]
This argument includes paths that reach \(-B\) first.

If each potential agent participates independently with probability \(\eta(\varepsilon)>0\), independently of the report process, the inter-report calendar gaps are independent geometric random variables with mean \(1/\eta(\varepsilon)\). Therefore,
\[
\mathbb E[T_B\mid\theta=+1]
=\frac{\mathbb E[\tau_B\mid\theta=+1]}{\eta(\varepsilon)}
\sim\frac{e^{\varepsilon\sigma^2B/2}}{\eta(\varepsilon)K(\varepsilon,\sigma)}.
\]
This completes the proof.
\endproof

\section{The Evolution of Log-Likelihood Ratio}
\label{app:simulation:thm:smoothly-learning-speed}
\begin{figure}[htbp]
    \centering
    \includegraphics[width=0.75\textwidth]{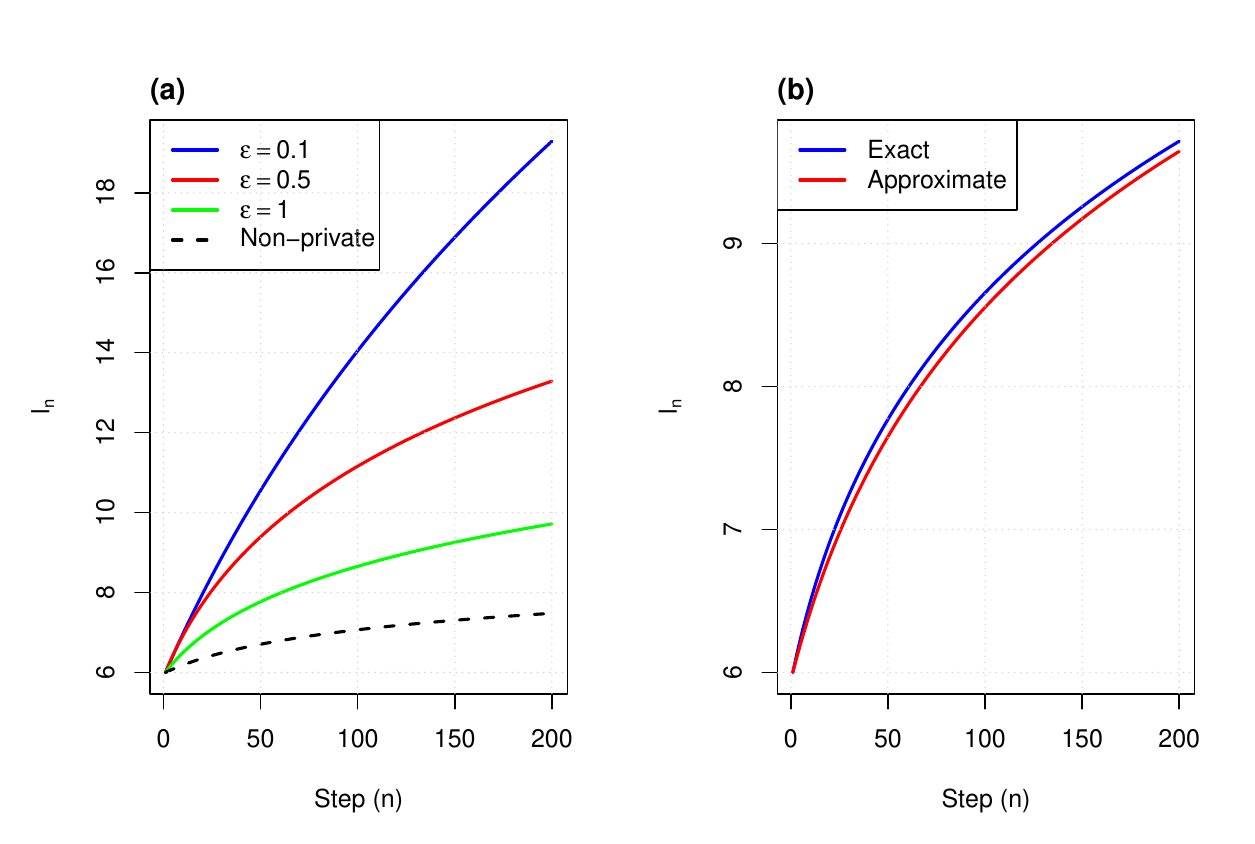}
    \caption{(a) Evolution of \( l_n \) over steps \( n \) for different privacy budgets and non-private scenario \( \varepsilon \). (b) Comparison between exact and approximate solutions for \( l_n \) when \( \varepsilon = 1 \).}
    \label{fig:ln_evolution}
\end{figure}

To support \Cref{thm:smoothly-learning-speed}, we performed simulations to analyze the evolution of the log-likelihood ratio \( l_n \) without an approximation. \Cref{fig:ln_evolution}(a) illustrates the trajectory of \( l_n \) under three different privacy budgets and the non-private scenario, revealing a clear relationship between privacy constraints and learning efficiency. Specifically, as the privacy budget decreases, the rate at which the log-likelihood ratio increases becomes steeper. The effect of privacy on public-belief growth depends on the complete expected log-likelihood increment, which includes both correct and incorrect released reports. The logarithmic rate is established by the stochastic drift argument in \Cref{app:proof:thm:smoothly-learning-speed}. Moreover, we observe a significant difference in learning speed between the private and non-private settings: In the private setting, the log-likelihood ratio grows at a rate of \( \Theta_{\varepsilon}(\log(n)) \), whereas in the non-private setting, where privacy budgets approach infinity, asymptotic learning occurs at the slower rate of \( \Theta(\sqrt{\log(n)}) \). This confirms that our privacy-aware mechanism leads to a provable acceleration of information aggregation compared to classical sequential learning models.

Furthermore, \Cref{fig:ln_evolution}(b) displays the original comparison between the exact and approximate solutions. This comparison alone does not justify replacing the stochastic report process by the positive-report recursion. That recursion describes a prescribed report history, whereas the learning-rate proof must include negative reports as well. The stochastic drift proof, rather than agreement of the plotted curves, establishes the learning-rate conclusion.

\section{Proof of \texorpdfstring{\Cref{thm:smoothly-asymptotic-learning-C-hete}}{thmsmoothlyasymptoticlearningChete}: Asymptotic Learning under Heterogeneous Privacy Budgets}\label{app:proof:thm:smoothly-asymptotic-learning-C-hete}

\proof
We use the same likelihood-ratio argument as in the homogeneous setting, averaging the reporting kernel over the unobserved privacy parameter. Let
\[
\overline Q(x\mid s,l)=\int_0^1Q_\varepsilon(x\mid s,l)\,d\varepsilon,
\qquad
q_\theta(l)=\int_{\mathbb R}\overline Q(-1\mid s,l)f_\theta(s)\,ds.
\]
For every \(\varepsilon>0\), \(Q_\varepsilon(-1\mid s,l)\) is strictly decreasing in \(s\). Coupling the two Gaussian signals by \(S_{+1}=S_{-1}+2\) therefore gives
\[
0<q_{+1}(l)<q_{-1}(l)<1
\qquad\text{for every finite }l.
\]
Both probabilities are continuous. The belief increments are
\[
D_+(l)=\log\frac{1-q_{+1}(l)}{1-q_{-1}(l)}>0,
\qquad
D_-(l)=\log\frac{q_{+1}(l)}{q_{-1}(l)}<0,
\]
and \(l_{n+1}-l_n=D_{x_{n+1}}(l_n)\).

Under \(\theta=-1\), \(e^{l_n-l_0}\) is a nonnegative martingale and has a finite limit almost surely. If this limit were strictly positive, \(l_n\) would converge to a finite value. Continuity and strict separation of the report probabilities would then bound both \(|D_+(l_n)|\) and \(|D_-(l_n)|\) away from zero, contradicting convergence of \(l_n\). Thus the martingale limit is zero and \(l_n\to-\infty\) almost surely. Applying the reciprocal likelihood-ratio martingale under \(\theta=+1\) proves \(l_n\to+\infty\) almost surely.

For each \(\varepsilon>0\) and each fixed signal, the probability of an incorrect report tends to zero as the public belief tends to the true-state extreme. Dominated convergence, first over signals and then over \(\varepsilon\sim U[0,1]\), gives
\[
\mathbb P(x_{n+1}=\theta\mid l_n,\theta)
\longrightarrow1
\qquad\text{almost surely under either state}.
\]
Since these conditional probabilities are bounded by one, taking expectations and averaging over the prior yields
\[
\lim_{n\to\infty}\mathbb P(x_n=\theta)=1.
\]
Thus, asymptotic learning occurs in the sense of \cref{def:AL}. This convergence of report correctness probabilities does not assert that every sufficiently late report is correct. The same likelihood-ratio argument and dominated convergence apply to any distribution \(G\) on \([0,1]\) with \(G(\{0\})=0\), establishing both public-belief convergence and asymptotic learning.
\endproof

\section{Proof of \texorpdfstring{\Cref{thm:smoothly-learning-speed-hete}}{thmsmoothlylearningspeedhete}: Learning Rate under Smooth Randomized Response with Heterogeneous Privacy Budget}
\label{app:proof:thm:smoothly-learning-speed-hete}

Before providing a formal proof of this result, we characterize
the conditional report probabilities and the resulting
increments of the public log-likelihood ratio.

\proof
We index the public log-likelihood ratio by the number of
submitted reports, with \(l_0=0\), and let
\[
\mathcal{F}_n=\sigma(x_1,\ldots,x_n).
\]
For \(\theta\in\{-1,+1\}\), define
\[
q_\theta(l)
=
\mathbb{P}(x_{n+1}=-1\mid l_n=l,\theta).
\]
Because the individual privacy parameters are not publicly
observed, these report probabilities integrate over
\(\varepsilon_{n+1}\sim U[0,1]\).

By definition, the public-belief increments associated with
reports \(+1\) and \(-1\) are
\[
D_+(l)
=
\log\frac{1-q_{+1}(l)}{1-q_{-1}(l)},
\qquad
D_-(l)
=
\log\frac{q_{+1}(l)}{q_{-1}(l)}.
\]
Thus,
\[
l_{n+1}
=
\begin{cases}
l_n+D_+(l_n), & x_{n+1}=+1,\\
l_n+D_-(l_n), & x_{n+1}=-1.
\end{cases}
\]

We first analyze the report probabilities as \(l\to+\infty\).
Let \(f_\theta\) denote the density of \(N(\theta,\sigma^2)\),
and write
\[
z=\frac{\sigma^2l}{2},
\qquad
t(l)=-z,
\qquad
a=\frac12.
\]
Under personalized smooth randomized response,
\begin{equation*}
\begin{aligned}
q_\theta(l)
={}&
\int_0^1
\int_{-\infty}^{-z}
f_\theta(s)
\left(1-ae^{\varepsilon(s+z)}\right)
\,ds\,d\varepsilon\\
&+
\int_0^1
\int_{-z}^{\infty}
f_\theta(s)
ae^{-\varepsilon(s+z)}
\,ds\,d\varepsilon.
\end{aligned}
\end{equation*}

By the Gaussian tail bound in Lemma~\ref{lem:tail-bound},
the first term is \(O(e^{-cz^2})\) for some \(c>0\).
The error from extending the second inner integral to
the entire real line has the same type of bound,
uniformly over \(\varepsilon\in[0,1]\).
Indeed, for \(s\le-z\),
\[
e^{-\varepsilon(s+z)}
\le e^{-(s+z)},
\]
and the corresponding exponentially weighted Gaussian
left tail is \(O(e^{-cz^2})\), after reducing \(c\) if necessary.
Using the Gaussian exponential moment, we therefore obtain
\begin{equation*}
q_\theta(l)
=
\frac12
\int_0^1
\exp\left(
-\theta\varepsilon+\frac{\varepsilon^2\sigma^2}{2}
\right)
e^{-z\varepsilon}
\,d\varepsilon
+
O(e^{-cz^2}).
\end{equation*}

As \(z\to\infty\), the integral is dominated by privacy
parameters near zero. The Taylor expansion
\[
\exp\left(
-\theta\varepsilon+\frac{\varepsilon^2\sigma^2}{2}
\right)
=
1-\theta\varepsilon+O(\varepsilon^2)
\]
has a remainder bounded uniformly on \([0,1]\).
Moreover,
\[
\int_0^1 e^{-z\varepsilon}\,d\varepsilon
=
\frac{1-e^{-z}}{z},
\]
\[
\int_0^1 \varepsilon e^{-z\varepsilon}\,d\varepsilon
=
\frac{1-(1+z)e^{-z}}{z^2},
\]
and
\[
\int_0^1 \varepsilon^2e^{-z\varepsilon}\,d\varepsilon
=
O(z^{-3}).
\]
Consequently,
\begin{equation*}
q_\theta(l)
=
\frac1{2z}
-
\frac{\theta}{2z^2}
+
O(z^{-3}).
\end{equation*}
In particular,
\begin{equation*}
q_{-1}(l)-q_{+1}(l)
=
\frac1{z^2}+O(z^{-3}).
\end{equation*}

It follows that
\begin{equation*}
\begin{aligned}
D_+(l)
&=
\log\left(
1+
\frac{q_{-1}(l)-q_{+1}(l)}
{1-q_{-1}(l)}
\right)\\
&=
\frac1{z^2}+O(z^{-3}),
\end{aligned}
\end{equation*}
whereas
\begin{equation*}
\begin{aligned}
D_-(l)
&=
\log\left(
1-
\frac{q_{-1}(l)-q_{+1}(l)}
{q_{-1}(l)}
\right)\\
&=
-\frac2z+O(z^{-2}).
\end{aligned}
\end{equation*}
Thus,
\[
D_+(l)\sim\frac4{\sigma^4l^2},
\qquad
D_-(l)\sim-\frac4{\sigma^2l}.
\]

The growth of the stochastic public-belief process depends
on both increments.
Conditional on \(\theta=+1\), define
\[
h_+(l)
=
\mathbb{E}_{+1}[l_{n+1}-l_n\mid l_n=l].
\]
Then
\begin{equation*}
\begin{aligned}
h_+(l)
&=
(1-q_{+1}(l))D_+(l)
+
q_{+1}(l)D_-(l)\\
&=
D_{\mathrm{KL}}\left(
\operatorname{Bern}(q_{+1}(l))
\,\Vert\,
\operatorname{Bern}(q_{-1}(l))
\right).
\end{aligned}
\end{equation*}
The contributions of order \(z^{-2}\) cancel, so we
calculate the next-order term using the KL divergence.

Let
\[
u=q_{+1}(l),
\qquad
v=q_{-1}(l),
\qquad
d=v-u.
\]
Since
\[
u\sim v\sim\frac1{2z},
\qquad
d\sim\frac1{z^2},
\]
a Taylor expansion gives
\[
D_{\mathrm{KL}}\left(
\operatorname{Bern}(u)
\,\Vert\,
\operatorname{Bern}(v)
\right)
=
\frac{d^2}{2v(1-v)}
+
O\left(\frac{|d|^3}{v^2}\right).
\]
Therefore,
\begin{equation*}
h_+(l)
\sim
\frac1{z^3}
=
\frac8{\sigma^6l^3}.
\end{equation*}
Similarly, under \(\theta=-1\),
\begin{equation*}
\begin{aligned}
h_-(l)
&:=
\mathbb{E}_{-1}[l_{n+1}-l_n\mid l_n=l]\\
&=
-
D_{\mathrm{KL}}\left(
\operatorname{Bern}(q_{-1}(l))
\,\Vert\,
\operatorname{Bern}(q_{+1}(l))
\right)\\
&\sim
-\frac8{\sigma^6l^3}.
\end{aligned}
\end{equation*}

The corresponding asymptotic differential equation is
\begin{equation}\label{eq:differential-equation}
\frac{\mathrm{d}f}{\mathrm{d}t}(t)
=
\frac8{\sigma^6 f(t)^3}.
\end{equation}
Separating variables and integrating yields
\[
f(t)^3\,df
=
\frac8{\sigma^6}\,dt,
\]
and hence
\[
\frac{f(t)^4}{4}
=
\frac8{\sigma^6}t+\text{constant}.
\]
This identifies the candidate rate
\[
f(t)
\sim
\left(
\frac{32}{\sigma^6}t
\right)^{1/4}.
\]
To establish this rate for the stochastic process,
we next control its martingale fluctuations.

Let
\[
\Delta_n=l_{n+1}-l_n,
\qquad
Y_n=l_n^4,
\]
and, for each state \(\theta\), define
\[
r_\theta(l)
=
\mathbb{E}_\theta[Y_{n+1}-Y_n\mid l_n=l].
\]
The expansions of \(q_\theta(l)\) and \(D_\pm(l)\) imply
that, for \(j\in\{2,3,4\}\),
\begin{equation*}
\begin{aligned}
\mathbb{E}_\theta[
|\Delta_n|^j\mid l_n=l]
&=
(1-q_\theta(l))|D_+(l)|^j
+
q_\theta(l)|D_-(l)|^j\\
&=
O(l^{-(j+1)})
\end{aligned}
\end{equation*}
as \(l\to+\infty\).
Expanding the fourth power exactly gives
\begin{equation*}
\begin{aligned}
r_\theta(l)
={}&
4l^3h_\theta(l)
+
6l^2\mathbb{E}_\theta[
\Delta_n^2\mid l_n=l]\\
&+
4l\mathbb{E}_\theta[
\Delta_n^3\mid l_n=l]
+
\mathbb{E}_\theta[
\Delta_n^4\mid l_n=l].
\end{aligned}
\end{equation*}
The last three terms tend to zero. Consequently,
\[
r_{+1}(l)\longrightarrow\frac{32}{\sigma^6},
\qquad
r_{-1}(l)\longrightarrow-\frac{32}{\sigma^6},
\qquad l\to+\infty.
\]
Reflection symmetry gives
\[
r_{+1}(-l)=r_{-1}(l).
\]
Hence \(r_{+1}\) has finite limits at both tails.
Since it is continuous, there exists a constant
\(C_1>0\) such that
\[
|r_{+1}(l)|\le C_1
\qquad\text{for all }l\in\mathbb{R}.
\]
It follows that
\[
\mathbb{E}_{+1}[Y_n]
=
Y_0+
\sum_{k=0}^{n-1}\mathbb{E}_{+1}[r_{+1}(l_k)]
\le C_1n,
\]
where \(Y_0=0\).

We also need a bound on the fluctuations of \(Y_n\).
Let
\[
\overline{Q}(x\mid s,l)
=
\int_0^1Q_\varepsilon(x\mid s,l)\,d\varepsilon
\]
denote the reporting kernel after averaging over privacy
parameters.
Since each \(Q_\varepsilon\) satisfies
\(\varepsilon\)-metric differential privacy and
\(\varepsilon\le1\),
\[
e^{-2}\overline{Q}(x\mid s,l)
\le
\overline{Q}(x\mid s+2,l)
\le
e^2\overline{Q}(x\mid s,l).
\]
The two Gaussian signal distributions differ by a
translation of \(2\). Integrating these inequalities yields
\[
e^{-2}
\le
\frac{\mathbb{P}(x_{n+1}=x\mid l_n=l,\theta=+1)}
{\mathbb{P}(x_{n+1}=x\mid l_n=l,\theta=-1)}
\le
e^2.
\]
Thus,
\[
|\Delta_n|\le2.
\]

For \(|d|\le2\), the polynomial expansion gives
\[
|(l+d)^4-l^4|
\le
C_2(1+|l|^3)|d|
\]
for some constant \(C_2>0\).
By the preceding second-moment estimate and reflection
symmetry,
\[
\mathbb{E}_{+1}[
\Delta_n^2\mid l_n=l]
=
O(|l|^{-3})
\qquad\text{as }|l|\to\infty.
\]
Combining these bounds and using continuity on compact
intervals, there exists \(C_3>0\) such that
\begin{equation*}
\mathbb{E}_{+1}\left[
(Y_{n+1}-Y_n)^2
\,\middle|\,\mathcal{F}_n
\right]
\le
C_3(1+|l_n|^3).
\end{equation*}

Define the martingale
\[
M_n
=
Y_n-Y_0-\sum_{k=0}^{n-1}r_{+1}(l_k)
\]
under \(\mathbb{P}_{+1}\), and write
\[
\xi_{n+1}=M_{n+1}-M_n.
\]
Since conditional variance is bounded by the corresponding
conditional second moment,
\[
\mathbb{E}_{+1}[
\xi_{n+1}^2\mid\mathcal{F}_n]
\le
C_3(1+|l_n|^3).
\]
Moreover, by the fourth-moment bound,
\[
\mathbb{E}_{+1}[|l_n|^3]
\le
\left(\mathbb{E}_{+1}[l_n^4]\right)^{3/4}
\le
(C_1n)^{3/4}.
\]
Therefore, for some constant \(C_4>0\),
\[
\mathbb{E}_{+1}[\xi_{n+1}^2]
\le
C_4(1+n^{3/4}),
\]
and
\[
\sum_{n=0}^{\infty}
\frac{\mathbb{E}_{+1}[\xi_{n+1}^2]}{(n+1)^2}
<
\infty.
\]
The martingale series
\[
\sum_{n=0}^{\infty}\frac{\xi_{n+1}}{n+1}
\]
thus converges almost surely.
By Kronecker's lemma,
\[
\frac{M_n}{n}\longrightarrow0,
\qquad\mathbb{P}_{+1}\text{-almost surely}.
\]

By the asymptotic-learning result under personalized
smooth randomized response,
\[
l_n\longrightarrow+\infty,
\qquad\mathbb{P}_{+1}\text{-almost surely}.
\]
Consequently,
\[
r_{+1}(l_n)\longrightarrow\frac{32}{\sigma^6},
\qquad\mathbb{P}_{+1}\text{-almost surely}.
\]
Using Cesaro convergence and the martingale decomposition,
we obtain
\begin{equation*}
\begin{aligned}
\frac{l_n^4}{n}
&=
\frac{Y_0}{n}
+
\frac1n\sum_{k=0}^{n-1}r_{+1}(l_k)
+
\frac{M_n}{n}\\
&\longrightarrow
\frac{32}{\sigma^6},
\qquad\mathbb{P}_{+1}\text{-almost surely}.
\end{aligned}
\end{equation*}
Since \(l_n\) is eventually positive almost surely,
taking the positive fourth root yields
\[
l_n
\sim
\left(
\frac{32}{\sigma^6}n
\right)^{1/4},
\qquad\mathbb{P}_{+1}\text{-almost surely}.
\]
Hence, the convergence rate is
\[
f(n)
=
\left(
\frac{32}{\sigma^6}n
\right)^{1/4}
=
\Theta(n^{1/4}).
\]
This completes the proof.
\endproof

\section{Proof of \texorpdfstring{\Cref{thm:smoothly-learning-speed-upperbound}}{thmsmoothlylearningspeedupperbound}: Learning Rate Bounds under Personalized Privacy Parameters}
\label{app:proof:thm:smoothly-learning-speed-upperbound}

\proof
We first prove the upper bound. Let \(G\) be the distribution of the unobserved privacy parameter on \([0,1]\), with \(G(\{0\})=0\), and set \(z=\sigma^2l/2\). Uniformly over \(\varepsilon\in[0,1]\), the Gaussian tail calculation used in \Cref{lem:simplify_diff_equation} gives
\[
q_\theta(l)
=\frac12\int_0^1e^{-\theta\varepsilon+\varepsilon^2\sigma^2/2}e^{-z\varepsilon}\,dG(\varepsilon)
+O(e^{-cz^2})
\]
for some \(c>0\). In particular, \(q_\theta(l)\to0\). Write \(\delta(l)=q_{-1}(l)-q_{+1}(l)>0\). Since
\[
e^{\varepsilon+\varepsilon^2\sigma^2/2}
-e^{-\varepsilon+\varepsilon^2\sigma^2/2}\le C_1\varepsilon
\qquad(0\le\varepsilon\le1),
\]
we obtain
\[
\delta(l)\le C_2\sup_{\varepsilon\in[0,1]}
\varepsilon e^{-z\varepsilon}+O(e^{-cz^2})
\le\frac{C_3}{l}
\]
for all sufficiently large \(l\).

The full expected increment, including both reports, is
\[
h_+(l)=(1-q_{+1}(l))\log\frac{1-q_{+1}(l)}{1-q_{-1}(l)}
+q_{+1}(l)\log\frac{q_{+1}(l)}{q_{-1}(l)}.
\]
It is a strictly positive KL divergence. Because the negative-report contribution is negative,
\[
0<h_+(l)\le D_+(l)
=\log\left(1+\frac{\delta(l)}{1-q_{-1}(l)}\right)
\le\frac{C_4}{l}.
\]

We next prove a logarithmic lower bound. There exist \(0<u\le v\le1\) such that \(G([u,v])>0\). The difference between the two negative-report probabilities is positive for every fixed \(\varepsilon>0\). Restricting its integral to \([u,v]\), and using the uniform Gaussian remainder, gives
\[
\delta(l)\ge C_5e^{-vz}
\]
for all sufficiently large \(l\). The Bernoulli KL lower bound \(D_{\mathrm{KL}}(P\Vert Q)\ge2\|P-Q\|_{\mathrm{TV}}^2\) then yields
\[
h_+(l)\ge2\delta(l)^2\ge C_6e^{-\sigma^2l}.
\]

It remains to transfer these drift bounds to the stochastic belief process. By \Cref{app:proof:thm:smoothly-asymptotic-learning-C-hete}, \(l_n\to\infty\) almost surely. Averaging the metric-DP bound over \(\varepsilon\le1\) gives \(|l_{n+1}-l_n|\le2\). Thus \Cref{lem:llr-accumulated-drift} applies to
\[
A_n=\sum_{k=0}^{n-1}h_+(l_k),
\qquad l_n/A_n\to1\quad\text{a.s.}
\]
Eventually \(A_n/2\le l_n\le2A_n\), so
\[
C_6e^{-2\sigma^2A_n}
\le A_{n+1}-A_n\le\frac{2C_4}{A_n}.
\]
The upper bound implies \(A_{n+1}^2-A_n^2=O(1)\), hence \(A_n=O(\sqrt n)\). The lower bound implies
\[
e^{2\sigma^2A_{n+1}}-e^{2\sigma^2A_n}\ge2\sigma^2C_6,
\]
and hence \(A_n=\Omega(\log n)\). Combining these bounds with \(l_n/A_n\to1\),
\[
l_n=O(\sqrt n),
\qquad l_n=\Omega(\log n),
\qquad\mathbb P_{+1}\text{-almost surely}.
\]
This completes the proof.
\endproof

\section{Proof of \texorpdfstring{\Cref{thm:personalized-stopping-time}}{thmpersonalizedstoppingtime}: Expected Stopping Times under Personalized Privacy}
\label{app:proof:thm:personalized-stopping-time}

\proof
Suppose that the privacy parameters are drawn independently
from \(U[0,1]\), independently of the private signals and the
state, and are not publicly observed.
Public beliefs are therefore updated by integrating over
the privacy parameters.

We index the public log-likelihood ratio by the number of
submitted reports, with \(l_0=0\) under the symmetric prior.
Let
\[
\mathcal{F}_n=\sigma(x_1,\ldots,x_n),
\qquad
\tau_B=\inf\{n\geq 1:|l_n|\geq B\}.
\]
Write \(\mathbb{P}_\theta\) and \(\mathbb{E}_\theta\) for
probability and expectation conditional on state \(\theta\).
Let \(\mathbb{P}\) and \(\mathbb{E}\) denote the corresponding
quantities under the symmetric prior mixture
\[
\mathbb{P}
=
\frac12\mathbb{P}_{+1}
+
\frac12\mathbb{P}_{-1}.
\]

For a current public belief \(l\), define
\[
P_\theta(x\mid l)
=
\mathbb{P}_\theta(x_{n+1}=x\mid l_n=l),
\qquad
q_\theta(l)=P_\theta(-1\mid l).
\]
The public-belief increments associated with the two reports are
\[
\Delta_+(l)
=
\log\frac{1-q_{+1}(l)}{1-q_{-1}(l)},
\qquad
\Delta_-(l)
=
\log\frac{q_{+1}(l)}{q_{-1}(l)}.
\]
Thus,
\[
l_{n+1}=l_n+\Delta_{x_{n+1}}(l_n).
\]

We first derive the large-positive-belief expansion of
\(q_\theta(l)\).
Let \(f_\theta\) denote the density of \(N(\theta,\sigma^2)\),
and write
\[
z=\frac{\sigma^2l}{2},
\qquad
t(l)=-z.
\]
Under personalized smooth randomized response,
\[
\begin{aligned}
q_\theta(l)
={}&
\int_0^1
\int_{-\infty}^{t(l)}
\left(1-\frac12e^{\varepsilon(s-t(l))}\right)
f_\theta(s)\,ds\,d\varepsilon\\
&+
\frac12
\int_0^1 e^{\varepsilon t(l)}
\int_{t(l)}^\infty
e^{-\varepsilon s}f_\theta(s)\,ds\,d\varepsilon.
\end{aligned}
\]
As \(l\to+\infty\), Gaussian tail bounds imply that the first
term and the error from extending the second inner integral
to the entire real line are \(O(e^{-c z^2})\) for some
\(c>0\).
These bounds hold uniformly over \(\varepsilon\in[0,1]\):
on \(s\le t(l)\),
\[
e^{\varepsilon(t(l)-s)}
\le e^{t(l)-s},
\]
and the resulting exponentially weighted Gaussian tail
still decays quadratically in \(z\).
Using the Gaussian exponential moment, we obtain
\[
q_\theta(l)
=
\frac12
\int_0^1
e^{-z\varepsilon}
\exp\left(
-\theta\varepsilon+\frac{\sigma^2\varepsilon^2}{2}
\right)
d\varepsilon
+
O(e^{-c z^2}).
\]

The expansion
\[
\exp\left(
-\theta\varepsilon+\frac{\sigma^2\varepsilon^2}{2}
\right)
=
1-\theta\varepsilon+O(\varepsilon^2)
\]
has a remainder bounded uniformly on \([0,1]\).
Moreover,
\[
\int_0^1 e^{-z\varepsilon}\,d\varepsilon
=
\frac1z+O\left(\frac{e^{-z}}z\right),
\]
\[
\int_0^1 \varepsilon e^{-z\varepsilon}\,d\varepsilon
=
\frac1{z^2}
+
O\left(\frac{e^{-z}}z\right),
\]
and
\[
\int_0^1 \varepsilon^2e^{-z\varepsilon}\,d\varepsilon
=
O(z^{-3}).
\]
Consequently,
\[
q_\theta(l)
=
\frac1{2z}
-
\frac{\theta}{2z^2}
+
O(z^{-3}),
\]
so that
\[
q_{-1}(l)-q_{+1}(l)
=
\frac1{z^2}+O(z^{-3}).
\]
It follows that
\[
\Delta_+(l)
=
\frac1{z^2}+O(z^{-3}),
\qquad
\Delta_-(l)
=
-\frac2z+O(z^{-2}).
\]

The conditional expected increment of the public
log-likelihood ratio under state \(+1\) is
\[
\begin{aligned}
d_+(l)
&:=
\mathbb{E}_{+1}[l_{n+1}-l_n\mid l_n=l]\\
&=
(1-q_{+1}(l))\Delta_+(l)
+
q_{+1}(l)\Delta_-(l)\\
&=
D_{\mathrm{KL}}\left(
\operatorname{Bern}(q_{+1}(l))
\,\Vert\,
\operatorname{Bern}(q_{-1}(l))
\right).
\end{aligned}
\]
Set
\[
a=q_{+1}(l),
\qquad
b=q_{-1}(l),
\qquad
h=b-a.
\]
Since \(a\sim b\sim 1/(2z)\) and
\(h\sim z^{-2}\), a Taylor expansion in \(h\) yields
\[
D_{\mathrm{KL}}\left(
\operatorname{Bern}(a)
\,\Vert\,
\operatorname{Bern}(b)
\right)
=
\frac{h^2}{2b(1-b)}
+
O\left(\frac{|h|^3}{b^2}\right).
\]
Therefore,
\[
d_+(l)
\sim
\frac1{z^3}
=
\frac8{\sigma^6l^3}.
\]
Similarly,
\[
\begin{aligned}
d_-(l)
&:=
\mathbb{E}_{-1}[l_{n+1}-l_n\mid l_n=l]\\
&=
-
D_{\mathrm{KL}}\left(
\operatorname{Bern}(q_{-1}(l))
\,\Vert\,
\operatorname{Bern}(q_{+1}(l))
\right)\\
&\sim
-\frac8{\sigma^6l^3}.
\end{aligned}
\]
Thus the terms of order \(l^{-2}\) in the expected
log-likelihood increment cancel, and its leading order
is \(l^{-3}\).

Consider the auxiliary function
\[
V(l)=l^4.
\]
For either state and \(j\in\{2,3,4\}\), the expansions of
\(q_\theta(l)\) and \(\Delta_\pm(l)\) imply
\[
\mathbb{E}_\theta\left[
|l_{n+1}-l_n|^j
\,\middle|\,l_n=l
\right]
=
O(l^{-(j+1)})
\qquad\text{as }l\to+\infty.
\]
Expanding the fourth power exactly, we have
\[
\begin{aligned}
&\mathbb{E}_\theta[
V(l_{n+1})-V(l_n)\mid l_n=l]\\
&\quad=
4l^3d_\theta(l)
+
6l^2\mathbb{E}_\theta[
(l_{n+1}-l_n)^2\mid l_n=l]\\
&\qquad+
4l\mathbb{E}_\theta[
(l_{n+1}-l_n)^3\mid l_n=l]\\
&\qquad+
\mathbb{E}_\theta[
(l_{n+1}-l_n)^4\mid l_n=l].
\end{aligned}
\]
The last three terms tend to zero. Hence,
\[
\mathbb{E}_{+1}[
V(l_{n+1})-V(l_n)\mid l_n=l]
\longrightarrow
\frac{32}{\sigma^6},
\]
whereas
\[
\mathbb{E}_{-1}[
V(l_{n+1})-V(l_n)\mid l_n=l]
\longrightarrow
-\frac{32}{\sigma^6}.
\]

Under the prior mixture, the posterior probability of state
\(+1\) at belief \(l\) is
\[
p(l)=\frac{e^l}{1+e^l}.
\]
Define
\[
g(l)
=
\mathbb{E}[
V(l_{n+1})-V(l_n)\mid l_n=l].
\]
By conditioning on the state,
\[
\begin{aligned}
g(l)
={}&
p(l)\,
\mathbb{E}_{+1}[
V(l_{n+1})-V(l_n)\mid l_n=l]\\
&+
(1-p(l))\,
\mathbb{E}_{-1}[
V(l_{n+1})-V(l_n)\mid l_n=l].
\end{aligned}
\]
Since \(p(l)\to1\) as \(l\to+\infty\),
\[
g(l)\longrightarrow \frac{32}{\sigma^6}.
\]
The signal distributions, reporting mechanism, and prior
are symmetric, and \(V\) is even. Therefore \(g\) is even,
and
\[
g(l)\longrightarrow c_\sigma
\qquad\text{as }|l|\to\infty,
\qquad
c_\sigma:=\frac{32}{\sigma^6}.
\]
In particular, \(g\) is continuous and bounded on
\(\mathbb{R}\).

We next establish integrability of \(\tau_B\) and control
the time spent in bounded belief regions.
Let
\[
\overline{Q}(x\mid s,l)
=
\int_0^1 Q_\varepsilon(x\mid s,l)\,d\varepsilon
\]
be the reporting kernel after integrating over the
unobserved privacy parameter.
Since each \(Q_\varepsilon\) satisfies
\(\varepsilon\)-metric differential privacy and
\(\varepsilon\le1\),
\[
e^{-2}\overline{Q}(x\mid s,l)
\le
\overline{Q}(x\mid s+2,l)
\le
e^2\overline{Q}(x\mid s,l).
\]
The two Gaussian signal distributions differ by a
translation of \(2\), so
\[
e^{-2}
\le
\frac{P_{+1}(x\mid l)}{P_{-1}(x\mid l)}
\le
e^2.
\]
Thus,
\[
|\Delta_x(l)|\le2.
\]

For every finite \(l\), the averaged probability of report
\(+1\) is strictly increasing in the private signal.
Consequently,
\[
P_{+1}(+1\mid l)>P_{-1}(+1\mid l),
\]
and the conditional KL divergence satisfies
\[
d_+(l)>0.
\]
By continuity, for every \(R>0\),
\[
m_R:=\inf_{|l|\le R}d_+(l)>0.
\]

For a positive integer \(N\), stopped summation yields
\[
\mathbb{E}_{+1}[l_{\tau_B\wedge N}]
=
\mathbb{E}_{+1}\left[
\sum_{n=0}^{(\tau_B\wedge N)-1}d_+(l_n)
\right].
\]
Before stopping, \(|l_n|<B\), and the increment bound gives
\[
l_{\tau_B\wedge N}\le B+2.
\]
Hence,
\[
m_B\,\mathbb{E}_{+1}[\tau_B\wedge N]
\le B+2.
\]
Letting \(N\to\infty\) proves
\[
\mathbb{E}_{+1}[\tau_B]<\infty.
\]
By reflection symmetry, the stopping time is also
integrable under state \(-1\) and under the prior mixture.

For a fixed \(R>0\), define
\[
N_R(B)
=
\sum_{n=0}^{\tau_B-1}
\mathbf{1}_{\{|l_n|\le R\}}.
\]
The stopped-sum identity and the nonnegativity of \(d_+\)
give
\[
\begin{aligned}
m_R\,\mathbb{E}_{+1}[N_R(B)]
&\le
\mathbb{E}_{+1}\left[
\sum_{n=0}^{\tau_B-1}d_+(l_n)
\right]\\
&=
\mathbb{E}_{+1}[l_{\tau_B}]
\le B+2.
\end{aligned}
\]
The same occupation bound holds under state \(-1\) by
symmetry, and therefore
\[
\mathbb{E}[N_R(B)]=O_R(B).
\]

At the stopping time, the bounded-increment property implies
\[
B\le |l_{\tau_B}|\le B+2.
\]
Since \(V(l)=l^4\),
\[
B^4
\le
\mathbb{E}[V(l_{\tau_B})]
\le
(B+2)^4.
\]
Consequently,
\[
\mathbb{E}[V(l_{\tau_B})]\sim B^4.
\]

By the definition of \(g\),
\[
V(l_n)-V(l_0)-\sum_{j=0}^{n-1}g(l_j)
\]
is a martingale under the prior mixture.
Stopping at \(\tau_B\wedge N\) and letting \(N\to\infty\)
gives
\[
\mathbb{E}[V(l_{\tau_B})]
=
\mathbb{E}\left[
\sum_{n=0}^{\tau_B-1}g(l_n)
\right],
\]
where we used \(V(l_0)=0\).
This passage is justified by the boundedness of the stopped
belief process, the boundedness of \(g\), and the
integrability of \(\tau_B\).

Fix \(\delta\in(0,c_\sigma)\).
Since \(g(l)\to c_\sigma\) as \(|l|\to\infty\), there exists
\(R>0\) such that
\[
|g(l)-c_\sigma|\le\delta
\qquad\text{for }|l|>R.
\]
Let
\[
M_R=\sup_{|l|\le R}|g(l)-c_\sigma|<\infty.
\]
Using the occupation bound, we obtain
\[
\begin{aligned}
\left|
\mathbb{E}[V(l_{\tau_B})]
-
c_\sigma\mathbb{E}[\tau_B]
\right|
&\le
\mathbb{E}\left[
\sum_{n=0}^{\tau_B-1}|g(l_n)-c_\sigma|
\right]\\
&\le
\delta\,\mathbb{E}[\tau_B]
+
M_R\,\mathbb{E}[N_R(B)]\\
&=
\delta\,\mathbb{E}[\tau_B]+O_R(B).
\end{aligned}
\]
Since
\[
\mathbb{E}[V(l_{\tau_B})]\sim B^4,
\qquad
B=o(B^4),
\]
it follows that
\[
\frac1{c_\sigma+\delta}
\le
\liminf_{B\to\infty}
\frac{\mathbb{E}[\tau_B]}{B^4}
\le
\limsup_{B\to\infty}
\frac{\mathbb{E}[\tau_B]}{B^4}
\le
\frac1{c_\sigma-\delta}.
\]
Letting \(\delta\downarrow0\), we conclude that
\[
\mathbb{E}[\tau_B]
\sim
\frac1{c_\sigma}B^4
=
\frac{\sigma^6}{32}B^4.
\]

The law of the public-belief process under state \(-1\)
is the reflection of its law under state \(+1\).
Since the stopping rule depends on \(|l_n|\), the law of
\(\tau_B\) is identical under the two states. Therefore,
\[
\mathbb{E}[\tau_B\mid\theta]
=
\mathbb{E}[\tau_B]
\sim
\frac{\sigma^6}{32}B^4,
\qquad
\theta\in\{-1,+1\}.
\]

Under personalized privacy, every potential agent participates.
Thus the number of potential agent arrivals equals the number
of submitted reports, and \(T_B=\tau_B\). Hence,
\[
\mathbb{E}[T_B\mid\theta]
=
\mathbb{E}[\tau_B\mid\theta]
\sim
\frac{\sigma^6}{32}B^4.
\]
This completes the proof.
\endproof

\section{Proof of \texorpdfstring{\Cref{thm:finite-expectation-homogeneous}}{thmfiniteexpectationhomogeneous}: Time to First Correct Submitted Report under Homogeneous Privacy}
\label{app:proof:thm:finite-expectation-homogeneous}

\proof
We condition on \(\theta=+1\), take \(l_0=0\), and define
\[
\tau_{\mathrm{fc}}=\inf\{n\ge1:x_n=+1\},
\qquad H_m=\{x_1=\cdots=x_m=-1\}.
\]
Thus \(\{\tau_{\mathrm{fc}}>m\}=H_m\), with \(H_0\) the sure event. Let \(l_m^-\) be the public log-likelihood ratio along this prescribed all-negative history and put \(y_m=-l_m^-\). Reflection symmetry gives
\[
y_0=0,
\qquad y_{m+1}=y_m+D_+(y_m).
\]
This deterministic recursion describes a particular report history; it does not assume eventual correctness along the realized stochastic process.

The increment \(D_+(y)\) is continuous and strictly positive at every finite \(y\), so \(y_m\to\infty\). Let \(\beta=\varepsilon\sigma^2/2\). By \Cref{lem:simplify_diff_equation},
\[
D_+(y)\sim\frac{C(\varepsilon)}{\beta}e^{-\beta y},
\quad
C(\varepsilon)=\frac{\varepsilon\sigma^2}{4}
\left(e^{\varepsilon+\varepsilon^2\sigma^2/2}
-e^{-\varepsilon+\varepsilon^2\sigma^2/2}\right).
\]
Consequently,
\[
e^{\beta y_{m+1}}-e^{\beta y_m}\longrightarrow C(\varepsilon),
\qquad
e^{\beta y_m}\sim C(\varepsilon)m.
\]

The likelihood ratio of the history \(H_m\) is exactly \(e^{l_m^-}=e^{-y_m}\). Hence
\[
\mathbb P_{+1}(H_m)
=e^{-y_m}\mathbb P_{-1}(H_m)
\le e^{-y_m}
\le C_0(1+m)^{-1/\beta}
\]
for some constant \(C_0>0\) and every \(m\ge0\). The tail-sum formula now gives
\[
\mathbb E_{+1}[\tau_{\mathrm{fc}}]
=\sum_{m=0}^{\infty}\mathbb P_{+1}(H_m)
\le C_1\sum_{n=1}^{\infty}n^{-2/(\varepsilon\sigma^2)}.
\]
The series is finite when \(\varepsilon<2/\sigma^2\), proving the stated sufficient condition. The appearance of \(C(\varepsilon)\) in the all-negative-history calculation does not identify a prefactor for the actual stochastic learning path or the two-sided expected stopping time.
\endproof

\section{Proof of \texorpdfstring{\Cref{thm:finite-incorrect-actions-homogeneous}}{thmfiniteincorrectactionshomogeneous}: Total Number of Incorrect Submitted Reports under Homogeneous Privacy}
\label{app:proof:thm:finite-incorrect-actions-homogeneous}

\proof
Let
\[
W=\sum_{n=1}^{\infty}\mathbf1_{\{x_n=-1\}},
\qquad \beta=\varepsilon\sigma^2/2,
\qquad Y_n=e^{\beta l_n},
\]
and condition on \(\theta=+1\). By \Cref{lem:simplify_diff_equation},
\[
e^{\beta l}(e^{\beta D_+(l)}-1)\longrightarrow C(\varepsilon)
\qquad(l\to+\infty).
\]
As \(l\to-\infty\), this expression tends to zero because \(|D_+(l)|\le2\varepsilon\). Continuity therefore gives a finite constant \(M>0\) such that every positive report increases \(Y_n\) by at most \(M\). A negative report decreases \(Y_n\), since \(D_-(l)<0\). Thus, along every report history,
\[
Y_n\le Y_0+Mn.
\]

Moreover, \(q_{+1}(l)\) is continuous and strictly positive, with
\[
q_{+1}(l)\to1\quad(l\to-\infty),
\qquad
q_{+1}(l)e^{\beta l}\to\tfrac12e^{-\varepsilon+\varepsilon^2\sigma^2/2}
\quad(l\to+\infty).
\]
It follows that \(q_{+1}(l)(1+e^{\beta l})\) has a strictly positive infimum \(c\). Hence
\[
\mathbb P_{+1}(x_{n+1}=-1\mid\mathcal F_n)
=q_{+1}(l_n)
\ge\frac{c}{1+Y_n}
\ge\frac{c}{1+Y_0+Mn}.
\]
The sum of these conditional probabilities diverges along every history. By the conditional Borel--Cantelli lemma, negative reports occur infinitely often almost surely. Therefore, for every fixed \(\varepsilon\in(0,\infty)\),
\[
\mathbb P_{+1}(W=\infty)=1,
\qquad \mathbb E_{+1}[W]=\infty.
\]
This is compatible with \(\mathbb P_{+1}(x_n=-1)\to0\): a vanishing error probability need not be summable.
\endproof

\section{Learning Efficiency under Personalized Privacy}
\label{app:personalized-classical-efficiency}

We now extend the two classical efficiency results to personalized smooth randomized response with
\[
    \varepsilon_n\sim U[0,1].
\]

Throughout this section, we continue to condition on the true state
\[
    \theta=+1.
\]

\begin{theorem}[Time to First Correct Submitted Report under Personalized Privacy]
\label{thm:finite-expectation-personalized}
Consider the Gaussian sequential learning model with agent-specific privacy parameters
\(\varepsilon_n\sim U[0,1]\). Suppose the platform applies personalized smooth randomized response. Then there exist positive constants \(C_3\) and \(\widetilde C\), independent of \(n\), such that
\[
    \mathbb{E}_{+}[\tau_{\mathrm{fc}}]
    \leq
    C_3
    \sum_{n=1}^{\infty}
    e^{-\widetilde C n^{1/3}}.
\]
The series is finite. Hence,
\[
    \mathbb{E}_{+}[\tau_{\mathrm{fc}}]<\infty.
\]
\end{theorem}

\begin{theorem}[Total Number of Incorrect Submitted Reports under Personalized Privacy]
\label{thm:finite-incorrect-actions-personalized}
Consider the Gaussian sequential learning model with agent-specific privacy parameters
\(\varepsilon_n\sim U[0,1]\). Suppose the platform applies personalized smooth randomized response. Then
\[
    \mathbb P_{+}(W=\infty)=1,
    \qquad \mathbb E_{+}[W]=\infty,
\]
where \(W=\sum_{n\ge1}\mathbf1_{\{x_n=-1\}}\).
\end{theorem}

\proof{Proof of \Cref{thm:finite-expectation-personalized}}
Let \(H_m=\{x_1=\cdots=x_m=-1\}\), and write \(-y_m\) for the public log-likelihood ratio along this history, starting from \(y_0=0\). As in the homogeneous case,
\[
y_{m+1}=y_m+D_+(y_m),
\qquad
\mathbb P_{+1}(H_m)=e^{-y_m}\mathbb P_{-1}(H_m)\le e^{-y_m}.
\]
The positive-report increment calculated in \Cref{app:proof:thm:smoothly-learning-speed-hete} satisfies
\[
D_+(y)\sim\frac4{\sigma^4y^2}.
\]
Since \(D_+>0\), \(y_m\to\infty\). Expanding the cube in the deterministic recursion gives
\[
y_{m+1}^3-y_m^3\longrightarrow\frac{12}{\sigma^4},
\qquad
y_m\sim\left(\frac{12}{\sigma^4}m\right)^{1/3}.
\]
Thus \(\mathbb P_{+1}(H_m)\le C_0e^{-\widetilde C(1+m)^{1/3}}\) for suitable positive constants. Summing the tails yields
\[
\mathbb E_{+1}[\tau_{\mathrm{fc}}]
=\sum_{m=0}^{\infty}\mathbb P_{+1}(H_m)
\le C_3\sum_{n=1}^{\infty}e^{-\widetilde C n^{1/3}}<\infty.
\]
Here the exponent \(1/3\) describes the prescribed all-negative history, not the realized stochastic learning rate.
\endproof

\proof{Proof of \Cref{thm:finite-incorrect-actions-personalized}}
By \Cref{app:proof:thm:smoothly-learning-speed-hete}, the actual report process satisfies
\[
l_n\sim\left(\frac{32}{\sigma^6}n\right)^{1/4}
\qquad\mathbb P_{+1}\text{-almost surely},
\qquad
q_{+1}(l)\sim\frac1{\sigma^2l}.
\]
Consequently,
\[
\mathbb P_{+1}(x_{n+1}=-1\mid\mathcal F_n)
\sim\frac1{\sigma^2}\left(\frac{\sigma^6}{32}\right)^{1/4}n^{-1/4}
\qquad\mathbb P_{+1}\text{-almost surely}.
\]
These conditional probabilities have a divergent sum. The conditional Borel--Cantelli lemma therefore gives
\[
\mathbb P_{+1}(W=\infty)=1,
\qquad\mathbb E_{+1}[W]=\infty.
\]
This completes the proof.
\endproof

\section{Extension to Pufferfish Privacy}
\label{app:pufferfish}

To support our findings and evaluate their robustness, we also consider an alternative and widely studied privacy definition: \emph{Pufferfish privacy}~\citep{kifer2014pufferfish,chatzikokolakis2013broadening}. Unlike standard differential privacy, which relies on a fixed notion of neighboring datasets, Pufferfish privacy provides a flexible framework that allows the analyst to explicitly define which pairs of private signals should be made indistinguishable. This is particularly suitable for settings with unbounded or continuous input domains, as in our Gaussian signal model. In our context, the company specifies a local reporting protocol and protects unbounded private signals by defining neighboring signals as those within a bounded distance. Specifically, we consider signal pairs whose distance is no greater than a predefined constant \(a>0\).

We use a local instantiation of the Pufferfish framework, conditional on a fixed public history \(h_{n-1}\). The protected secrets are the possible values of the current agent's private signal, and the protected pairs are values separated by at most \(a\). To provide protection under any prior over this signal, we require the privacy bound to hold for every protected pair and every possible report. The bound compares the probabilities of a report given specified signal values, so the same formulation applies to both discrete and continuous priors.

\begin{definition}[\(\varepsilon\)-Pufferfish privacy]\label{def:dp-cont-relaxed}
A local reporting mechanism \( \mathcal{M} \), defined over a continuous domain \( \mathbb{R} \), satisfies \(\varepsilon\)-Pufferfish privacy if, for every fixed public history \(h_{n-1}\), for all reports \(x_n\in\{-1,+1\}\), and for all neighboring signals \(s_n,s_n'\in\mathbb{R}\) such that \(\|s_n-s_n'\|_1\leq a\), the following holds:
\begin{align}\label{eq:pufferfish}
\mathbb{P}(\mathcal{M}(s_n;h_{n-1})=x_n)
\leq
\exp(\varepsilon)\,
\mathbb{P}(\mathcal{M}(s_n';h_{n-1})=x_n).
\end{align}
\end{definition}

Under this privacy notion, the company can use a staircase version of randomized response. The mechanism adapts the flipping probability to the distance between the private signal and the decision threshold, while enforcing the same indistinguishability requirement for all signal pairs within distance \(a\).

\begin{definition}[Staircase Randomized Response]\label{def:staircase-RR}
A local reporting mechanism \( \mathcal{M}(s_n;h_{n-1}) \) is called a staircase randomized response protocol with flip probability \(u_n\) if, given an initial recommendation \(a_n\), the released report \(x_n\) is determined as follows:
\begin{equation*}
    \mathbb{P}_{\mathcal{M}}(x_n=+1\mid a_n=-1)
    =
    \mathbb{P}_{\mathcal{M}}(x_n=-1\mid a_n=+1)
    =
    u_n(s_n),
\end{equation*}
\begin{equation*}
    \mathbb{P}_{\mathcal{M}}(x_n=+1\mid a_n=+1)
    =
    \mathbb{P}_{\mathcal{M}}(x_n=-1\mid a_n=-1)
    =
    1-u_n(s_n),
\end{equation*}
where the flip probability is defined as
\begin{equation*}
u_n(s_n)
=
\frac{e^{-k\varepsilon}}{1+e^{\varepsilon}},
\quad
\text{if }
s_n\in
(t(l_n)+ka,\ t(l_n)+(k+1)a)
\cup
(t(l_n)-(k+1)a,\ t(l_n)-ka),
\quad
k\in\mathbb{N}.
\end{equation*}
Here, \(t(l_n)\) is the decision threshold induced by the public belief \(l_n\). If \(s_n>t(l_n)\), the initial recommendation is \(a_n=+1\); if \(s_n<t(l_n)\), the initial recommendation is \(a_n=-1\). The mechanism is symmetric around \(t(l_n)\), and the flip probability decays by a factor \(e^{-\varepsilon}\) each time the signal moves one additional distance interval of length \(a\) away from the threshold.
\end{definition}

\begin{theorem}[Pufferfish Feasibility of Staircase Randomized Response]
\label{thm:smoothly-random-response-C-puffer}
In the sequential learning model with Gaussian signals, the staircase randomized response protocol satisfies \(\varepsilon\)-Pufferfish privacy.
\end{theorem}

\begin{proof}
Given public belief \(l_n\), there exists a decision threshold \(t(l_n)\) such that the initial recommendation is \(a_n=+1\) if \(s_n>t(l_n)\), and \(a_n=-1\) otherwise.

First consider two neighboring signals in the interval around the threshold,
\[
s_n,s_n'\in(t(l_n)-a,\ t(l_n)+a).
\]
In this region, the flipping probability is
\[
u_n(s_n)=\frac{1}{1+e^\varepsilon}.
\]
Therefore, the ratio between the output probabilities induced by two neighboring signals is bounded by \(e^\varepsilon\).

More generally, suppose two neighboring signals lie in adjacent staircase intervals at distance level \(k\). By construction, the flip probability changes by at most a multiplicative factor \(e^\varepsilon\) between neighboring intervals. Thus, for every pair \(s_n,s_n'\) satisfying \(\|s_n-s_n'\|_1\leq a\) and every report \(x_n\in\{-1,+1\}\), we have
\[
\frac{
\mathbb{P}(\mathcal{M}(s_n;h_{n-1})=x_n)
}{
\mathbb{P}(\mathcal{M}(s_n';h_{n-1})=x_n)
}
\leq
e^\varepsilon.
\]
Hence, the staircase randomized response protocol satisfies \(\varepsilon\)-Pufferfish privacy.
\end{proof}

Compared to the smooth randomized response defined under metric differential privacy, the staircase randomized response exhibits the same asymptotic behavior in terms of flip probability: both decay exponentially as the signal deviates from the decision threshold. For the staircase randomized response, the flip probability \(u_n(s_n)\) satisfies the bounds
\begin{equation*}
  \frac{1}{1+e^{\varepsilon}}
  e^{-\frac{\varepsilon}{a}|s_n-t(l_n)|}
  \leq
  u_n(s_n)
  \leq
  \frac{e^{\varepsilon}}{1+e^{\varepsilon}}
  e^{-\frac{\varepsilon}{a}|s_n-t(l_n)|}.
\end{equation*}
As a result, the belief dynamics under staircase randomized response have exponential tail scale \(\varepsilon/a\). To obtain the learning rate, however, the expected log-likelihood increment must average over both possible reports.

\begin{theorem}[Learning Rate under Staircase Randomized Response]
\label{thm:smoothly-learning-speed-relax}
Consider the sequential learning model with Gaussian signals and staircase randomized response under \(\varepsilon\)-Pufferfish privacy with neighborhood radius \(a>0\). Then, conditional on \(\theta=+1\), the public log-likelihood ratio satisfies
\[
l_n\sim\frac{2a}{\varepsilon\sigma^2}\log n
\qquad\text{a.s.}
\]
Equivalently, the convergence rate is
\[
f(n)=\frac{2a}{\varepsilon\sigma^2}\log n.
\]
\end{theorem}

\proof
Set \(r=\varepsilon/a\), \(\beta_a=\varepsilon\sigma^2/(2a)\), and \(z=\sigma^2l/2\). Write \(\{t\}=t-\lfloor t\rfloor\) for the fractional part. The negative-report probabilities satisfy
\[
q_\theta(l)=e^{-rz}v_\theta(z)+O(e^{-cz^2}),
\qquad
v_\theta(z)=\frac1{1+e^\varepsilon}
\mathbb E_\theta\!\left[e^{-rS}e^{\varepsilon\{(z+S)/a\}}\right].
\]
Indeed, on \(S>-z\), the negative-report probability is
\((1+e^\varepsilon)^{-1}e^{-\varepsilon\lfloor(z+S)/a\rfloor}\); extending its Gaussian integral to the whole real line introduces only the displayed Gaussian-tail remainder. The functions \(v_\theta\) are positive, continuous, and \(a\)-periodic. Gaussian convolution of the nonincreasing staircase kernel, coupled with \(S_{+1}=S_{-1}+2\), gives \(v_{-1}(z)>v_{+1}(z)\) for every \(z\). Periodicity makes this separation uniform over a period.

With \(R(z)=v_{-1}(z)/v_{+1}(z)>1\), the complete conditional mean is
\[
\begin{aligned}
h_+(l)
&=(1-q_{+1}(l))D_+(l)+q_{+1}(l)D_-(l)\\
&=e^{-\beta_a l}v_{+1}(z)
\bigl(R(z)-1-\log R(z)\bigr)+o(e^{-\beta_a l}).
\end{aligned}
\]
The coefficient is continuous, periodic, and strictly positive. Hence \(h_+(l)=\Theta(e^{-\beta_a l})\).

At every finite belief, the two report distributions are strictly separated. The likelihood-ratio martingale argument of \Cref{app:proof:thm:smoothly-asymptotic-learning-C} proves \(l_n\to\infty\) almost surely under \(\theta=+1\). Chaining the privacy bound over a signal displacement of \(2\) gives a finite uniform increment bound, for example \(|D_x(l)|\le\varepsilon(\lceil2/a\rceil+1)\). Applying \Cref{lem:llr-accumulated-drift} now proves the stated logarithmic rate.

The report-probability expansion also gives \(q_{+1}(l)\to0\) as \(l\to+\infty\). Together with public-belief convergence and symmetry between the two states, this implies
\[
\mathbb P(x_{n+1}=\theta\mid l_n,\theta)
\longrightarrow1
\qquad\text{almost surely under either state}.
\]
Since these conditional probabilities are bounded by one, taking expectations and averaging over the prior yields
\[
\lim_{n\to\infty}\mathbb P(x_n=\theta)=1.
\]
Thus, the staircase protocol also achieves asymptotic learning in the sense of \cref{def:AL}.

\endproof

Theorem~\ref{thm:smoothly-learning-speed-relax} and its proof show that the staircase protocol preserves asymptotic learning and yields a logarithmic public-belief growth rate, analogous to the smooth randomized response benchmark. The effective decay rate is \(\varepsilon/a\): a larger neighborhood radius \(a\) weakens the effective decay rate and changes the leading constant in the learning rate.

Under staircase randomized response, reaching either boundary of a large belief threshold requires an expected stopping time of exponential order. We establish this expectation directly using the stopped potential, since inversion of an almost-sure learning rate alone does not justify a precise expected stopping-time asymptotic.

\begin{theorem}[Stopping Time under Staircase Randomized Response]
\label{thm:pufferfish-stopping-time}
Under the same conditions as Theorem~\ref{thm:smoothly-learning-speed-relax}, consider the submitted-report belief-threshold stopping rule
\[
\tau_B=\inf\{n\ge1:|l_n|\ge B\}.
\]
Conditional on \(\theta=+1\), as \(B\to\infty\),
\[
\log\tau_B\sim\frac{\varepsilon\sigma^2}{2a}B
\quad\text{a.s.},
\qquad
\mathbb E[\tau_B\mid\theta=+1]
=\Theta\!\left(\exp\left(\frac{\varepsilon\sigma^2}{2a}B\right)\right).
\]
Consequently, if each potential agent participates independently with probability \(\eta(\varepsilon)>0\), independently of reports, the expected calendar stopping time satisfies
\[
\mathbb E[T_B\mid\theta=+1]
=\frac{\mathbb E[\tau_B\mid\theta=+1]}{\eta(\varepsilon)}
=\frac1{\eta(\varepsilon)}\Theta\!\left(\exp\left(\frac{\varepsilon\sigma^2}{2a}B\right)\right).
\]
\end{theorem}

\proof
The almost-sure logarithmic relation follows from the learning-rate theorem: on a path with \(l_n\to\infty\), the negative excursion is bounded, so sufficiently high two-sided exits occur at the positive boundary. Sandwiching \(l_n\) between \((1\pm\delta)\beta_a^{-1}\log n\) gives \(\log\tau_B/B\to\beta_a\).

For the expectation, use
\[
V(l)=\frac{\cosh((\beta_a+1/2)l)}{\cosh(l/2)}.
\]
The preceding report-probability expansion gives
\[
\mathbb E_{+1}[\Delta e^{\beta_a l_n}\mid l_n=l]
=v_{+1}(z)\left[\beta_a(R(z)-1)+R(z)^{-\beta_a}-1\right]+o(1).
\]
The bracket is strictly positive for \(R(z)>1\), since it vanishes at \(R=1\) and has derivative \(\beta_a(1-R^{-\beta_a-1})>0\). Its periodic coefficient is therefore bounded above and below by positive constants. As in \Cref{app:proof:thm:smoothly-stopping-time}, the prior-mixture drift of \(V\) has these same positive bounds in both tails. Bounded increments give \(V(l_{\tau_B})=\Theta(e^{\beta_a B})\) uniformly. The second assertion of \Cref{lem:two-sided-stopped-potential} proves the expected stopping-time order. Independent geometric calendar gaps give the stated participation factor.
\endproof

The theorem shows that, under staircase randomized response, reaching a high threshold \(B\) requires a stopping time that grows exponentially in \(B\). The effective exponent is \(\varepsilon\sigma^2/(2a)\), reflecting the fact that the staircase mechanism behaves asymptotically like an exponential flip function with decay rate \(\varepsilon/a\). The periodic coefficients do not supply a single constant \(C_{\varepsilon,a}\) for a precise exponential prefactor.

\section{Polynomial-Decay Metric-DP Reporting Mechanisms}
\label{app:alternative-mdp-mechanisms}

In the main text, we focus on smooth randomized response, whose flipping probability decays exponentially with the distance between the private signal and the current decision threshold. This appendix considers an alternative class of metric-DP reporting mechanisms in which the flipping probability decays polynomially. The purpose is to show that the learning rate is not tied to the exponential form of smooth randomized response; rather, it is governed by the tail behavior of the flip function.

For a given public belief \(l_n\), recall that the Gaussian decision threshold is
\[
t(l_n)=-\frac{\sigma^2 l_n}{2}.
\]
The agent's initial recommendation is \(a_n=+1\) if \(s_n>t(l_n)\), and \(a_n=-1\) otherwise. Consider the following polynomial-decay flip function:
\[
\psi_{\alpha}(d)
=
\frac{1}{2}(1+d)^{-\alpha},
\qquad d\geq 0,
\]
where \(\alpha>0\). The released report is generated by flipping the initial recommendation with probability
\[
\mathbb{P}(x_n\neq a_n\mid s_n,l_n)
=
\psi_{\alpha}(|s_n-t(l_n)|)
=
\frac{1}{2}\left(1+|s_n-t(l_n)|\right)^{-\alpha}.
\]
We refer to this mechanism as polynomial randomized response.

\begin{proposition}[Metric-DP Feasibility of Polynomial Randomized Response]
\label{prop:polynomial-mdp-feasibility}
The polynomial randomized response mechanism with exponent \(\alpha>0\) satisfies \((\alpha,d_{\mathbb R})\)-metric differential privacy. Consequently, for any target privacy level \(\varepsilon\), it is feasible under \((\varepsilon,d_{\mathbb R})\)-metric differential privacy whenever \(\alpha\leq \varepsilon\).
\end{proposition}

\proof
Fix a public belief \(l\) and let \(t=t(l)\). Define
\[
q(s;l)=\mathbb{P}(x=+1\mid s,l).
\]
Under the polynomial randomized response mechanism,
\[
q(s;l)
=
\begin{cases}
\frac{1}{2}(1+t-s)^{-\alpha}, & s<t,\\[3pt]
1-\frac{1}{2}(1+s-t)^{-\alpha}, & s>t.
\end{cases}
\]
At \(s=t\), both expressions equal \(1/2\), so \(q(s;l)\) is continuous. It is enough to show that both \(\log q(s;l)\) and \(\log(1-q(s;l))\) are \(\alpha\)-Lipschitz in \(s\). For \(s<t\),
\[
\left|\frac{d}{ds}\log q(s;l)\right|
=
\frac{\alpha}{1+t-s}
\leq \alpha.
\]
For \(s>t\),
\[
q(s;l)=1-\frac{1}{2}(1+s-t)^{-\alpha},
\]
and therefore
\[
\left|\frac{d}{ds}\log q(s;l)\right|
=
\frac{\frac{\alpha}{2}(1+s-t)^{-\alpha-1}}
{1-\frac{1}{2}(1+s-t)^{-\alpha}}
\leq \alpha.
\]
The same argument applies to \(\log(1-q(s;l))\) by symmetry. Hence, for any \(s,s'\in\mathbb R\),
\[
\left|
\log q(s;l)-\log q(s';l)
\right|
\leq
\alpha |s-s'|,
\]
and
\[
\left|
\log(1-q(s;l))-\log(1-q(s';l))
\right|
\leq
\alpha |s-s'|.
\]
Thus, for both outputs \(x\in\{-1,+1\}\),
\[
\mathbb{P}(x\mid s,l)
\leq
e^{\alpha |s-s'|}
\mathbb{P}(x\mid s',l).
\]
This proves \((\alpha,d_{\mathbb R})\)-metric differential privacy.
\endproof

We next characterize the learning rate induced by the polynomial-decay mechanism. The key difference from smooth randomized response is the tail behavior of the flip probability. Under smooth randomized response, the flip probability decays exponentially in the distance from the decision threshold, leading to logarithmic growth of the public log-likelihood ratio. Under polynomial randomized response, the flip probability decays more slowly, and the public belief grows at a polynomial rate.

\begin{theorem}[Learning Rate under Polynomial Randomized Response]
\label{thm:polynomial-learning-rate}
Consider the Gaussian sequential learning model with binary states and private signals \(s_n\sim\mathcal N(\theta,\sigma^2)\). Suppose the platform uses the polynomial randomized response mechanism with exponent \(\alpha>0\). Conditional on \(\theta=+1\), the public log-likelihood ratio satisfies
\[
l_n=\Theta\left(n^{\frac1{\alpha+3}}\right)
\qquad\text{a.s.}
\]
Moreover, the belief-threshold stopping time
\[
\tau_B=\inf\{n\ge1:|l_n|\ge B\}
\]
satisfies
\[
\mathbb E[\tau_B\mid\theta=+1]
=\Theta\left(B^{\alpha+3}\right),
\qquad B\to\infty.
\]
\end{theorem}

\proof
We prove the result by characterizing the full conditional drift of the public log-likelihood ratio. Set
\[
z=\frac{\sigma^2l}{2},\qquad w=1+z,
\qquad q_\theta(l)=\mathbb P(x_{n+1}=-1\mid l_n=l,\theta).
\]
For large positive \(l\), the threshold \(-z\) lies in the left tail of both Gaussian signal distributions. On \(|s|\le z/2\), the negative-report probability equals \(\frac12(w+s)^{-\alpha}\). Taylor expansion on this region and Gaussian tail bounds on its complement give
\[
q_\theta(l)
=\frac12w^{-\alpha}
\left[1-\frac{\alpha\theta}{w}
+\frac{\alpha(\alpha+1)(\sigma^2+1)}{2w^2}
+O(w^{-3})\right].
\]
The truncation ensures that the expanded power has a positive base. In particular,
\[
q_{\pm1}(l)\sim\frac12w^{-\alpha},
\qquad
\delta(l):=q_{-1}(l)-q_{+1}(l)
=\alpha w^{-\alpha-1}+O(w^{-\alpha-3}).
\]

The difference \(\delta(l)\) alone does not give the expected log-likelihood increment. The two report updates are
\[
D_+(l)=\log\frac{1-q_{+1}(l)}{1-q_{-1}(l)}\sim\delta(l),
\qquad
D_-(l)=\log\frac{q_{+1}(l)}{q_{-1}(l)}\sim-\frac{2\alpha}{w}.
\]
Thus the positive- and negative-report contributions cancel at order \(w^{-\alpha-1}\). Expanding their sum as a Bernoulli KL divergence, using \(\delta/q_{-1}\to0\), yields
\[
\begin{aligned}
h_+(l)
&=(1-q_{+1}(l))D_+(l)+q_{+1}(l)D_-(l)\\
&=\frac{\delta(l)^2}{2q_{-1}(l)(1-q_{-1}(l))}
+O\!\left(\frac{\delta(l)^3}{q_{-1}(l)^2}\right)\\
&\sim\alpha^2w^{-\alpha-2}
\sim H_{\alpha,\sigma}l^{-\alpha-2},
\end{aligned}
\]
where
\[
H_{\alpha,\sigma}=\alpha^2\left(\frac2{\sigma^2}\right)^{\alpha+2}.
\]
Under \(\theta=-1\), the same calculation with the reversed KL divergence gives \(h_-(l)\sim-H_{\alpha,\sigma}l^{-\alpha-2}\) as \(l\to+\infty\).

For each finite \(l\), the reporting kernel is strictly monotone in the signal, so the state-dependent report distributions are strictly separated. The likelihood-ratio martingale argument in \Cref{app:proof:thm:smoothly-asymptotic-learning-C} therefore gives \(l_n\to\infty\) almost surely under \(\theta=+1\). By \Cref{prop:polynomial-mdp-feasibility} and the Gaussian mean displacement of \(2\), the increments satisfy \(|l_{n+1}-l_n|\le2\alpha\). Applying \Cref{lem:llr-accumulated-drift} with \(p=\alpha+2\) proves the more precise rate
\[
l_n\sim\left((\alpha+3)H_{\alpha,\sigma}n\right)^{1/(\alpha+3)}
\qquad\mathbb P_{+1}\text{-almost surely}.
\]
Equivalently, the associated asymptotic differential equation is
\[
\frac{df}{dt}=H_{\alpha,\sigma}f^{-\alpha-2},
\qquad
f(t)^{\alpha+3}=(\alpha+3)H_{\alpha,\sigma}t+O(1).
\]
The martingale argument above justifies this rate for the stochastic process.

We next establish the expectation of the two-sided stopping time directly. Put \(r=\alpha+3\) and \(V(l)=|l|^r\). The report-probability expansions imply
\[
\mathbb E_\theta[(l_{n+1}-l_n)^2\mid l_n=l]
=O(l^{-\alpha-2})
\qquad(l\to+\infty)
\]
for both states. Using the bounded increments, Taylor's theorem gives
\[
\begin{aligned}
\mathbb E_{+1}[\Delta V(l_n)\mid l_n=l]&\longrightarrow rH_{\alpha,\sigma},\\
\mathbb E_{-1}[\Delta V(l_n)\mid l_n=l]&\longrightarrow-rH_{\alpha,\sigma}.
\end{aligned}
\]
Indeed, the first-order term is \(rl^{r-1}h_\theta(l)\), while the expected remainder is \(O(l^{r-2}l^{-\alpha-2})=O(l^{-1})\).

Under the symmetric prior mixture, the posterior probability of \(\theta=+1\) tends to one as \(l\to+\infty\). Reflection therefore implies that the mixture drift of the even potential \(V\) tends to \(rH_{\alpha,\sigma}\) in both tails. The true-state drift is continuous and strictly positive at every finite belief, and bounded overshoots give
\[
B\le|l_{\tau_B}|\le B+2\alpha,
\qquad V(l_{\tau_B})\sim B^r
\]
uniformly. Since \(B=o(B^r)\), \Cref{lem:two-sided-stopped-potential} gives
\[
\mathbb E[\tau_B\mid\theta=+1]
\sim\frac{B^{\alpha+3}}{(\alpha+3)H_{\alpha,\sigma}}.
\]
This proves the stated order without assuming unproved concentration or uniform integrability.
\endproof

\clearpage %
\end{APPENDICES}




\end{document}